%% file: arXiv_main.tex
\documentclass[11pt,letterpaper]{article}
\usepackage{fullpage}
\usepackage[utf8]{inputenc}

\usepackage{amsmath, amssymb, amsfonts, amsthm}

\usepackage{natbib}
\usepackage{url}
\usepackage[colorlinks,linkcolor=blue,citecolor=blue,urlcolor=blue]{hyperref}

\usepackage{mathtools}
\usepackage{authblk}
\usepackage{graphicx}
\usepackage[usenames,dvipsnames]{xcolor}
\usepackage{verbatim, color}
\usepackage{soul}

\usepackage{setspace}

\input{cldaNotations}

\newcommand{\ech}{\color{black}\rm}

\begin{document}
\title{Sparse semiparametric discriminant analysis for high-dimensional zero-inflated data }
\author[1]{Hee Cheol Chung \thanks{hchung13@uncc.edu}}
\author[2]{Yang Ni \thanks{yni@stat.tamu.edu}}
\author[2]{Irina Gaynanova \thanks{Corresponding author, irinag@stat.tamu.edu}}
\affil[1]{\small Department of Mathematics and Statistics, University of North Carolina at Charlotte, Charlotte, NC 28223, USA}
\affil[2]{\small Department of Statistics, Texas A\&M University, College Station, TX 77843, USA}
\date{}
\maketitle
\begin{abstract}
Sequencing-based technologies provide an abundance of high-dimensional biological datasets with skewed and zero-inflated measurements. Classification of such data with linear discriminant analysis leads to poor performance due to the violation of the Gaussian distribution assumption. To address this limitation, we propose a new semiparametric discriminant analysis framework based on the truncated latent Gaussian copula model that accommodates both skewness and zero inflation. By applying sparsity regularization, we demonstrate that the proposed method leads to the consistent estimation of classification direction in high-dimensional settings. On simulated data, the proposed method shows superior performance compared to the existing method. We apply the method to discriminate healthy controls from patients with Crohn's disease based on microbiome data and to identify genera with the most influence on the classification rule.
\end{abstract}

\section{Introduction}
Linear discriminant analysis (LDA) is a popular classification method 
due to its simple linear classification rule and Bayes optimality under the assumption that each class population is Gaussian with equal covariance matrix but different means. However, 
the complexity of modern high-dimensional data raises many challenges in application of classical LDA. For example, microbiome and single-cell RNA sequencing data not only contain a large number of variables relative to sample size, but also are highly skewed and zero-inflated \citep{silverman2020naught}. 
In our motivating example of microbiome data from \citet{vandeputte2017quantitative}, we are interested in discriminating $n_1=29$ patients with Crohn's disease from $n_2=106$ healthy controls based on $p=101$ genera, and identifying key genera that affect the classification rule.
First, a large number of genera relative to sample size makes standard LDA suffer in terms of interpretation and accuracy, as it performs no variable selection and the sample covariance matrix becomes close to singular \citep{bickel2004some}. 
Second, the advantages of LDA under the Gaussian assumption -- simple linear classification rule and Bayes optimality -- can clearly not be transferred to skewed and zero-inflated data. Since LDA is very sensitive to outliers \citep{hastie2009elements}, large measurement values due to skewness and large proportions of zeros can significantly bias LDA parameter estimates leading to inaccurate classification. 

There is a rich body of work extending classical LDA to high-dimensional settings. A common approach is to add regularization to the classification direction vector, e.g., the $\ell_{2}$ regularization \citep{guo2007regularized} 
or the $\ell_{1}$ regularization \citep{witten2011penalized,  Mai:2012bf, Gaynanova:2016wk}.
An alternative approach is to consider data piling phenomenon in high dimensions, that is that observations from each class can be projected to a single point.
 \citet{ahn2010maximal} estimate classification direction by maximizing the distance between data piling sites, and \citet{lee2013hdlss} regularize the degree of data piling. While these approaches improve accuracy of LDA in the presence of larger number of variables, they still rely on assumption of Gaussian distribution, leading to poor performance in the presence of skewness and zero-inflation. 
 Several works consider relaxations of the Gaussian assumption. \citet{ahn2021trace} consider an alternative trace ratio formulation of Fisher's LDA, which is more robust to violations of Gaussianity than standard LDA. \citet{clemmensen2011sparse} model each class as a mixture of Gaussian distributions with subclass-specific means and common covariance matrices.  
 \citet{witten2011classification} and \citet{dong2016nblda} consider extensions of LDA model for RNA-seq data based on Poisson distribution and negative binomial distribution, respectively.  
 \citet{hernandez2005dimension} consider fully nonparametric approach with kernel LDA. \citet{lapanowski2019sparse} consider optimal scoring formalation of kernel LDA with sparsity regularization. These methods, although useful for high-dimensional data, are still not fully tailored for skewness and zero-inflation. 


In this work, we aim to simultaneously address the issues of high-dimensionality, skewness and zero-inflation in the LDA context. We propose to achieve this by formulating a new semiparametric LDA model via the truncated latent Gaussian copula designed for zero-inflated data \citep{yoon2020sparse}. Latent Gaussian copula models provide an elegant framework for the analysis of non-Gaussian data of (possibly) mixed types, such as skewed continuous \citep{Liu:2009wz}, binary \citep{Fan:2016um}, ordinal \citep{Quan:2018tv, Feng:2019wh}, and zero-inflated \citep{yoon2020sparse}. These models capture dependencies among variables via the latent correlation matrix, which can be consistently estimated by inverting a bridge function that connects latent correlations to a rank-based measure of association (Kendall’s $\tau$). Subsequently, these models have been used for graphical model estimation \citep{Fan:2016um,Feng:2019wh, Yoon2019microbial} and canonical correlation analysis \citep{yoon2020sparse} with non-Gaussian data.  However, the use of latent Gaussian copulas in classification context has been limited. The existing approaches \citep{lin2003discriminant,Han:2014je} are restricted to continuous data type, treating zeros as absolute. Furthermore, since standard LDA model assumes class-specific means, but the means are not identifiable under the Gaussian copula, \citet{Han:2013ju} and \citet{lin2003discriminant} enforce identifiability by enforcing marginal transformations to preserve means and variances between the observed and the latent Gaussian variables. Consequently, these approaches still rely on observation-level moment estimates, which are sensitive to zero inflation and outliers.

There are two major difficulties in adopting the truncated Gaussican copula model of \citet{yoon2020sparse} for LDA setting. The first difficulty is the identifiability of mean and variance parameters as described above, since the copula model is invariant under shifting and scaling of the bijective marginal transformations. Secondly, the aforementioned unsupervised problems -- graphical model estimation and canonical correlation analysis -- only require a consistent estimator of latent correlation matrix. In contrast, LDA encompasses both estimations of classification direction (based on latent Gaussian level) and prediction of class labels on new data (observed at non-Gaussian level), thus requiring mapping of observed data to the latent level. This is a particularly challenging task for zero-inflated data, as the mapping from observed zeros to underlying latent Gaussian variables is not one to one.




To address the identifiability difficulty, we consider a joint mixed binary-truncated copula model, where the class label is treated as dichotomized latent Gaussian variable, and zero-inflated measurements follow the truncated Gaussian copula. The joint framework encapsulates all relationships between the class label and covariates via joint latent correlation matrix. Unlike \citet{lin2003discriminant,Han:2014je}, our approach does not require marginal transformations to be mean and variance preserving, and consequently we do not rely on observation-level moment estimates. The optimal classification direction depends only on the joint latent correlation matrix. Furthermore, by adapting $\ell_1$-regularization, we prove that this classification direction can be estimated consistently in high-dimensional settings with the same rate as in sparse linear regression \citep{bickel2009simultaneous}. This is a non-trivial result as direct application of element-wise consistency of rank-based estimator of joint latent correlation matrix \citep{yoon2020sparse} leads to suboptimal rate. 
To our knowledge, the closest result is obtained by \citet{barber2018rocket} for continuous Gaussian copula. However, their proof takes advantage of the closed form of the inverse bridge function. Due to significant complexity of bridge function in truncated case \citep{yoon2020sparse}, its inverse is not available in closed form, presenting significant new challenges for theoretical analyses. 

To address the difficulty of predicting classes based on observed non-Gaussian data, we derive the posterior class probabilities conditional on observed measurements (zeros and non-zeros). We demonstrate that these probabilities can be computed using the truncated Gaussian distribution. Furthermore, we derive a Taylor approximation of the posterior probabilities that leads to a simple classification rule where the latent measurements corresponding to zeros are substituted with their conditional expectations. 

In summary, our main contributions are: (a) a new LDA framework for skewed and zero-inflated data based on semiparametric mixed binary-truncated latent Gaussian copula model; (b)~theoretical guarantees that the proposed framework leads to consistent estimation of classification direction in high-dimensional settings; (c) a principled approach for assigning class labels to observed zero-inflated data with non-unique mapping to latent space; (d) superior classification accuracy compared to existing methods on simulated data and microbiome data of \cite{vandeputte2017quantitative}.

The rest of the paper is organized as follows. In Section \ref{sec:method}, we introduce the proposed classification model and describe the estimation procedures. In Section \ref{sec:theory}, we show the consistency of the estimated classification direction. In Section \ref{sec:simulation}, we evaluate the classification accuracy of the proposed method with simulated datasets. In Section \ref{sec:application}, we apply the proposed method to the quantitative microbiome dataset \citep{vandeputte2017quantitative} to discriminate healthy controls from patients with Crohn's disease. In Section \ref{sec:disc}, we conclude the paper with a brief discussion.

\section{Methodology}\label{sec:method}
In Section \ref{subsec:method_notations}, we define notations that are used throughout the paper. In Section \ref{subsec:method_model}, we propose a semiparametric linear classification model based on latent Gaussian copula after reviewing the probit regression model and its latent variable formulation. In Section \ref{subsec:classrule}, we derive the posterior probability under the proposed model and provide the Bayes classification rule.

\subsection{Notation}\label{subsec:method_notations}
For a vector $\ba \in \R^{p}$, we denote the $\ell_{q}$-norm, $q\in [0,\infty)$, by $\|\ba\|_{q}=(\sum_{j=1}^{p}|a_{j}|^{q})^{1/q}$ and the $\ell_{\infty}$-norm by $\|\ba\|_{\infty} = \max_{1\leq j \leq p} |a_{j}|$. For two vectors with the same size, $\ba,\bb \in \R^{p}$, we write $\ba<\bb$ to denote element-wise inequalities such that $a_{j}<b_{j}$, $\forall j=1,\ldots,p$. For a matrix $\bA \in \R^{n \times p}$, $\|\bA\|_{\infty}=\max_{j,k}|a_{jk}|$ denotes its $\ell_{\infty}$-norm, and for a square matrix $\bC \in \R^{p \times p}$, $|\bC|$ denotes its determiant. For two matrices with the same size, $\bA,\bB \in \R^{n \times p}$, $\bA \circ \bB $ denotes the Hadamard product defined as $\bA \circ \bB=[a_{jk}b_{jk}] \in \R^{n \times p}$. For two functions $f$ and $g$, we denote their composite function by $f \circ g = f(g(x))$. We let $1(\cdot)$ denote the indicator function taking the value 1 when its argument is true and 0 otherwise. For a sequence of random variables, $X_{1},\ldots,X_{n},\ldots$, we write $X_{n}=O_{p}(a_{n})$ if, for any $\varepsilon>0$, there exists $M,N>0$ such that $\PP(|x_{n}/a_{n}|>M) < \varepsilon$ for all $n > N$. We let $\Phi_{d}(a_{1},\ldots,a_{d};\bSigma)$ and $\Phi(\cdot)$ respectively denote the $d$-dimensional standard Gaussian distribution function with the correlation matrix $\bSigma$ evaluated at $(a_{1},\ldots,a_{d})^{\top}$ and  the univariate standard Gaussian distribution function, respectively. We use $C,C_{1},\ldots$ to denote generic constants that do not depend on the sample size $n$, dimension $p$, and model parameters.

\subsection{Model}\label{subsec:method_model}
Let $Y\in\{0, 1\}$ be a random variable corresponding to class label and $\bx = (X_{1},\ldots,X_{p})^{\top}\in \R^{p}$ be a random vector of covariates. To accommodate possibly skewed and zero-inflated $\bx$, we propose to model $\bx$ using truncated latent Gaussian copula of \citet{yoon2020sparse}. We first review standard Gaussian copula model, which is also referred to as nonparanormal (NPN) model \citep{Liu:2009wz}.

\begin{definition}[Gaussian copula model]\label{def:NPN}
A random vector $\bx \in \R^{p}$ satisfies the Gaussian copula model if there exist strictly increasing transformations $\{f_{j}\}_{j=1}^{p}$ such that $(Z_{1},\ldots,Z_{p} )^{\top}  = (f_{1}(X_{1}),\ldots,f_{p}(X_{p}))^{\top} \sim \N_{p}(\zeros,\bSigma)$, where $\bSigma$ is a correlation matrix. We write $\bx \sim \NPN(\zeros,\bSigma,\bdf)$.
\end{definition}
While Gaussian copula model can accommodate skewness through transformation functions $\{f_{j}\}_{j=1}^{p}$, it does not allow zero-inflatied variables. The model of \citet{yoon2020sparse} allows for both zero-inflation and skewness through the following extra truncation step.
\begin{definition}[Truncated Gaussian copula model]\label{def:trun}
A random vector $\bx \in \R^{p}$ satisfies the truncated latent Gaussian copula model if there exist a random vector $\bx^*\sim \NPN(\zeros,\bSigma,\bdf)$ and constants $D_j$, $j=1, \dots, p$, such that $X_j = 1(X_{j}^{*}>D_{j})X_{j}^{*}$.
\end{definition}

Combining the truncated Gaussian copula model for $\bx$ with the latent Gaussian copula model for binary variable \citep{Fan:2016um} leads to the joint model for the class assignment $Y$ and covariates $\bx$.

\begin{definition}[Latent Gaussian copula model for mixed binary-truncated data]\label{def:LNPN_BT}
A random vector $(Y,\bx^{\top})^{\top} \in \{0,1\}\times\R^{p}$ satisfies the latent Gaussian copula model for mixed binary-truncated data if there exist a random vector $(X_0^*, \bx^*)\sim \NPN(\zeros,\bSigma,\bdf)$ and constants $D_{j}$, $j=0, \dots, p$, such that 
\begin{equation}
\begin{split}\label{eq:model}
Y &=1(X_{0}^{*}>D_{0}),\\
X_{j} &=1(X_{j}^{*}>D_{j})X_{j}^{*}, \quad j=1,\ldots,p.
\end{split}
\end{equation}
\end{definition}
\noindent
A special case is $D_{j}=-\infty$, $j=1,\dots, p$, in which $\bx$ follows standard Gaussian copula (without truncation). Thus model~\eqref{eq:model} can also be used in binary classification settings with skewed (not necessarily zero-inflated) $\bx$.

The Bayes classification rule assigns a new observation $\bx$ to class 1 if $\PP(Y=1|\bx)>\PP(Y=0|\bx)$, and to class 0, otherwise. In the model~\eqref{eq:model}, the class label $Y$ depends on the covariates $\bx$ through underlying joint latent correlation matrix $\bSigma$. Due to the latent Gaussian layer, the conditional probability $\PP(Y=1|\bx)$ under the model~\eqref{eq:model} is closely connected to probit regression model as we demonstrate below.

\subsection{Classification rule}\label{subsec:classrule}

The probit regression model is given by $\PP(Y=1|\bx) = \Phi(\beta_{0}+\bbeta^{\top}\bx)$, where $\beta_{0} \in \R$ and $\bbeta \in \R^{p}$ are the intercept and regression coefficient, respectively. Under this model, the Bayes classifier is $\delta(\bx) = 1(\beta_0 +\bbeta^{\top}\bx > 0).$ We show that the Bayes classifier under the model~\eqref{eq:model} takes a similar form with appropriate choices of $\beta_0$ and $\bbeta$.

First, consider the special case, where $D_{j}=-\infty$, $j=1,\dots, p$; that is, $\bx$ follows standard Gaussian copula (without truncation). By the 
definition, there exists latent Gaussian vector $(Z_{y},Z_{1},\ldots,Z_{p})  \sim \N_{1+p}(\zeros,\bSigma)$ such that $Y=1(X_{0}>D_{0})= 1(Z_{y}>\Delta_{y})$, $\Delta_{y}=f_{0}(D_{0})$, and $X_{j} = f_j^{-1}(Z_j)$. Let $\bz = (Z_{1},\ldots,Z_{p})^{\top}$. Since $f_j$'s are strictly increasing, it follows that
$$
\PP(Y = 1|\bx)  = \PP(Z_{y}> \Delta_{y}|\bz).
$$
Divide the correlation matrix $\bSigma$ into blocks corresponding to the latent $Z_y$ and covariates $\bz$ so that $\cov(Z_{y},\bz)=\bSigma_{21}^{\top}$ and $\cov(\bz) = \bSigma_{22}$. By properties of the multivariate Gaussian distribution, $Z_y|\bz\sim \N({\bbeta^*}^{\top}\bz,\ v^2)$, where $\bbeta^*=\bSigma_{22}^{-1}\bSigma_{21}$ and $v^2 = 1-\bSigma_{21}^{\top}\bSigma_{22}^{-1}\bSigma_{21}$. Hence,
$$
\PP(Y = 1|\bx)  = \PP(Z_{y}> \Delta_{y}|\bz) = \PP\left(\frac{Z_{y} - {\bbeta^*}^{\top}\bz}{v}> \frac{\Delta_{y} - {\bbeta^*}^{\top}\bz}{v} \Bigg|\bz\right) = \Phi\left( \frac{ {\bbeta^*}^{\top}\bz - \Delta_{y}}{v}\right).
$$
Accordingly, since $\bz = (f_1(X_1), \dots, f_p(X_p))^{\top} = \bdf(\bx)$, we have the same linear Bayes classifier as the probit model, $\delta(\bx) = 1({\bbeta^*}^{\top}\bdf(\bx) - \Delta_{y} >0)$ with $\beta_0 = -\Delta_y$ and $\bbeta = \bbeta^*=\bSigma_{22}^{-1}\bSigma_{21}$. 





Now, consider the general truncated case of model~\eqref{eq:model}. Then $X_{j} = 1(Z_{j}>\Delta_{j})f_j^{-1}(Z_j)$ with $\Delta_{j} = f_{j}(D_{j})$, $j=1,\ldots,p$. For a new vector of covariates $\bx \in \R^p$, 
let $\bx_{t}\in \R^{p_{t}}$ and $\bx_{o}\in \R^{p_{o}}$ be the subvectors with truncated and observed realizations, respectively, where $p_{t}+p_{o}=p$. Likewise, let $\bz_{t}$ and $\bz_{o}$ be the corresponding latent Gaussian vectors, and $\bDelta_{t}$ and $\bDelta_{o}$ be the corresponding threshold vectors. Then it follows that
$$
\PP(Y=1|\bx) = \PP(Z_{y}>\Delta_{y}|\bz_{o},\bz_{t}<\bDelta_{t}).
$$
Since $Z_y|\bz\sim \N(\bbeta^{\top}\bz,\ v^2)$ as before, the posterior probability can be written as
\bfl{\nonumber
\PP(Z_{y}>\Delta_{y}|\bz_{o},\bz_{t}<\bDelta_{t})
&=
\left\{\PP(\bz_{t}<\bDelta_{t}|\bz_{o})\right\}^{-1} \int_{\Delta_{y}}^{\infty} \int_{\bz_{t}<\bDelta_{t}} p(z_{y},\bz_{t}|\bz_{o})d\bz_{t} dz_{y} \\ \nonumber
&=
\left\{\PP(\bz_{t}<\bDelta_{t}|\bz_{o})\right\}^{-1}
 \int_{\bz_{t}<\bDelta_{t}}
 \left\{
 \int_{\Delta_{y}}^{\infty}
 p(z_{y}|\bz_{t},\bz_{o}) dz_{y}
 \right\}
 p(\bz_{t}|\bz_{o}) d\bz_{t}\\ \nonumber
 &=
\left\{\PP(\bz_{t}<\bDelta_{t}|\bz_{o})\right\}^{-1}
 \int_{\bz_{t}<\bDelta_{t}}
\Phi
\left( 
\frac{{\bbeta^*_{t}}^{\top}\bz_{t} + {\bbeta^*_{o}}^{\top}\bz_{o} - \Delta_{y} }{v}
\right)
 p(\bz_{t}|\bz_{o}) d\bz_{t} \\ 
 \label{eq:clda_posterior_prob_expectationform}
&  = \E\left\{ 
    \Phi
    \left( 
    \frac{{\bbeta^*_{t}}^{\top}\bz_{t} + {\bbeta^*_{o}}^{\top}\bz_{o} - \Delta_{y} }{v}
    \right)
    \bigg|
    \bz_{o}, \bz_{t}<\bDelta_{t}
    \right\},
}
where the expectation is over the multivariate Gaussian distribution of $\bz_{t}$ given $\bz_{o}$ truncated to the region $\{\ba\in\R^{p_{t}}| \ba < \bDelta_{t}\}$. If the expectation \eqref{eq:clda_posterior_prob_expectationform} is larger than 0.5, the Bayes classification rule assigns a new observation $\bx$ to class 1, otherwise it assigns it to class 0.

Unfortunately, the expectation above has no closed form. We propose two approaches to address this difficulty as follows. First, we can sample $\bz_{t}$ from the multivariate truncated Gaussian distribution and approximate the posterior probability using the Monte-Carlo method as
\begin{align}\label{eq:postProb_MC}
    {\PP(Y=1|\bx) } = \E&\left\{ 
    \Phi 
    \left( 
    \frac{{\bbeta^*_{t}}^{\top}\bz_{t} + {\bbeta^*_{o}}^{\top}\bz_{o} - \Delta_{y} }{v}
    \right)
    \bigg|
    \bz_{o}, \bz_{t}<\bDelta_{t}
    \right\}\approx \frac{1}{S}\sum_{s=1}^{S}
    \Phi
    \left( 
    \frac{{\bbeta^*_{t}}^{\top}{\bz}_{t}^{(s)} + {\bbeta^*_{o}}^{\top}{\bz}_{o} - \Delta_{y} }{{v}}
    \right),
\end{align}
where $\{\bz_{t}^{(s)}\}_{s=1}^S$, is an independent sample from the $p_{t}$-variate truncated Gaussian. This approach, however, is computationally demanding, and makes the classification rule dependent on the scaling factor $v$ (recall that the classification rules of probit and standard Gaussian copula model do not depend on $v$). As an alternative, we consider the first-order Taylor approximation of the expectation around the mean $\bmu_{t}=\E ( \bz_{t}| {\bz}_{o}, \bz_{t}<{\bDelta}_{t})$, which leads to  (see Appendix \ref{appendix:Taylor_posterior_prob} for derivation), 
\begin{align}\label{eq:postProb_linear}
\PP(Y=1|\bx) =\E&\left\{ 
    \Phi 
    \left( 
    \frac{{\bbeta^*_{t}}^{\top}\bz_{t} + {\bbeta^*_{o}}^{\top}\bz_{o} - \Delta_{y} }{v}
    \right)
    \bigg|
    \bz_{o}, \bz_{t}<\bdelta_{t}
    \right\} 
    \approx
    \Phi
    \left( 
    \frac{{\bbeta^*_{t}}^{\top}\bmu_{t} + {\bbeta^*_{o}}^{\top}\bz_{o} - \Delta_{y} }{v}
    \right).
\end{align}
In general, the mean of multivariate truncated Gaussian distribution $\bmu_{t}$ has no closed form expression, and thus, it also needs to be estimated by Monte Carlo (MC) sampling. However, the classification rule based on~\eqref{eq:postProb_linear} is linear, $\delta(\bx) = 1({\bbeta^*_{t}}^{\top}\bmu_{t} + {\bbeta^*_o}^{\top}\bdf_o(\bx_o)  - \Delta_{y} >0)$, so its main advantage over~\eqref{eq:postProb_MC} is that it does not depend on the scaling factor $v$ .

For both~\eqref{eq:postProb_MC} and~\eqref{eq:postProb_linear}, the truncated variables enter the classification rule only through the inner-product with $\bbeta^*$. In Section~\ref{subsec:estimation_beta},
we estimate $\bbeta^*$ using sparsity regularization, and thus, in practice, we only need to generate a subvector of $\bz_t$ corresponding to non-zero elements in $\bbeta^{*}_{t}$, which makes MC sampling faster.


\subsection{Estimation of the classification direction}\label{subsec:estimation_beta}
The Bayes classification rule under model~\eqref{eq:model} depends crucially on $\bbeta^* = \bSigma_{22}^{-1}\bSigma_{21}\in \R^p$, which we refer to as classification direction.
Recall that the best linear unbiased predictor of $Z_{y}$, the latent Gaussian variable of the class label, is $\E(Z_y|\bz) = {\bbeta^*}^{\top}\bz$, and thus we can view $\bbeta^*$ as the minimizer of the mean squared error criterion:
\bfl{\label{eq:mse}
\bbeta^* = \underset{\bbeta\in \R^{p}}{\argmin} ~\E\left[ (Z_{y} - \bbeta^{\top}\bz)^{2}\right] =  \underset{\bbeta\in \R^{p}}{\argmin} ~ 
\left(
\bbeta^{\top}\bSigma_{22}\bbeta
 - 2\bbeta^{\top}\bSigma_{12}   \right).
}
In practice, $\bSigma_{12}$ and $\bSigma_{22}$ are unknown, and need to be estimated from the data. However, as $Z_{y}$ and $\bz$ are unobservable latent variables, $\bSigma_{21}$ and $\bSigma_{22}$ cannot be directly estimated using the sample correlation matrices. Instead, we propose to utilize rank-based estimators for $\bSigma_{12}$ and $\bSigma_{22}$ that take advantage of the bridge function connecting latent correlations to Kendall's $\tau$ values \citep{Fan:2016um,yoon2020sparse}. The advantage of this connection is that it enables consistent estimation of latent correlations based on ranks without requiring estimation of underlying monotone transformations $f_j$.


Concretely, a strictly increasing bridge function $G$ is defined such that $G(\sigma_{jk})=\E(\tauhat_{jk})=\tau_{jk}$, where $\sigma_{jk}$ is an element of the full correlation matrix $\bSigma$ corresponding to variables $j$ and $k$, $\tau_{jk}$ is the corresponding population Kendall's $\tau$, and $\tauhat_{jk}$ is the sample Kendall's $\tau$ 
\bfln{
\widehat{\tau}_{jk} = \frac{2}{n(n-1)}\sum_{1\leq i \leq i' \leq n} \sign(X_{ij} - X_{i'j})\sign(X_{ik} - X_{i'k}),
}
with $X_{ij}$ being the $i$-th observation of $X_j$ and $n$ being the sample size. The specific form of the bridge function $G$ depends on the type of observed variables. In our case, we are interested in binary/truncated pairs (correlations between binary class label and zero-inflated variables), and truncated/truncated pairs (correlations between zero-inflated variables).
The corresponding bridge functions $G_{BT}$ and $G_{TT}$ have the following closed forms \citep{yoon2020sparse}:
\begin{align}
G_{BT}(r;\Delta_{j},\Delta_{k})&= 2\{1-\Phi(\Delta_{j})\}\Phi(\Delta_{k}) - 2\Phi_{3}(-\Delta_{j},-\Delta_{k},0; \bSigma_{3a}(r))
         - 2\Phi_{3}(-\Delta_{j},-\Delta_{k},0; \bSigma_{3a}(r)), \notag\\
         \label{eq:bridgeForm}
G_{TT}(r;\Delta_{j},\Delta_{k})&= -2\Phi_{4}(-\Delta_{j},-\Delta_{k},0,0; \bSigma_{4a}(r)) + 2\Phi_{4}(-\Delta_{j},-\Delta_{k},0,0; \bSigma_{4b}(r)),
\end{align}
with
\bfln{
\begin{array}{cc}
\bSigma_{3a}(r) = 
\begin{pmatrix}
1   &   -r  &   1/\sqrt{2} \\
-r   &   1  &   -r/\sqrt{2} \\
1/\sqrt{2}   &   -r/\sqrt{2}  &   1 
\end{pmatrix},
& 
\bSigma_{3b}(r) = 
\begin{pmatrix}
1   &   0  &   1/\sqrt{2} \\
0   &   1  &   -r/\sqrt{2} \\
1/\sqrt{2}   &   -r/\sqrt{2}  &   1 
\end{pmatrix},
\\
\bSigma_{4a}(r) = 
\begin{pmatrix}
1   &   0  &   1/\sqrt{2} &   -r/\sqrt{2} \\
0   &   1  &   -r/\sqrt{2} &   1/\sqrt{2} \\
-r/\sqrt{2} &   1/\sqrt{2} &    1   &   -r \\
1/\sqrt{2} &   -r/\sqrt{2} &    -r   &   1 
\end{pmatrix},
&
\bSigma_{4b}(r) = 
\begin{pmatrix}
1   &   r  &   1/\sqrt{2} &    r/\sqrt{2} \\
r   &   1  &   r/\sqrt{2} &   1/\sqrt{2} \\
r/\sqrt{2} &  1/\sqrt{2} &    1   &   r \\
1/\sqrt{2} &   r/\sqrt{2} &   r   &   1 
\end{pmatrix}.
\end{array}
}

The moment equation $G(\sigma_{jk};\Delta_{j}, \Delta_{k}) = \E(\tauhat_{jk})$ and the strict monotonicity of $G$ enable estimation of the latent correlation matrix $\bSigma$ using the method of moments. To account for unknown $\Delta_{j}$,
we replace it with the moment estimator $\Deltahat_{j} = \Phi^{-1}(\pihat_{j})$ as in \citet{Fan:2016um}, where $\pihat_{j}=n_{0j}/n$ is the proportion of zeros in $X_j$ with $n_{0j} = \sum_{i=1}^{n}  1(x_{ij}=0)$,  leading to $\sigmahat_{jk} = G^{-1}(\tauhat_{jk};\Deltahat_{j},\Deltahat_{k})$. As the resulting $\bSigmah=[\sigmahat_{jk}]_{1\leq j,k \leq p}$ is not guaranteed to be positive-semidefinite \citep{Fan:2016um,yoon2020sparse}, it is projected onto the cone of positive-semidefinite matrices to obtain $\bSigma_p$ \citep{Fan:2016um}, and  \citet{yoon2020sparse} further consider $\widetilde{\bSigma} = (1-\nu)\bSigmah_{p} + \nu\bI$ with a small positive constant $\nu$ to make $\widetilde{\bSigma}$ positive definite. When $\nu=o(\sqrt{\log p/n})$, these modifications do not affect the consistency rates of resulting $\widetilde \bSigma$, see Corollary 2 in \citet{Fan:2016um} and Theorem 7 in \citet{yoon2020sparse}. In our implementation, we define $\widehat{\bSigma}=\widetilde{\bSigma}$ from the R-package \textsf{mixedCCA} \citep{yoon2021mixedcca} which uses the default value of $\nu=0.01$, and we refer to $\widehat{\bSigma}$ as rank-based estimator.

In summary, to estimate $\bbeta^*$, we propose to replace $\bSigma_{21}$ and $\bSigma_{22}$ in \eqref{eq:mse} with the corresponding rank-based estimators $\widehat{\bSigma}_{21}$ and  $\widehat{\bSigma}_{22}$, respectively. In addition, we consider a $\ell_{1}$-regularization to account for high-dimensionality and to enhance the interpretability of the resulting classification rule. Specifically, we consider the following minimization problem:
\begin{flalign}\label{eq:cda_sampleObj}
    \bbetah = \argmin_{\bbeta \in \R^{p}} ~   
    \left(
    \frac12 \bbeta^{\top} \widehat{\bSigma}_{22}\bbeta - \bbeta^{\top} \widehat{\bSigma}_{21} + \lambda \|\bbeta\|_{1}
    \right),
\end{flalign}
where $\lambda>0$ is the tuning parameter that controls the sparsity level of $\bbetah$. This optimization problem is convex and can be efficiently solved via the coordinate descent algorithm. To obtain the solution for fixed $\lambda$, we utilize the solver written in \textsf{C} in the R--package \textsf{MGSDA} \citep{gaynanova2016MGSDApackage}. 

\subsection{Estimation of the classification rule}\label{subsec:estimation_classificationrule}
Here we describe how to obtain sample Bayes classification rule based on the optimal rule in Section~\ref{subsec:classrule} and estimated classification direction $\widehat \bbeta$. The key difficulty is that both approximations of posterior probability~\eqref{eq:postProb_MC} and~\eqref{eq:postProb_linear} rely on the latent $\bz$, which is unobservable. We first show how to estimate subvector of $\bz$ corresponding to non-zero observed values in $\bx$, $\bz_o$. We then use this estimator to generate posterior samples for $\bz_t$, subvector of $\bz$ corresponding to zero values in $\bx$, to use directly in~\eqref{eq:postProb_MC}, and in computing conditional mean of $\bz_t$ for~\eqref{eq:postProb_linear}.

Let $(Y_{i},\bx_{i}^{\top})^{\top} \in \R^{1+p}$, $i=1,\ldots,n$, be a sample from the latent Gaussian copula model for binary/truncated mixed data as in Definition \ref{def:LNPN_BT}. For each $i$, we write $\bx_{i,t}$ and $\bx_{i,o}$ to denote the truncated and observed subvectors of $\bx_{i}$, respectively. Similarly,
we denote the truncated and observed subvectors of a new observation $\bx^{\new}$ by $\bx^{\new}_{t} \in \R^{p_{t}}$ and $\bx^{\new}_{o} \in \R^{p_{o}}$.  We let $\widehat{\bDelta}=(\widehat{\Delta}_{1},\ldots,\widehat{\Delta}_{p})^{\top}$ be the threshold vector estimate, where $\widehat{\Delta}_{j}=\Phi^{-1}(\pihat_{j})$ and $\pihat_{j}=\sum_{i=1}^{n} 1(x_{ij}=0)/n$.

To estimate $\bz_o^{\new}$ corresponding to observed $\bx^{\new}_{o}$, we recall that from Definition~\ref{def:LNPN_BT} it holds that $z^{\new}_{j,o} = f_{j}(x^{\new}_{j,o}) = \Phi^{-1}\circ F_{j}(x^{\new}_{j,o})$, where $F_{j}$ is the marginal distribution function of the $j$th latent variable $X_{j}^{*}$ \citep{Liu:2009wz}. We propose to estimate $F_{j}$ using empirical cumulative distribution function, where we apply the winsorization similar to \citet{Han:2013ju} to avoid  $f_j(x^{\new}_{j,o})$ being infinite. Specifically, based on observations $\bx_{1},\ldots,\bx_{n}$, we consider 
\bfln{
\widehat{F}_{j}(t ; \delta_{n},x_{1j},\ldots,x_{nj}) = 
W_{j}^{\delta_{n}}
\left( 
\frac{1}{n} \sum_{i=1}^{n} 1(x_{ij}\leq t)
\right),
}
where
\bfln{
W_{j}^{\delta_{n}}(x) =
\begin{cases}
\pihat_{j},& \quad \text{if} \quad x \leq \pihat_{j}\\
x,& \quad \text{if} \quad \pihat_{j} < x \leq 1-\delta_{n}  \\
1-\delta_{n},& \quad \text{if} \quad 1-\delta_{n}< x.
\end{cases}
}
In our numerical studies, we use $\delta_{n}=1/(2n)$ as recommended by \citet{Han:2013ju}. Based on $\widehat{F}_{j}$'s, we set $\widehat{f}_{j} = \Phi^{-1}\circ \widehat{F}_{j}$ and obtain $\widehat{\bz}^{\new}_{o}$ as $\widehat{z}^{\new}_{j,o}=\widehat{f}_{j}(x^{\new}_{j,o})$. 

Using $\widehat{\bz}^{\new}_{o}$, we generate posterior samples $\bz^{\new(s)}_{t}$, $s=1,\ldots,S$, as follows. Let $\var(\bz_{t}^{\new}) = \bSigma_{t}$, $\var(\bz_{o}^{\new}) = \bSigma_{o}$, and $\cov(\bz_{o}^{\new},\bz_{t}^{\new}) = \bSigma_{ot}$. By properties of multivariate Gaussian distribution, $\E(\bz_t|\bz_o)= \bSigma_{ot}^{\top} \bSigma_{o}^{-1}{\bz}_{o}$ and $\var(\bz_t|\bz_0) = \bSigma_{t} - \bSigma_{ot}^{\top} \bSigma_{o}^{-1}
\bSigma_{ot}$. By plugging estimators $\widehat{\bz}^{\new}_{o}$, $\bSigmah_t$, $\bSigmah_o$, $\bSigmah_{ot}$, and using estimated truncation levels $\widehat{\bDelta}_{t}$ (the subvector of $\widehat{\bDelta}$ corresponding to $\bz_{t}^{\new}$), we obtain multivariate truncated Gaussian distribution with the probability density
\bfl{\label{eq:conditionalpdf}
p(\bz_{t}|\widehat{\bz}^{\new}_{o}, \bz_{t} < \widehat{\bDelta}_{t}) =
\frac
{\N_{p_{t}}(\bz_{t}|\widehat{\bmu},\widehat{\bGamma})}
{\PP(\bz_{t}<\widehat{\bDelta}_{t}| \bz^{\new}_{o}=\widehat{\bz}^{\new}_{o}) } 1(\bz_{t} < \widehat{\bDelta}_{t} ),
}
where 
$
\widehat{\bmu}=\bSigmah_{ot}^{\top} \bSigmah_{o}^{-1}\widehat{\bz}_{o}^{\new}$, $ \widehat{\bGamma}=\bSigmah_{t} - \bSigmah_{ot}^{\top} \bSigmah_{o}^{-1}
\bSigmah_{ot}.$ Combing posterior samples with the estimate $\widehat \beta$ from Section~\ref{subsec:estimation_beta} and $\widehat v = \sqrt{1-\bSigmah_{21}^{\top}\bSigmah_{22}^{-1}\bSigmah_{21}}$, gives estimate of the posterior probability \eqref{eq:postProb_MC}
\bfl{\label{eq:posterior_prob_mc}
    \PP\widehat{(Y=1|\bx)} = \frac{1}{S}\sum_{s=1}^{S}
    \Phi
    \left( 
    \frac{ \widehat{\bbeta}_{t}^{\top}{\bz}_{t}^{\new(s)} + \widehat{\bbeta}_{o}^{\top}\widehat{\bz}^{\new}_{o} - \widehat{\Delta}_{y} }{{\widehat v}}
    \right).
}
Similarly, using $\widetilde{\bmu}_{t}=S^{-1}\sum_{s=1}^{S} \bz^{\new(s)}_{t}$, we obtain estimate of posterior probability~\eqref{eq:postProb_linear}:
\bfl{\label{eq:posterior_prob_linear}
    \PP\widehat{(Y=1|\bx)} =
    \Phi
    \left( 
    \frac{\bbetah_{t}^{\top}\widetilde{\bmu}_{t} + \bbetah_{o}^{\top}\widehat{\bz}^{\new}_{o} - \widehat{\Delta}_{y} }{\widehat{v} }
    \right).
}

The corresponding sample Bayes rule assigns a new observation $\bx^{\new}$ to class 1 if the sample posterior probability is greater than 0.5 and to class 0, otherwise. Note that while expression~\eqref{eq:posterior_prob_linear} for posterior probability depends on $\widehat v$, the corresponding classification rule does not, similarly to standard LDA rule.


\section{Theoretical Results}\label{sec:theory}
In this section we demonstrate that the estimated classification direction $\widehat \bbeta$ from~\eqref{eq:cda_sampleObj} is a consistent estimator for $\bbeta^*$. We make the following assumptions.

\begin{assumption}[Latent correlations]\label{assumption1}
All the elements of $\bSigma$ satisfy $|\sigma_{jk}|\leq 1-\varepsilon_{r}$ for some $\varepsilon_{r}>0$.
\end{assumption}
\begin{assumption}[Thresholds]\label{assumption2}
All the thresholds $\Delta_{j}$ satisfy $|\Delta_{j}|\leq M$ for some constant  $M>0$.
\end{assumption}
\begin{assumption}[Condition number]\label{assumption:condition_number}
$\mathsf{C}(\bSigma) = \frac{\lambda_{\max}(\bSigma)}{\lambda_{\min}(\bSigma)} \leq C_{\cov}$ for some constant $C_{\cov}$.
\end{assumption}
\begin{assumption}[Sparsity] \label{assumption:sparsity} 
$\bbeta^{*}$ is sparse with the support $\mathcal{S}=\{j:\beta_{j}^{*}\neq 0\}$ with $s=\card(\mathcal{S})$.
\end{assumption}
\begin{assumption}[Sample size] \label{assumption:samplesize} 
$s\log(p) = o(n)$.
\end{assumption}

Assumptions~\ref{assumption1}-\ref{assumption2} are needed to guarantee consistency of estimated latent correlations in $\widehat \bSigma_{22}$, $\widehat \bSigma_{21}$ \citep{yoon2020sparse}. Assumptions~\ref{assumption:condition_number}--\ref{assumption:samplesize} are used to account for high-dimensional setting when $p$ is large, potentially much greater than $n$. We also take advantage of restricted eigenvalue condition.

\begin{definition}[Restricted eigenvalue condition]\label{def:restrictedeigenvalue}
A $p \times p$ matrix $\bSigma$ satisfies restricted eigenvalue condition $RE(s,3)$ with parameter $\gamma=\gamma(\bSigma)$ if for all sets $S \subset \{1,\ldots,p\}$ with $\card(S)\leq s$, and for all $\ba \in \mathcal{C}(S,c)=\{\ba \in \R^{p}: \|\ba_{S^{c}}\|_{1} \leq 3 \|\ba_{S}\|_{1} \}$, it holds that
\bfln{
\ba^{\top}\bSigma\ba \geq \frac{ \|\ba_{S}\|_{2}^{2} }{ \gamma}. }
\end{definition}
First, we provide deterministic bound on estimation error, which is a standard bound for high-dimensional sparse regression \citep{bickel2009simultaneous,hastie2015statistical,negahban2012unified}. For completeness, the proof is presented in the Appendix.
\begin{theorem}\label{thm:deterministic2}
Under Assumption \ref{assumption:sparsity}, if $\lambda \geq 2\| \bSigmah_{21}  - \bSigmah_{22}\bbeta^{*}\|_{\infty} $ and $\bSigmah_{22}$ satisfies RE(s,3) with parameter $\gamma$, then
\begin{flalign*}
 \|\bbetah - \bbeta^{*}\|_{2} \leq \frac{15}{2} \gamma \sqrt{s} \lambda.
\end{flalign*}
\end{theorem}

To derive probabilistic bound, we need to control the size of the tuning parameter, that is $\| \bSigmah_{21}  - \bSigmah_{22}\bbeta^{*}\|_{\infty}$, and also ensure that restricted eigenvalue condition on $\bSigma$ implies the condition holds for $\widehat \bSigma$. Both proofs are non-trivial under the model~\eqref{eq:model}. To illustrate the difficulty, the existing results on consistency of $\widehat \bSigma$ \citep{yoon2020sparse} provide the following high probability bounds: $\|\bSigmah_{21}-\bSigma_{21}\|_{\infty}\leq C\sqrt{\log p/n}$ and $\|\bSigmah_{22}-\bSigma_{22}\|_{\infty}\leq C\sqrt{\log p/n}$. A direct application of these results to control $\| \bSigmah_{21}  - \bSigmah_{22}\bbeta^{*}\|_{\infty}$ gives
\bfln{
\| \bSigmah_{21}  - \bSigmah_{22}\bbeta^{*}\|_{\infty} 
&\leq \| \bSigmah_{21}  -\bSigma_{21}\|_{\infty} +\|(\bSigma_{22}- \bSigmah_{22})\bbeta^{*}\|_{\infty} \\
&\leq C\sqrt{\log p/n} +\|\bSigmah_{22}-\bSigma_{22}\|_{\infty}\|\bbeta^*\|_1
\leq C_1\sqrt{\log p/n}\|\bbeta^*\|_1.
}
The above bound is clearly suboptimal as $\|\bbeta^*\|_1$ scales approximately like $\sqrt{s}$, implying that the knowledge of true sparsity level is required to choose the tuning parameter $\lambda$. In contrast, the results from sparse high-dimensional regression \citep{bickel2009simultaneous,hastie2015statistical,negahban2012unified} suggest that the optimal rate should be of the order $\sqrt{\log p/n}$ without the extra dependence on $s$.

Our main theoretical result is obtaining such a bound for $\| \bSigmah_{21}  - \bSigmah_{22}\bbeta^{*}\|_{\infty}$ term under model~\eqref{eq:model}.
\begin{theorem}
Under Assumptions \ref{assumption1}--\ref{assumption:condition_number}, for some constant $C>0$, 
\begin{flalign*}
    \| \bSigmah_{21}  - \bSigmah_{22}\bbeta^{*}\|_{\infty} \leq C \sqrt{ \frac{\log(p)}{n} }
\end{flalign*}
with probability at least $1-p^{-1}$.
\end{theorem}
The full proof is presented in Appendix, and here we summarize argument at a high level. To our knowledge, the only similar result is obtained by \citet{barber2018rocket} in the case of continuous Gaussian copula of Definition~\ref{def:NPN}. Their proof, however, takes advantage of the closed form of the inverse bridge function $G^{-1}$ in the continuous case (it is a scaled cosine function). Due to significantly higher complexity of bridge function $G_{TT}$~\eqref{eq:bridgeForm} in the truncated Gaussian copula case, its inverse is not available in closed-form, leading to a more challenging proof. Additional complication arises from substitution of true thresholds $\Delta_j$ with their estimators $\widehat \Delta_j$, these thresholds being unique to truncated case. To overcome these challenges, we consider a 2nd-order Taylor expansion of $\widehat \sigma_{jk} = G_{TT}^{-1}(\widehat \tau_{jk}, \widehat \Delta_j, \widehat \Delta_k)$ with respect to $ \sigma_{jk} = G_{TT}^{-1}(\tau_{jk}, \Delta_j, \Delta_k)$. To control the first-order terms, we combine the bound on first derivatives of inverse bridge functions \citep{yoon2020sparse} with a concentration bound for deviations of quadratic forms involving the Kendall’s $\tau$ correlation matrix \citep{barber2018rocket}[Lemma~E.2] and sign sub-Gaussian property of the Gaussian vectors \citep{barber2018rocket}[Lemma~4.5]. To control the second-order terms, we establish that the second derivatives of inverse bridge functions are bounded, and use these bounds in conjunction with element-wise convergence of $\widehat \bSigma$ and $\widehat \bDelta$. Due to inverse bridge function not being available in closed form, establishing bounds on second derivative is highly non-trivial, and is a major technical part of the proof. A similar technique is used to prove that the restricted eigenvalue condition on $\bSigma$ implies the condition holds for $\widehat \bSigma$, leading to our final estimation bound.
\begin{theorem}
Under Assumptions \ref{assumption1}--\ref{assumption:samplesize}, if $\lambda = C\sqrt{\log(p)/n}$ for some constant $C>0$ and $\bSigma_{22}$ satisfies RE(s,3) with parameter $\gamma$,
then
\bfln{
\|\bbetah - \bbeta^{*}\|_{2}^{2} = O_{p}\left(\gamma^2 \frac{s \log(p)}{n}\right).
}
\end{theorem}
The obtained rate in estimation error coincides with the optimal rate in sparse linear regression \citep{bickel2009simultaneous,hastie2015statistical,negahban2012unified}.

\section{Simulation}\label{sec:simulation}
In this section, we empirically evaluate the performance of the proposed method, which we refer to as Copula Linear Discriminant Analysis (CLDA). We consider two approaches for estimating class-conditional probabilities as described in Section~\ref{subsec:estimation_classificationrule}: formula \eqref{eq:posterior_prob_mc} based on Monte Carlo approximation (CLDA\_MC) and formula \eqref{eq:posterior_prob_linear} based on Taylor approximation (CLDA\_linear), where we use MC samples of size $S=100$ from \eqref{eq:conditionalpdf}. For comparison, we consider high-dimensional COpula Discriminant Analysis (CODA) of \citet{Han:2013ju}, and Oracle classifier, where Oracle classifier uses the population classification rule at the latent Gaussian level.  

To generate synthetic data, we fix the number of covariates $p=300$, and consider three correlation structures for the latent Gaussian vector associated with the covariates:
\begin{itemize}
	\item Auto-regressive (AR): $ \bSigma_{22} = [0.7^{|j-j'|}]_{1 \leq j,j' \leq p}$. 
	\item Compound symmetry (CS): $ \bSigma_{22} = 0.3  \bI_p + 0.7  \ones_{p,p}$, where $\bI_p$ is an identity matrix and $\ones_{p,p}$ is the $p \times p$ matrix of ones.
	\item Geometric decaying eigenvalues (GD): $ \bSigma_{22} =  \bGamma  \bN \bGamma^{\top} $ where $\bGamma$ is generated from the uniform distribution on $p$-dimensional orthogonal group following Theorem 2.2.1 of \citet{chikuse2012statistics} and $\bN$ is a diagonal matrix with geometrically decaying eigenvalues $\nu_1>\nu_2>\cdots>\nu_p$. 
	\[
	\nu_{j}=\frac{ p(0.9^{j-1}-0.9^{j})}{ 1-0.9^{p} }, \quad j=1,\ldots,p.
	\]
\end{itemize}
In the AR setting, each variable is highly correlated with only a few  variables. In the CS setting, all variables are moderately correlated. In the GD setting, the correlation structure mimics the one frequently observed in real high-dimensional data \citep{lee2013hdlss}. 

For fair comparison with CODA, we consider two types of models for data generation: joint model as in Definition~\ref{def:LNPN_BT} (favoring the proposed CLDA), and mixture model as defined in \citet{Han:2013ju} (favoring CODA). For both joint and mixture settings, we choose monotone transformations and zero-inflation thresholds to mimic the quantitative microbiome profiling (QMP) data of \citet{vandeputte2017quantitative}, where data generation details are given in Sections \ref{subsec:sim_joint} and \ref{subsec:sim_mixture}. We generate the true classification direction vector $\bbeta^*$ so that only the first $s=p\times 5\%=15$ variables are non-zero,  $\mathcal{S}=\{j:\beta_{j}^{*}\neq 0 \}=\{1,\ldots,15\}$. 
Throughout, we fix the sample sizes of training and test data at $n=150$ and $n_{\textup{test}}=300$, respectively. 



\subsection{Joint model}\label{subsec:sim_joint}

We generate data from the latent Gaussian copula model for binary/truncated mixed data as in Definition \ref{def:LNPN_BT}. Recall that given full correlation matrix $\bSigma$, the population direction is given by $\bbeta^* = \bSigma_{22}^{-1}\bSigma_{21}$. To generate $\bbeta^*$ with a given support $S=\{j: \beta_{j}^{*}\neq 0 \}$ for each of the three correlation structures $\bSigma_{22}$ from above, we define $\bSigma_{21}$ as follows.

Let $\bb=(b_{1},\ldots,b_{p})^{\top} \in \{0,1\}^{p}$ be the indicator vector for the signal variables such that $b_{j}=1$ if $j \in \mathcal{S}$ and $b_{j}=0$, otherwise.  Let $v^2=1-\bSigma_{21}^{\top}\bSigma_{22}^{-1}\bSigma_{21} = 0.05$ be the prespecified conditional variance of $Z_{y}|\bz$. Then setting $\bSigma_{21} =  (\sqrt{1-v^2}/\sqrt{\bb^{\top}\bSigma_{22}\bb}) \cdot \bSigma_{22}\bb$ ensures positive-definiteness of the full correlation matrix $\bSigma$ with corresponding $\bbeta^* = (\sqrt{1-v^2}/\sqrt{\bb^{\top}\bSigma_{22}\bb})\cdot \bb$.

Given $\bSigma$, we follow the synthetic microbiome data generation mechanism proposed in \citet{Yoon2019microbial}. Specifically, we select monotone transformations and truncation levels so that the resulting synthetic $\bx$ follow the empirical marginal cumulative distributions of the real QMP variables in \citet{vandeputte2017quantitative}. To investigate the effect of truncation, we divide all 101 QMP variables according to three truncation levels: no truncation (0\%), low truncation (10\%-50\%), and high truncation (40\%-80\%). For each level, we use empirical cdfs of corresponding QMP variables to generate $p=300$ covariates (as the number of QMP variables is less than 300, we use the same empirical cdf to generate multiple synthetic variables).

Let $\widetilde{F}_{j}$ be the empirical cdf chosen to represent variable $X_{j}$. For $i=1,\ldots,n$, we generate $(Z_{y,i},\bz_{i})^{\top} \sim \N_{1+p}(\zeros,\bSigma)$ and obtain $y_{i}$ and $\bx_{i}=(x_{i1},\ldots,x_{ip})^{\top}$ as
$
Y_{i}=1(Z_{y,i}>0)$, $x_{ij} = \widetilde{F}_{j}^{-}\circ \Phi(z_{ij})$, $j=1,\ldots,p$,
where $\widetilde{F}_{j}^{-}(u)=\min \{x_{ij}| \widetilde{F}_{j}(x_{ij}) \geq u \}$. As we threshold $Z_{y,i}$ at $\Delta_y = 0$, the resulting class sizes are approximately equal. Marginally, this data generation scheme for $\bx_{j}$ is the uniform sampling with replacement of the observations of the $j$th QMP variable but the joint association structure is induced by the prespecified latent correlation matrix $\bSigma$. 
Under the joint model, we define the Oracle classification rule as
\bfl{\label{eq:oracle_joint}
 \delta_{J}(\bx) = 1\left\{\bdf(\bx)^{\top}\bbeta^* >0\right\}.
}

\subsection{Mixture model}\label{subsec:sim_mixture}

\cite{Han:2013ju} consider the following model:
\bfl{\label{eq:coda_mixture_model}
\bx|Y=g \sim \NPN(\bmu_{g},\bSigma,\bdf), \quad g=0,1,
}
where $\bmu_{g}$ is the mean of class $g=0,1$ and $\bSigma \in \R^{p \times p}$ is a common covariance matrix. Thus, unlike Definition~\ref{def:NPN}, CODA allows latent Gaussian vector to have non-zero mean and covariance matrix by restricting monotone transformations $\bdf$ to be mean and variance preserving, i.e., each $f_j$ satisfies
\bfl{\label{eq:coda_moment_assumption}
\E\{f_{j}(X_{j})|Y=g\} &= \E(X_{j}|Y=g) =  \mu_{g,j},\\
 \var\{f_{j}(X_{j})|Y=g\} &= \var(X_{j}|Y=g) = \sigma^{2}_{j}, \quad j=1,\ldots,p.
}
Note that this model does not account for zero inflation; that is, it assumes continuous $\bx$.

To generate realistic simulation data, we set $\bSigma= \bS^{-1}\bSigma_{22}\bS^{-1}$, where $\bS=\diag(s_{1},\ldots,s_{p})$ contains the sample standard deviations from QMP data and $\bSigma_{22}$ is one of the three correlation structures described above. We set the means $\bmu_g$ based on discriminant direction $\bbeta^*$ as follows.

Let $\bmu_{a}=(\bmu_{1}+\bmu_{0})/2$ and $\bmu_{d}=\bmu_{1}-\bmu_{0}$ be the overall mean and mean difference, respectively. Under model \eqref{eq:coda_mixture_model}, the Bayes classification direction is given by $\bbeta^* = \bSigma^{-1}\bmu_{d}$. 
When $\PP(Y=1) = \PP(Y=0)$, the Bayes error rate is given by $\alpha = \Phi\left(-\sqrt{{\bbeta^*}^{\top}\bSigma\bbeta^*}/2\right)$. Given the support $\mathcal{S}$, let $\bb \in \{0,1\}^{p}$ be the corresponding indicator vector such that $b_{j}=1$ if $j \in \mathcal{S}$ and $b_{j}=0$, otherwise. Fixing the Bayes error rate $\alpha = 0.2$, we generate $\bbeta^*$ as $\bbeta^{*} = -2\Phi^{-1}(\alpha)\bb/\sqrt{\bb^{\top}\bSigma \bb}$, and obtain $\bmu_d = \bSigma\bbeta^*$. Finally, we set $\bmu_0 = C$ and $\bmu_1 = C + \bmu_d$, where the constant $C>0$ is chosen sufficiently large to mimic the means of QMP data, leading to generated synthetic data with non-negative values. 

Given $\bmu_g$ and $\bSigma$ from above, we generate Gaussian $\bz_{1},\ldots,\bz_{n}$ with equal class sizes, i.e., $n_{1}=n_{2} = n/2 = 75$, such that $\bz_i|Y_i = g\sim \mathcal{N}(\bmu_g, \bSigma)$. To obtain continuous $\bx_{1}^{*},\ldots,\bx_{n}^{*}$ that follow model~\eqref{eq:coda_mixture_model}, we use the following mean and variance preserving strictly monotone transformation \citep{Han:2013ju,Liu:2009wz}:
\bfln{
x_{ij}^{*} = s_{j}\left[ \sqrt{12} \left\{ f^{-1}_{j}(z_{ij}) - 0.5 \right\} \right] + \mu_{g,j},
}
where $f_{j}^{-1}(z) = \Phi\{(z-\mu_{g,j})/s_{j}\}$. 

To generate zero-inflated data, we further obtain $\bx_1, \ldots, \bx_n$ by  applying truncation
 $x_{ij}=1(x_{ij}^{*}>D_{j})x_{ij}^{*}$, where $D_{j}$, $j=1, \dots, p$, are independently drawn from $\Unif(0.1,0.5)$ (low truncation level) or $\Unif(0.4,0.8)$ (high truncation level). We set $\bx_{i}=\bx_{i}^{*}$ for the setting of no truncation (no zero-inflation). 
 Under the mixture model, the Oracle classification rule is given by
\bfl{\label{eq:oracle_mixture}
 \delta_{M}(\bx) = 1\left\{
 \left(\bdf(\bx)-\bmu_{a}\right)^{\top}\bbeta^* >0\right\}.
}

\subsection{Training and evaluation}\label{subsec:training}

The sparsity tuning parameters for both CLDA and CODA are chosen from the grid of length 100 using 5-fold cross-validation with the misclassification error rate as the tuning criterion. The optimization problems of CLDA and CODA involve the $\ell_{1}$ regularization, which shrinks the solution towards zero vector. Thus, the intercept terms $\widehat{\Delta}_{y}$ and $\log(n_{1}/n_{0})$ for the proposed model and CODA no longer produce optimal predictions. For the proposed model, we compensate the shrinkage effect by tuning $\Delta_{y}$ with the 100 equally spaced grid ranging from -1.5 to 1.5. For CODA, we use the optimal intercept for a sparse LDA \citep{Mai:2012bf}; the details are in Appendix \ref{subsec:coda_implementation}.

We assess the prediction performances using the out-of-sample misclassification rate on test data of size $n_{\text{test}}=300$. We consider 100 replications for each model type, correlation structure and truncation level. 
Figure \ref{fig:sim_joint} displays the results for the joint model in Section~\ref{subsec:sim_joint}. The proposed CLDA uniformly outperforms CODA, and the difference becomes more significant as truncation level increases. These results are expected, as the joint model favors the proposed CLDA, and CODA is not designed for zero-inflated data, leading to poor performance at high truncation levels.  Both approximations~\eqref{eq:posterior_prob_mc} and~\eqref{eq:posterior_prob_linear} of the conditional class probability lead to similar misclassification rates across all settings.

\begin{figure}[!t]
	\centering
	 \includegraphics[width=\textwidth]{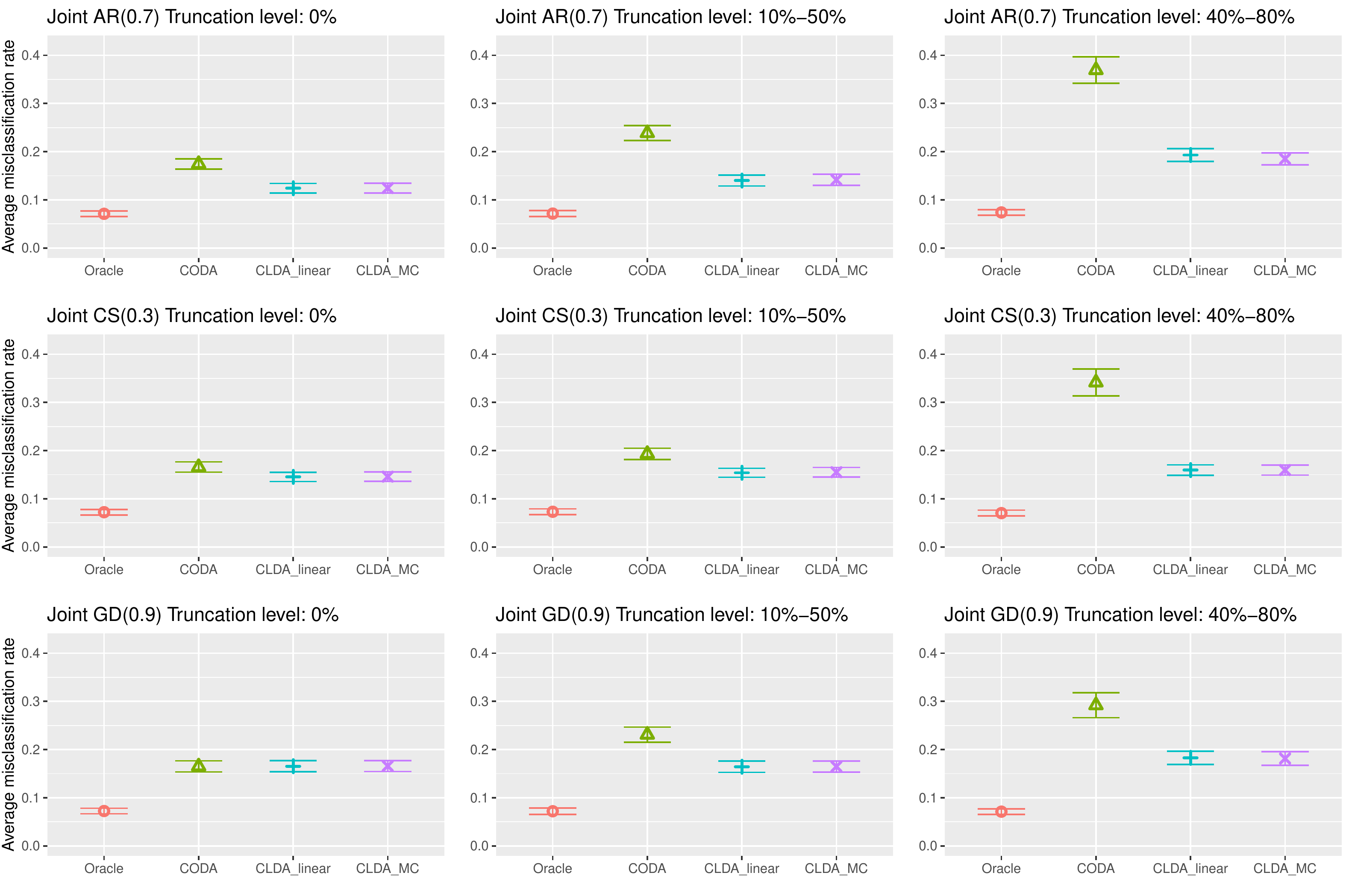}\\
	\caption{Joint model, average misclassification rates and 2 standard error bars over 100 replications. Columns correspond to truncation levels and rows correspond to correlation structures. Compared models are Oracle \eqref{eq:oracle_joint}, CLDA\_linear \eqref{eq:posterior_prob_linear}, CLDA\_MC \eqref{eq:posterior_prob_mc}, and CODA \citep{Han:2013ju}. } \label{fig:sim_joint}
\end{figure}

\begin{figure}[!t]
	\centering
	 \includegraphics[width=\textwidth]{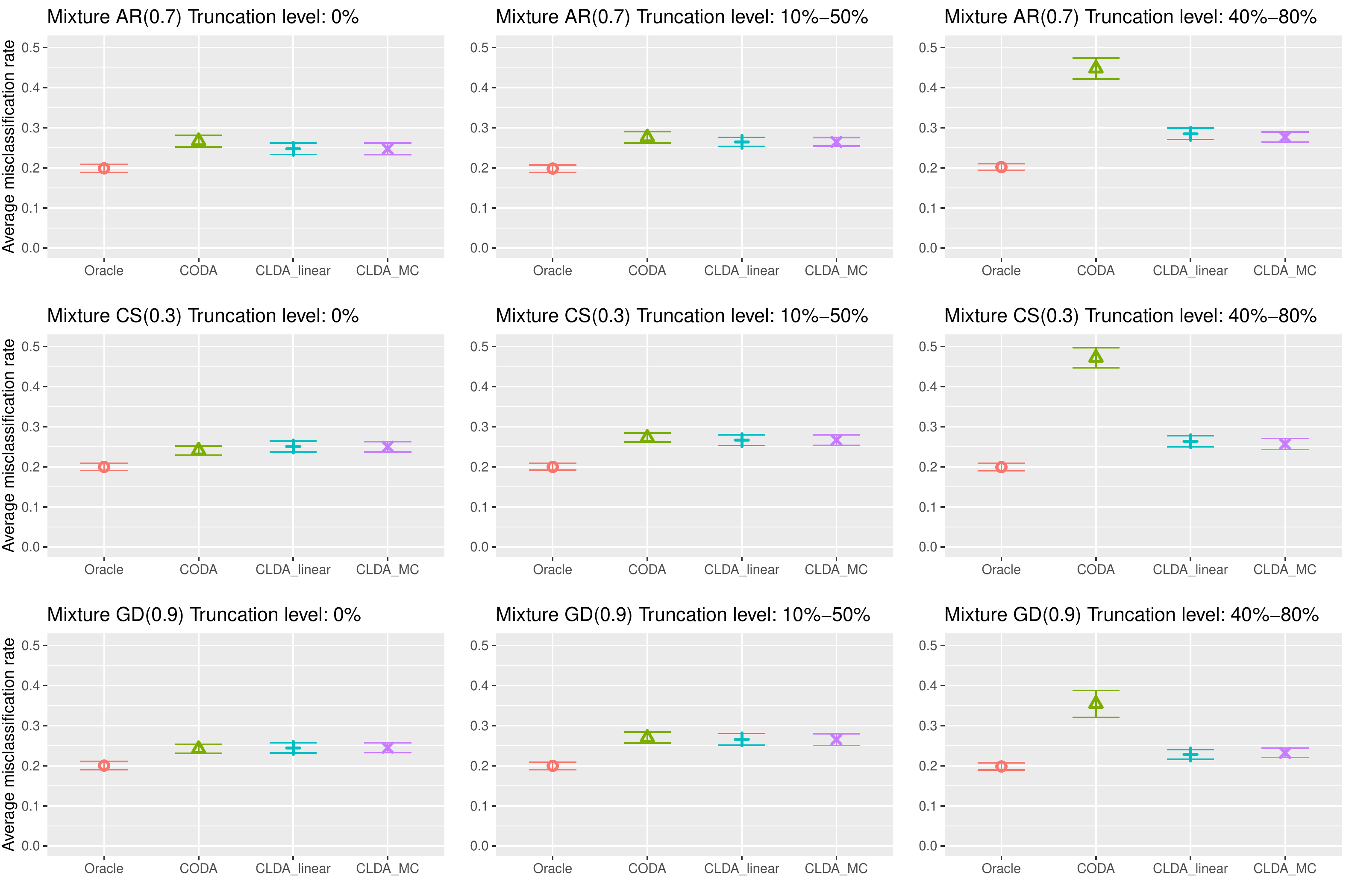}\\
	\caption{Mixture model, average misclassification rates and 2 standard error bars under over 100 replications. Columns correspond to truncation levels and rows correspond to correlation structures. Compared models are Oracle \eqref{eq:oracle_mixture}, CLDA\_linear \eqref{eq:posterior_prob_linear}, CLDA\_MC \eqref{eq:posterior_prob_mc}, and CODA \citep{Han:2013ju}.} \label{fig:sim_mixture}
\end{figure}

Figure \ref{fig:sim_mixture} displays the results for the mixture model in Section~\ref{subsec:sim_mixture}. While the mixture model favors CODA, the proposed CLDA performs as well as CODA when truncation levels are less than 50\%: no truncation and low truncation settings. However, CODA shows a significant deterioration under the high truncation level, having average misclassification rates close to random guessing (0.5) under AR and CS covariance settings. This is not the case for proposed CLDA, which maintains excellent performance even when truncation levels are high. As in the joint model, the two approximations~\eqref{eq:posterior_prob_mc} and~\eqref{eq:posterior_prob_linear} of the conditional class probability, lead to similar misclassification rates across all truncation levels and covariance settings.

\section{Application to Quantitative Microbiome Data}\label{sec:application}

We consider the QMP data  of \citet{vandeputte2017quantitative}, processed as in \citet{Yoon2019microbial}. The data contain $p=101$ genera from $n_1=29$ patients with Crohn's disease $(Y=0)$, a type of inflammatory bowel disease, and $n_2=106$ healthy controls $(Y=1)$. Out of 101 genera, 39 have unknown names, and we label them as Unknown 1 -- Unknown 39. The proportions of zeros across all 101 genera range from 1\% to 96\%, with 14 genera having no zeros. Our goal is to construct a classification rule that can separate patients with Crohn's disease from the healthy controls, and identify genera that have the most influence on classification.

We compare the performance of CLDA\_linear, CLDA\_MC and CODA as in Section~\ref{sec:simulation}. We randomly split the data into 4/5 and 1/5 and use the former to train the classifiers via 5-fold cross-validation and the latter to evaluate their misclassification error rate. We choose the tuning parameters as described in Section~\ref{subsec:training} and use MC samples of size $S=300$ from \eqref{eq:conditionalpdf} to compute the posterior probabilities \eqref{eq:posterior_prob_mc} and \eqref{eq:posterior_prob_linear}. Table \ref{tab:qmp_summary} displays the average misclassification rates and model sizes over 30 random splits. Overall, the proposed CLDA significantly outperforms CODA with lower misclassification rate and smaller model size. Specifically, CLDA has the average misclassification rate of less than 3\%, whereas CODA's rate is around 9\%. Futhermore, CLDA uses fewer genera for classification, selecting 14 variables on average, which is about 1/3 of the average model size of CODA, 37. The associated standard error of CLDA is also significantly smaller, indicating a more stable variable selection. 
The two approximations of the posterior probability used by CLDA,~\eqref{eq:posterior_prob_mc} and~\eqref{eq:posterior_prob_linear}, lead to nearly identical misclassification rates, model sizes, and standard errors.

\begin{table}[!t]
\centering
\caption{QMP data analysis, average misclassification rates and model sizes over 30 random splits with corresponding standard errors in parentheses. } \label{tab:qmp_summary}
\begin{tabular}{l|cc}
\hline
              & Misclassification rate & Average model size \\
\hline
CLDA\_linear  &  0.028 (0.005)    &  13.4 (0.8) \\
CLDA\_MC      &  0.027 (0.006)    &  13.7 (1.0) \\
CODA                    &  0.092 (0.011)    &  36.8 (2.8) \\
\hline
\end{tabular}
\end{table}

We further investigate and compare the selected genera by CLDA and CODA by considering their selection probabilities across 30 random splits. Table \ref{tab:qmp_variables} lists the genera that are selected more than 90\% of the times, i.e., at least 27 times out of 30. As before, there is agreement between the results of CLDA\_linear and CLDA\_MC. Both CLDA and CODA consistently select \emph{Fusobacterium}. We find that the CLDA-estimated coefficients of \emph{Fusobacterium} in $\widehat \bbeta$ are negative across all random splits indicating that an increase in \emph{Fusobacterium} is positively associated with Crohn's disease development  (labelled as $Y=0$). This result is consistent with existing findings \citep{han2015fusobacterium, strauss2011invasive} of positive correlation between \emph{Fusobacterium} and inflammatory bowel disease. When looking beyond \emph{Fusobacterium}, CLDA and CODA select non-overlapping genera. Observe that the remaining 12 genera with highest frequency for CODA all have no zeros or just a few zeros (less than 12\%). In constrast, the remaining 2 genera for CLDA have a high proportion of zeros, with \emph{Succinivibrio} (58.2\% of zeros) selected 28 out of 30 times, and the unknown genera in the \emph{Comamonadaceae} family (79.3\% of zeros) selected 30 out of 30 times. We consider the CLDA-estimated coefficients of \emph{Succinivibrio}, which are positive across all splits. This indicates that a subject with a high level of \emph{Succinivibrio} is likely to be healthy. This finding is consistent with existing results from animal studies. \citet{werner2011succini} argue that depletion of luminal iron can prevent Crohn's disease, and subsequently find that mice feeding on an iron sulfate free diet show a significant increase of \emph{Succinivibrio}. Given this biological evidence and the fact that CLDA has better misclassification error rate, we suspect that high zero inflation in \emph{Succinivibrio} makes CODA miss this important variable, and that CODA has selection bias towards variables with smaller proportion of zeros as is seen in Table~\ref{tab:qmp_variables}.

\begin{table}[!t]
\centering
\caption{QMP data analysis, list of genera with selection frequency over 90\% (more than 27 out of 30 random splits). \ech The numbers in parentheses indicate zero proportions.} \label{tab:qmp_variables}
\begin{tabular}{lll}
\hline
 \multicolumn{3}{c}{Genera selected more than 27 times} \\
  CLDA\_linear & CLDA\_MC & CODA \\ 
 \hline
   Fusobacterium  & Fusobacterium (48.9\%) & Fusobacterium  (48.9\%) \\ 
   Succinivibrio & Succinivibrio (58.2\%)  & Bacteroides (0\%) \\ 
   Unknown 8 & Unknown 8 (79.3\%) & Bifidobacterium (2.2\%)\\ 
    &  & Blautia (0\%) \\
    &  & Coprococcus (0\%) \\ 
    &  & Faecalibacterium (0\%) \\ 
    &  & Prevotella (11.9\%) \\ 
    &  & Roseburia (0\%) \\ 
    &  & Ruminococcus (0\%) \\ 
    &  & Unknown 11 (0\%) \\ 
    &  & Unknown 25 (0\%) \\ 
    &  & Unknown 35 (0\%) \\ 
    &  & Unknown 5 (0\%) \\ 
   \hline
\end{tabular}
\end{table}




\section{Discussion}\label{sec:disc}
In this work we develop a binary classification model for zero-inflated data, and show the consistency of estimated classification direction in high-dimensional settings. On simulated data, our method achieves similar classification accuracy as existing competitor on continuous data, and significantly better accuracy on zero-inflated data. In an application to quantitative microbiome data, our method achieves significantly smaller misclassification error rate while also leading to a more parsimonious model.  The selected genera align with existing results in microbiome literature.

There are several further research directions that could be pursued. First, our estimation consistency result is non-trivial leading us to develop a new technique to facilitate underlying theoretical analyses by combining sub-Gaussian properties of sign vector with newly established bounds on second derivatives of inverse bridge function for truncated/truncated cases. While this theoretical analysis has been restricted to mixed binary-truncated model in Definition~\ref{def:LNPN_BT}, the developed technique can be applied more generally. In particular, it can be used to establish estimation consistency in high-dimensional settings within general semiparametric Gaussian copula regression modeling framework \citep{dey2022semiparametric}, which encompasses all continuous, binary, ordinal and truncated variable types. This generalization requires establishing bounds on first and second derivatives of corresponding inverse bridge functions, which while technical, can be accomplished in the similar ways as Section \ref{suppl:subsec:bridge_inv_bound} in Appendix.
Secondly, our focus here is on binary classification problem, and multi-class extensions are of interest. In case the classes have a natural ordering, e.g., a disease classification as ``Mild", ``Moderate" or ``Severe", the extension is straightforward by considering mixed ordinal-truncated model, with corresponding bridge function for ordinal/truncated case as derived in \citet{huang2021latentcor}. However, it is unclear how to incorporate unordered class labels due to ambiguity in the underlying latent Gaussian representation, making it a compelling question for future study. 


\section*{Acknowledgements}
The authors thank Grace Yoon for constructive discussions on QMP data analysis. This work has been partially supported by the Texas A\&M Institute of Data Science (TAMIDS) and the Texas A\&M Strategic Transformative Research Program. Gaynanova's research was partially supported by the National Science Foundation (NSF CAREER DMS-2044823). Ni's research was partially supported by the National Science Foundation (NSF DMS-1918851, NSF DMS-2112943).
Portions of this research were conducted with the advanced computing resources provided by Texas A\&M High Performance Research Computing.

\newpage
\appendix
\section{Appendix}
\subsection{First-order Taylor approximation of the posterior probability}\label{appendix:Taylor_posterior_prob}
Let $\bmu_{t}=\E(\bz_{t}|\bz_{o},\bz_{t}\leq \bDelta_{t})$ and
\begin{align*}
    g(\bz_{t}) = 
    \Phi 
    \left( 
    \frac{{\bbeta^*_{t}}^{\top}\bz_{t} + {\bbeta^*_{o}}^{\top}\bz_{o} - \Delta_{y} }{v}
    \right). 
\end{align*}
Then, by Taylor,
\begin{align*}
\PP(Y=1|\bx) &=
    \E\left\{ 
    g(\bz_{t})
    \big|
    \bz_{o}, \bz_{t}<\bdelta_{t}
    \right\}\\
    &\approx \E\left\{ 
    g(\bmu_{t}) + \nabla g(\bmu_{t})^{\top}(\bz_{t}-\bmu_{t})
    \big|
    \bz_{o}, \bz_{t}<\bdelta_{t}
    \right\} \\
    &= g(\bmu_{t})
\end{align*}

\subsection{Implementation details of CODA} \label{subsec:coda_implementation}
CODA \citep{Han:2013ju} assumes that data satisfy the linear discriminant analysis assumption at the latent Gaussian level; that is,
\bfl{\label{eq:coda_mixture_model2}
\bx|(Y=g) \sim \NPN(\bmu_{g},\bSigma,\bdf), \quad g=0,1,
}
where $\bmu_{g}$ is the mean of class $g=0,1$ and $\bSigma \in \R^{p \times p}$ is a common covariance matrix. To address the identifiability issue, CODA assumes that $f_{j}$'s preserves the population mean and variance, i.e.,
\bfl{\label{eq:coda_moment_assumption_appendix}
\E(X_{j}) = \E\{f_{j}(X_{j})\} = \bmu_{g} \quad \text{and} \quad \var(X_{j}) = \var\{f_{j}(X_{j})\} = \sigma^{2}_{j}.
}
Under this model, the Bayes rule assigns a new observation $\bx$ to class 1 if
\bfln{
(\bdf(\bx)-\bmu_{d})^{\top} \bbeta + \log\left( \frac{\PP(Y=1)}{\PP(Y=0)} \right) >0
}
and class 0, otherwise, where $\bbeta=\bSigma^{-1}(\bmu_{1}-\bmu_{0})$ and $\bmu_{a} = (\bmu_{1}+\bmu_{0})/2$. 
The Bayes classification direction is estimated by solving
\bfl{\label{eq:coda_sampleObj}
\bbetah_{\textup{CODA}} = \underset{\bbeta \in \R^{p}}{\argmin}
\left\{ 
\frac{1}{2} \bbeta^{\top} \widehat{\bS} \bbeta + \frac{\nu}{2}( \bbeta^{\top}\bmu_{d} -1 )^2 + \lambda\|\bbeta\|_{1},
\right\}
}
where $\widehat{\bS} = n_{1}\widehat{\bS}_{1}/n + n_{0}\widehat{\bS}_{0}/n$, $\widehat{\bS}_{1}$ and $\widehat{\bS}_{0}$ are sample Kendall's $\tau$ covariance matrices and $\nu = (n_{0}n_{1})/n^2$.
Given $\widehat{\bbeta}_{\textup{CODA}}$, the optimal intercept $\widehat{\beta}_{0}^{opt}$ for a sparse LDA \citep{Mai:2012bf} is given by
\bfln{
\widehat{\beta}_{0}^{opt} = -\widehat{\bmu}_{a}^{\top}\bbetah_{\textup{CODA}} + 
\frac{  \bbetah_{\textup{CODA}}^{\top} \widehat{\bS}\bbetah_{\textup{CODA}} }
{\widehat{\bmu}_{d}^{\top}\bbetah_{\textup{CODA}}} \log\left( \frac{n_{1}}{n_{0}} \right).
}

Let $\widehat{\bmu}_{g}$, $g=0,1$, be the sample class means and $\widehat{\bmu} = \widehat{\bmu}_{1} - \widehat{\bmu}_{0}$. The sample classification rule assigns a new observation $\bx^{\new}$ to class 1 if
\bfln{
(\widehat{\bdf}(\bx^{\new})-\widehat{\bmu}_{d})^{\top} \widehat{\bbeta}_{\new} + \widehat{\beta}_{0}^{\textup{opt}} >0.
}

\section{Proofs of theoretical results}

Throughout, we use $G$ to denote the bridge function such that for TT case $\tau_{jk} = G(\Sigma_{jk}, \Delta_j, \Delta_k)$ with $\bSigma_{22} = \left[ G^{-1}(\tau_{jk},\Delta_{j},\Delta_{k}) \right]_{1\leq j,k \leq p} = G^{-1}(\bT, \bDelta)$ and $\bSigmah_{22} = G^{-1}(\bThat, \bDeltahat)$. Here $\bT$ and $\bThat$ are matrices of population and sample Kendall's $\tau$ values, respectively, $\bDelta = (\Delta_1, \dots, \Delta_p)^{\top}$, $\bDeltahat = (\widehat \Delta_1, \dots, \widehat \Delta_p)^{\top}$ with $\Delta_j = \Phi^{-1}(\pi_j)$ with $\pi_j = \PP(Z_j \leq \Delta_j) = \PP(X_j =0)$, and $\widehat \Delta_j = \Phi^{-1}(\widehat \pi_j)$ with $\widehat \pi_j = n_0/n$ being the observed proportion of zeros in the $j$th variable.

\subsection{Proofs of theorems from the main manuscript}

\begin{definition}[Restricted eigenvalue condition]\label{def:restrictedeigenvalue}
A $p \times p$ matrix $\bSigma$ satisfies restricted eigenvalue condition $RE(s,c)$ with parameter $\gamma=\gamma(s,c,\bSigma)$ if for all sets $S \subset \{1,\ldots,p\}$ with $\card(S)\leq s$, and for all $\ba \in \mathcal{C}(S,c)=\{\ba \in \R^{p}: \|\ba_{S^{c}}\|_{1} \leq c \|\ba_{S}\|_{1} \}$, it holds that
\bfln{
\ba^{\top}\bSigma\ba \geq \frac{ \|\ba_{S}\|_{2}^{2} }{ \gamma }.
}
\end{definition}

\begin{theorem}\label{thm:deterministic2}
Under Assumption \ref{assumption:sparsity}, if $\lambda \geq 2\| \bSigmah_{21}  - \bSigmah_{22}\bbeta^{*}\|_{\infty} $ and $\bSigmah_{22}$ satisfies RE(s,3) with parameter $\gamma$, then
 $$
 \|\bbetah - \bbeta^{*}\|_{2} \leq \frac{15}{2} \gamma \sqrt{s} \lambda.
 $$
\end{theorem}
\begin{proof}[\textbf{ Proof of Theorem \ref{thm:deterministic2}}] This proof follows the proof of Theorem~2 in \citet{MyBernoulli}. We repeat it here for completeness.

By the optimality condition of equation (9) in the main paper, we have
\begin{flalign*}
    \bSigmah_{22}\widehat{\bbeta} - \bSigmah_{21} + \lambda \bg = \zeros,
\end{flalign*}
where $\bg$ is a subgradient of $\|\bbeta\|_{1}$ at $\bbetah$. This gives
\begin{flalign*}
    (\widehat{\bbeta} - \bbeta^{*})^{\top}(\bSigmah_{22}\widehat{\bbeta} - \bSigmah_{21} + \lambda \bg) = 0,
\end{flalign*}
and thus
\begin{flalign} \label{eq:optimal1}
(\widehat{\bbeta} - \bbeta^{*})^{\top} \bSigmah_{22} (\widehat{\bbeta} - \bbeta^{*}) - 
    (\widehat{\bbeta} - \bbeta^{*})^{\top}( \bSigmah_{21}  - \bSigmah_{22}\bbeta^{*}) + 
    \lambda(\widehat{\bbeta} - \bbeta^{*})^{\top}  \bg = 0.
\end{flalign}
Since $\bg$ is a subgradient of $\|\bbeta\|_{1}$ at $\bbetah$, 
\begin{flalign}\label{eq:subgrad}
(\widehat{\bbeta} - \bbeta^{*})^{\top} \bg \geq \|\bbetah\|_{1} - \|\bbeta^{*}\|_{1}.
\end{flalign}
By combining \eqref{eq:optimal1}, \eqref{eq:subgrad}, and H{\"o}lder's and triangle inequalities, we have
\begin{flalign*}
(\widehat{\bbeta} - \bbeta^{*})^{\top} \bSigmah_{22} (\widehat{\bbeta} - \bbeta^{*}) 
&\leq 
(\widehat{\bbeta} - \bbeta^{*})^{\top}( \bSigmah_{21}  - \bSigmah_{22}\bbeta^{*}) + \lambda \|\bbeta^{*}\|_{1} - \lambda\|\bbetah\|_{1} \\
&\leq
\|\widehat{\bbeta} - \bbeta^{*}\|_{1}  \|\bSigmah_{21}  - \bSigmah_{22}\bbeta^{*}\|_{\infty} + \lambda \|\bbeta^{*}\|_{1} - \lambda\|\bbetah\|_{1}\label{eq:s4}.
\end{flalign*}
Using condition on $\lambda$ and Assumption \ref{assumption:sparsity},
\begin{flalign*}
(\widehat{\bbeta} - \bbeta^{*})^{\top} \bSigmah_{22} (\widehat{\bbeta} - \bbeta^{*}) 
&\leq
\frac{\lambda}{2}\|\widehat{\bbeta} - \bbeta^{*}\|_{1}  + \lambda \|\bbeta^{*}\|_{1} - \lambda\|\bbetah\|_{1}\\
&=
\frac{\lambda}{2}\|\widehat{\bbeta}_{S} - \bbeta^{*}_{S}\|_{1} + 
\frac{\lambda}{2}\|\widehat{\bbeta}_{S^c} \|_{1} +
\lambda \|\bbeta^{*}\|_{1} - \lambda\|\bbetah\|_{1}\\
&=
\frac{\lambda}{2}\|\widehat{\bbeta}_{S} - \bbeta^{*}_{S}\|_{1} + 
\frac{\lambda}{2}\|\widehat{\bbeta}_{S^c} \|_{1} +
\lambda \|\bbeta^{*}_{S}\|_{1} - \lambda\|\bbetah\|_{1}\\
&=
\frac{\lambda}{2}\|\widehat{\bbeta}_{S} - \bbeta^{*}_{S}\|_{1} + 
\frac{\lambda}{2}\|\widehat{\bbeta}_{S^c} \|_{1} +
\lambda \|\bbeta^{*}_{S}\|_{1} - \lambda\|\bbetah_{S}\|_{1}- \lambda\|\bbetah_{S^c}\|_{1}.
\end{flalign*}
Using the triangle inequality, 
\begin{flalign}
(\widehat{\bbeta} - \bbeta^{*})^{\top} \bSigmah_{22} (\widehat{\bbeta} - \bbeta^{*}) 
&\leq
\frac{\lambda}{2}\|\widehat{\bbeta}_{S} - \bbeta^{*}_{S}\|_{1} + 
\frac{\lambda}{2}\|\widehat{\bbeta}_{S^c} \|_{1} +
 \lambda\|\bbetah_{S} - \bbeta_{S}^{*}\|_{1} - \lambda\|\bbetah_{S^c}\|_{1}\\
 \label{eq:bound1forQ}
 &=
 \frac{3\lambda}{2}\|\widehat{\bbeta}_{S} - \bbeta^{*}_{S}\|_{1}-\frac{\lambda}{2}\|\bbetah_{S^c}\|_{1}\\
 \label{eq:bound2forQ}
 &\leq
 \frac{3\lambda}{2}\|\widehat{\bbeta}_{S} - \bbeta^{*}_{S}\|_{1}.
\end{flalign}
As $(\widehat{\bbeta} - \bbeta^{*})^{\top} \bSigmah_{22} (\widehat{\bbeta} - \bbeta^{*})$ is non-negative and $\bbeta_{S^{c}}^{*}=\zeros$, \eqref{eq:bound1forQ} implies that $\|\bbetah_{S^c}\|_{1} \leq 3 \|\widehat{\bbeta}_{S} - \bbeta^{*}_{S}\|_{1}$, and thus, $\widehat{\bbeta} - \bbeta^{*}$ is in the cone $\Ccal(s,3)$. Since $\bSigmah_{22}$ satisfies RE($s,3$) with paramter $\gamma$, we have
\begin{flalign}\label{eq:REforQ}
    \|\widehat{\bbeta}_{S} - \bbeta^{*}_{S}\|_{2} \leq \sqrt{\gamma}\sqrt{ (\widehat{\bbeta} - \bbeta^{*})^{\top} \bSigmah_{22} (\widehat{\bbeta} - \bbeta^{*}) }.
\end{flalign}
Using $\|\widehat{\bbeta}_{S} - \bbeta^{*}_{S}\|_{1}\leq \sqrt{s}\|\widehat{\bbeta}_{S} - \bbeta^{*}_{S}\|_{2}$, \eqref{eq:bound2forQ} and \eqref{eq:REforQ} imply that
\begin{flalign}\label{eq:thm1bound1}
    (\widehat{\bbeta} - \bbeta^{*})^{\top} \bSigmah_{22} (\widehat{\bbeta} - \bbeta^{*})
    \leq
    \frac{9}{4} \gamma s \lambda^2
\end{flalign}
The bound for $\|\widehat{\bbeta}_{S} - \bbeta^{*}_{S}\|_{2}$ can be obtained as follows. Since $\widehat{\bbeta} - \bbeta^{*} \in \Ccal(S,3)$, 
\begin{flalign}
    \|\widehat{\bbeta} - \bbeta^{*}\|_{1}
    &= 
    \|\widehat{\bbeta}_{S} - \bbeta^{*}_{S}\|_{1} + 
    \|\widehat{\bbeta}_{S^c} - \bbeta^{*}_{S^c}\|_{1}\\
    &\leq 
    4\|\widehat{\bbeta}_{S} - \bbeta^{*}_{S}\|_{1}\\
    &\leq 
    4\sqrt{s}\|\widehat{\bbeta}_{S} - \bbeta^{*}_{S}\|_{2}\\
    &\leq 
    4\sqrt{s\gamma}\sqrt{ (\widehat{\bbeta} - \bbeta^{*})^{\top} \bSigmah_{22} (\widehat{\bbeta} - \bbeta^{*}) }\\
    \label{eq:boundL1ofa}
    &\leq 
    4\sqrt{s\gamma}\frac{3}{2}\sqrt{s \gamma} \lambda
    =6s\gamma\lambda.
\end{flalign}
Let $\ba = \widehat{\bbeta} - \bbeta^{*}$ for notational simplicity. 
For $j=0,1,\ldots, J$, let $T_{j}$ be the index set of $(j+1)$th $s$ largest (in absolute) components of $\ba$. 
Then, $\ba \in \Ccal(T_{0},3)$ as
\begin{flalign*}
    \|\ba_{T_{0}^{c}}\|_{1}  = \|\ba\|_{1} - \|\ba_{T_{0}}\|_{1}
    & \leq \|\ba\|_{1} - \|\ba_{S}\|_{1} \\
    & =  \| \ba_{S^c}\|_{1}\\
    & \leq 3 \|\ba_{S}\|_{1} \quad \text{since $\ba \in \Ccal(S,3)$}\\
    & \leq 3 \|\ba_{T_{0}}\|_{1}.
\end{flalign*}
Furthermore, it follows that $\|\ba_{T_j}\|_0 = s$ for $j=0, \dots, J-1$ with last $\|\ba_{T_J}\|_0 \leq s$. Also, for $j\geq 1$, $\|\ba_{T_j}\|_2\leq \sqrt{s}\|\ba_{T_j}\|_{\infty} \leq \sqrt{s}\|\ba_{T_{j-1}}\|_{1}/s$.
Thus, by the triangle inequality,
\begin{flalign}
    \|\ba\|_{2} &\leq \|\ba_{T_{0}}\|_{2} + \sum_{j=1}^{J} \|\ba_{T_{j}}\|_{2}
    \leq \|\ba_{T_{0}}\|_{2} + \sum_{j= 1}^{J} \sqrt{s}\|\ba_{T_{j}}\|_{\infty} \\
    \label{eq:a2normConeBound}
    & \leq  \|\ba_{T_{0}}\|_{2} + \sum_{j=0}^{J-1} \sqrt{s}\frac{1}{s}\|\ba_{T_{j}}\|_{1}
     \leq \|\ba_{T_{0}}\|_{2} + \frac{1}{\sqrt{s}}\|\ba \|_{1}.
\end{flalign}
Using that $\bSigmah_{22}$ satisfies RE$(s,3)$ and $\ba \in \Ccal(T_{0},3)$,
\begin{flalign*}
    \|\widehat{\bbeta} - \bbeta^{*}\|_{2} & = \|\ba\|_{2} \leq \|\ba_{T_{0}}\|_{2} + \frac{1}{\sqrt{s}}\|\ba \|_{1} \\
    &\leq \sqrt{\gamma} \sqrt{ (\widehat{\bbeta} - \bbeta^{*})^{\top} \bSigmah_{22} (\widehat{\bbeta} - \bbeta^{*}) } + 6\sqrt{s} \gamma \lambda \quad \text{as \eqref{eq:REforQ} and \eqref{eq:boundL1ofa} }\\
    &\leq \frac{3}{2} \gamma \sqrt{s} \lambda + 6\sqrt{s} \gamma \lambda \quad \text{by \eqref{eq:thm1bound1} }\\
    & = \frac{15}{2} \gamma \sqrt{s} \lambda.
\end{flalign*}

\end{proof}

\begin{theorem}\label{thm:key}
Under Assumptions \ref{assumption1}--\ref{assumption:samplesize}, for some constant $C>0$, with probability at least $1-p^{-1}$
\begin{flalign*}
    \| \bSigmah_{21}  - \bSigmah_{22}\bbeta^{*}\|_{\infty} \leq C \sqrt{ \frac{\log(p)}{n} }.
\end{flalign*}
\end{theorem}

\begin{proof} Using $\bbeta^* = \bSigma_{22}^{-1}\bSigma_{21}$ and triangle inequality, we have
\begin{flalign*}
    \| \bSigmah_{21}  - \bSigmah_{22}\bbeta^{*}\|_{\infty} 
    &= \| \bSigmah_{21}  - \bSigmah_{22}\bbeta^{*} +  \bSigma_{21} - \bSigma_{21}\|_{\infty}  \\
    &= \| \bSigmah_{21} - \bSigma_{21} + \bSigma_{21} - \bSigmah_{22}\bbeta^{*} \|_{\infty} \\
&\leq 
    \| \bSigmah_{21}  - \bSigma_{21}\|_{\infty} + 
    \| (\bSigma_{22}  - \bSigmah_{22})\bbeta^{*}\|_{\infty}.
\end{flalign*}
For $\| \bSigmah_{21}  - \bSigma_{21}\|_{\infty}$, it follows from Theorem 7 of \citet{yoon2020sparse} that, for some $C_{1}>0$,
\begin{flalign}\label{eq:sigma21bound}
    \| \bSigmah_{21}  - \bSigma_{21}\|_{\infty} \leq C_{1} \sqrt{ \frac{\log p}{n}}
\end{flalign}
with probability at least $1-p^{-1}$. 

Consider $\|(\bSigma_{22}  - \bSigmah_{22})\bbeta^{*}\|_{\infty}$.  Recall that  $\bSigma_{22} = G^{-1}(\bT, \bDelta) =\left[ G^{-1}(\tau_{jk},\Delta_{j},\Delta_{k}) \right]_{1\leq j,k \leq p} $, and  $\bSigmah_{22} = G^{-1}(\bThat, \bDeltahat)$, and let  $G^{-1}_{\tau} = \partial G^{-1}(\tau,\Delta_{j},\Delta_{k})/\partial \tau$ be the partial derivative of the inverse bridge function with respect to $\tau$. By adding and subtracting $G^{-1}(\bThat, \bDelta)$ from $G^{-1}(\bThat, \bDeltahat)$ and applying the mean value theorem to $G^{-1}(\bThat, \bDelta)$ with respect to $\bThat$, 
\bfln{
\bSigmah_{22} &= G^{-1}(\bThat, \bDeltahat) =  G^{-1}(\bThat, \bDelta) + \{G^{-1}(\bThat, \bDeltahat) - G^{-1}(\bThat, \bDelta)\} \\
&= G^{-1}(\bT, \bDelta) + G^{-1}_{\tau}(\widetilde{\bT}, \bDelta)\circ(\bThat - \bT) + \{G^{-1}(\bThat, \bDeltahat) - G^{-1}(\bThat, \bDelta)\} \\
&= \bSigma_{22} + G^{-1}_{\tau}(\widetilde{\bT}, \bDelta)\circ(\bThat - \bT) + \{G^{-1}(\bThat, \bDeltahat) - G^{-1}(\bThat, \bDelta)\},
}
where  $\widetilde{\bT}=[\widetilde{\tau}_{jk}]_{1\leq j,k \leq p}$ and $\widetilde{\tau}_{jk} \in (\widehat{\tau}_{jk},\tau_{jk})$.
Therefore
\bfln{
&(\bSigma_{22} - \bSigmah_{22})\bbeta^{*} =  -G^{-1}_{\tau}(\widetilde{\bT}, \bDelta)\circ(\bThat - \bT)\bbeta^{*} - \{G^{-1}(\bThat, \bDeltahat) - G^{-1}(\bThat, \bDelta)\}\bbeta^{*}.
}
By letting
\bfln{ 
G^{-1}_{\tau}(\widetilde{\bT},\bDelta) =  G^{-1}_{\tau}({\bT},\bDelta) +\{ G^{-1}_{\tau}(\widetilde{\bT},\bDelta) - G^{-1}_{\tau}({\bT},\bDelta) \},
}
we further have, by the triangle inequality, that
\bfln{
\|(\bSigma_{22} - \bSigmah_{22})\bbeta^{*}\|_{\infty} &\leq 
\underbrace{ \|G^{-1}_{\tau}(\bT, \bDelta)\circ(\bThat - \bT)\bbeta^{*}\|_{\infty} }_{\coloneqq I_{1}} + 
\underbrace{ \| \{ G^{-1}_{\tau}(\widetilde{\bT}, \bDelta) - G^{-1}_{\tau}(\bT, \bDelta)\}\circ(\bThat - \bT)\bbeta^{*}\|_{\infty} }_{\coloneqq I_{2}} \\
&+
\underbrace{  \|\{G^{-1}(\bThat, \bDeltahat) - G^{-1}(\bThat, \bDelta)\}\bbeta^{*} \|_{\infty}}_{\coloneqq I_{3}}.
}
We separately bound $I_{1}$, $I_{2}$, and $I_{3}$ in Lemmas \ref{lem:I1}, \ref{lem:I2}, and \ref{lem:I3}, respectively. Combining these bounds with~\eqref{eq:sigma21bound} completes the proof.
\end{proof}

\begin{theorem}
Under Assumptions \ref{assumption1}--\ref{assumption:samplesize}, if $\lambda = C\sqrt{\log(p)/n}$ for some constant $C>0$ and $\bSigma_{22}$ satisfies RE(s,3) with parameter $\gamma$,
then
\bfln{
\|\bbetah - \bbeta^{*}\|_{2}^{2} = O_{p}\left(\gamma^2 \frac{s \log(p)}{n}\right).
}
\end{theorem}
\begin{proof}
From Theorem \ref{thm:deterministic2}, if $\widehat \bSigma_{22}$ satisfies RE(s, 3) and $\lambda \geq 2\|\widehat \bSigma_{21} - \widehat\bSigma_{22}\bbeta\|_{\infty}$, then
$$
\|\bbetah - \bbeta^{*}\|_{2}^{2} \leq C_1 \gamma^2 s \lambda^2.
$$
From Theorem \ref{thm:key}, if $\lambda = C\sqrt{\log(p)/n}$, then $\lambda \geq 2\|\widehat \bSigma_{21} - \widehat\bSigma_{22}\bbeta\|_{\infty}$ holds with probability at least $1 - p^{-1}$. From Lemma~\ref{lem:RE}, if $\bSigma_{22}$ satisfies RE(s,3) with parameter $\gamma$, then with high probability so does $\widehat \bSigma_{22}$ with $\gamma(\widehat \bSigma_{22}) = C_1\gamma(\bSigma_{22})$. Combining these results gives that, with high probability,
$$
\|\bbetah - \bbeta^{*}\|_{2}^{2} \leq C_2 \gamma^2 s \frac{\log p}{n},
$$
leading to the desired bound.
\end{proof}

\subsection{Main supporting lemmas}

\begin{lemma}\label{lem:I1}
Under Assumptions \ref{assumption1}--\ref{assumption:samplesize}, for some constant $C>0$, 
\bfln{
\|G^{-1}_{\tau}(\bT,\bDelta) \circ(\widehat{\bT} - \bT)\bbeta^{*} \|_{\infty}  \leq C\sqrt{\frac{\log(p)}{n}}
}
with probability at least $1-2/p$.
\end{lemma}

\begin{proof}
 Let $\be_{j}\in\R^{p}$ be the vector with 1 in the $j$th component and 0 otherwise. Then, we have
\bfln{
\|G^{-1}_{\tau}(\bT,\bDelta) \circ(\widehat{\bT} - \bT)\bbeta^{*} \|_{\infty} &= \max_{1\leq j \leq p}|\be_{j}^{\top} G^{-1}_{\tau}(\bT,\bDelta) \circ(\widehat{\bT} - \bT)\bbeta^{*} | \\
&= \max_{1\leq j \leq p}
\|\bm_{j}\|_{2} |\bu^{\top} (\widehat{\bT} - \bT)\be_{j}|.
}
where $\bm_{j} = (G^{-1}_{\tau}(\tau_{1j},\Delta_{1},\Delta_{j}) \beta^{*}_{1j},\ldots,G^{-1}_{\tau}(\tau_{pj},\Delta_{p},\Delta_{j}) \beta^{*}_{pj})^{\top}$ and $\bu=\bm_{j}/\|\bm_{j}\|_{2}$ is a deterministic unit vector. 
Since $|G^{-1}_{\tau}|\leq C_{1}$ for some constant $C_{1}>0$ by Theorem 6 of \citet{yoon2020sparse}, we have $\|\bm_{j}\|_{2} \leq C_{1}\|\bbeta^{*}\|_{2} \leq C_{1}C_{\cov}$ by Lemma \ref{lem:beta_norm_bounds}. Therefore, 
\bfln{
\|G^{-1}_{\tau}(\bT,\bDelta) \circ(\widehat{\bT} - \bT)\bbeta^{*} \|_{\infty} \leq C_{1}C_{\cov} \max_{1\leq j \leq p} |\bu^{\top} (\widehat{\bT} - \bT)\be_{j}|.
}
Consider $\bu^{\top} (\widehat{\bT} - \bT)\be_{j}$. By Lemma \ref{lem:u(Th-T)v_subgaussian}, for any $\epsilon>0$ and $0<t\leq n/C_{\cov}$,
\bfln{
\PP\left( \max_{1\leq j \leq p} |\bu^{\top} (\widehat{\bT} - \bT)\be_{j}| \geq \epsilon \right) 
&\leq
p \cdot \PP\left(  |\bu^{\top} (\widehat{\bT} - \bT)\be_{1}| \geq \epsilon \right) \\
&\leq
2p \cdot \PP\left\{  \exp\left( t \cdot \bu^{\top} (\widehat{\bT} - \bT)\be_{1} \right) \geq \exp(t\epsilon) \right\}
\\
&\leq
2p \cdot
\E\left[ \exp \left( t \cdot \bu^{\top} (\widehat{\bT} - \bT)\be_1 \right) \right] \exp(-t\epsilon) \\
&\leq
2p \cdot
\exp\left(\frac{t^2C_{\cov}^2}{n} - t\epsilon \right)\\
&=2\exp\left(
\log(p) + \frac{t^2C_{\cov}^2}{n} - t\epsilon
\right)
.
}
Letting $\epsilon=3C_{\cov}\sqrt{\log(p)/n}$ and $t=\sqrt{n\log(p)}/C_{\cov}$, we have
\bfln{
\PP\left( \max_{1\leq j \leq p} |\bu^{\top} (\widehat{\bT} - \bT)\be_{j}| \leq 3C_{\cov}\sqrt{\frac{\log(p)}{n} }  \right) 
\leq 1 - \frac{2}{p}.
}
Thus, for some constant $C>0$, $\|G^{-1}_{\tau}(\bT,\bDelta) \circ(\widehat{\bT} - \bT)\bbeta^{*} \|_{\infty} \leq C \sqrt{\log(p)/n}$ with probability at least $1-2p^{-1}$.
\end{proof}

\begin{lemma}\label{lem:I2}
Under Assumptions \ref{assumption1}--\ref{assumption:samplesize}, for some constant $C>0$, 
\bfln{
\left\| \left\{ G^{-1}_{\tau}(\widetilde{\bT}, \bDelta) - G^{-1}_{\tau}(\bT, \bDelta) \right\}\circ(\bThat - \bT)\bbeta^{*} \right\|_{\infty}  \leq C\sqrt{\frac{\log(p)}{n}}
}
with probability at least $1-p^{-1}$.
\end{lemma}

\begin{proof}
By the mean value theorem, for some $\bar{\bT}=[\bar{\tau}_{jk}]_{1 \leq j,k \leq p}$, where $\bar{\tau}_{jk} \in (\widetilde{\tau}_{jk},\tau_{jk})$, we have
\bfln{
\| \{ G^{-1}_{\tau}(\widetilde{\bT}, \bDelta) - G^{-1}_{\tau}(\bT, \bDelta)\}\circ(\bThat - \bT)\bbeta^{*}\|_{\infty} 
=
\| G^{-1}_{\tau\tau}(\bar{\bT}, \bDelta) \circ(\widetilde{\bT} - \bT) \circ(\bThat - \bT)\bbeta^{*}\|_{\infty},
}
where $G^{-1}_{\tau\tau}$ is the 2nd partial derivative of inverse bridge function with respect to $\tau$. By Lemma~\ref{lem:Ginv_second_tau_tau},  $|G^{-1}_{\tau\tau}| \leq C_{1}$ for some constant $C_{1}>0$.  Since $|\bar{\tau}_{jk} - \tau_{jk}|\leq |\tilde{\tau}_{jk} - \tau_{jk}|\leq |\widehat{\tau}_{jk} - \tau_{jk}|$, by H\"older's inequality,
\bfln{
\| \{ G^{-1}_{\tau}(\widetilde{\bT}, \bDelta) - G^{-1}_{\tau}(\bT, \bDelta)\}\circ(\bThat - \bT)\bbeta^{*}\|_{\infty} 
&\leq
\| G^{-1}_{\tau\tau}(\bar{\bT}, \bDelta) \circ(\widetilde{\bT} - \bT) \circ(\bThat - \bT)\|_{\infty} \|\bbeta^{*}\|_{1} \\
&\leq 
C_{1} \|\widetilde{\bT} - \bT \|_{\infty} \|\widehat{\bT} - \bT \|_{\infty}  \| \bbeta^{*} \|_{1} \\
&\leq
C_{1} \|\widehat{\bT} - \bT \|_{\infty}^{2}  \| \bbeta^{*} \|_{1}.
}
By Lemma \ref{lem:Tau_infinity_norm_bound}, $\|\widehat{\bT} - \bT \|_{\infty}^{2}\leq C_{2}\log(p)/n$ with probability at least $1-p^{-1}$, and by Lemma \ref{lem:beta_norm_bounds}, $\|\bbeta^{*}\|_{1} \leq \sqrt{s}C_{\cov}$. Thus, under Assumption \ref{assumption:samplesize}, for sufficiently large $n$ and some constant $C>0$,
\bfln{
\|I_{2}\|_{\infty} 
\leq C \frac{\sqrt{s} \log( p )}{n} 
\leq C \sqrt{\frac{\log(p)}{n}}.
}
with probability at least $1-p^{-1}$.
\end{proof}

\begin{lemma}\label{lem:I3}
Under Assumptions \ref{assumption1}--\ref{assumption:samplesize}, for some constant $C>0$, 
\bfln{
\left\| \left\{G^{-1}(\bThat, \bDeltahat) - G^{-1}(\bThat, \bDelta)\right\}\bbeta^{*} \right\|_{\infty}  \leq C\sqrt{\frac{\log(p)}{n}}
}
with probability at least $1-p^{-1}$.
\end{lemma}

\begin{proof}
We reparameterize $\Delta_{j}=\Phi^{-1}(\pi_{j})$ and write $G^{-1}(\bT,\bDelta) = G^{-1}(\bT,\bpi)$, where $\bpi=(\pi_{1},\ldots,\pi_{p})^{\top}$. We also write $G^{-1}_{\pi_{1}} = \partial G^{-1}(\tau,\pi_{1},\pi_{2})/\partial \pi_{1}$ and $G^{-1}_{\pi_{2}} = \partial G^{-1}(\tau,\pi_{1},\pi_{2})/\partial \pi_{2}$. For each element of $G^{-1}(\bThat,\bpihat) - G^{-1}(\bThat,\bpi)$, the multivariate mean value theorem gives
\bfln{
G^{-1}(\widehat{\tau}_{jk},\widehat{\pi}_{j},\widehat{\pi}_{k}) - G^{-1}(\widehat{\tau}_{jk},\pi_{j},\pi_{k})
=
\underbrace{ G^{-1}_{\pi_{1}}(\widehat{\tau}_{jk},\widetilde{\pi}_{j},\widetilde{\pi}_{k})(\pihat_{j} - \pi_{j})}_{\coloneqq I_{3,1}}
+
\underbrace{ G^{-1}_{\pi_{2}}(\widehat{\tau}_{jk},\widetilde{\pi}_{j},\widetilde{\pi}_{k}) (\pihat_{k}-\pi_{k}) }_{\coloneqq I_{3,2}},
}
for some $\widetilde{\pi}_{j} \in (\widehat{\pi}_{j}, \pi_{j})$ and $\widetilde{\pi}_{k} \in (\widehat{\pi}_{k}, \pi_{k})$, respectively.

Consider $I_{3,1}$. By adding and subtracting $G^{-1}_{\pi_{1}}({\tau}_{jk},{\pi}_{j},{\pi}_{k})(\pihat_{j} - \pi_{j})$ to $I_{3,1}$, we have that
\bfln{
G^{-1}_{\pi_{1}}(\widehat{\tau}_{jk},\widetilde{\pi}_{j},\widetilde{\pi}_{k})(\pihat_{j} - \pi_{j}) 
&=
\left\{
G^{-1}_{\pi_{1}}(\widehat{\tau}_{jk},\widetilde{\pi}_{j},\widetilde{\pi}_{k})
-G^{-1}_{\pi_{1}}({\tau}_{jk},{\pi}_{j},{\pi}_{k})
+G^{-1}_{\pi_{1}}({\tau}_{jk},{\pi}_{j},{\pi}_{k})
\right\}
(\pihat_{j} - \pi_{j}) \\
&=
\left\{
G^{-1}_{\pi_{1}}(\widehat{\tau}_{jk},\widetilde{\pi}_{j},\widetilde{\pi}_{k})
-G^{-1}_{\pi_{1}}({\tau}_{jk},{\pi}_{j},{\pi}_{k})
\right\}(\pihat_{j} - \pi_{j})
+G^{-1}_{\pi_{1}}({\tau}_{jk},{\pi}_{j},{\pi}_{k})
(\pihat_{j} - \pi_{j}).
}
By applying the multivariate mean value theorem to $\{
G^{-1}_{\pi_{1}}(\widehat{\tau}_{jk},\widetilde{\pi}_{j},\widetilde{\pi}_{k})
-G^{-1}_{\pi_{1}}({\tau}_{jk},{\pi}_{j},{\pi}_{k})\}$, we further have that
\bfln{
G^{-1}_{\pi_{1}}(\widehat{\tau}_{jk},\widetilde{\pi}_{j},\widetilde{\pi}_{k})(\pihat_{j} - \pi_{j}) 
=&
G^{-1}_{\pi_{1}\tau}(\bar{\tau}_{jk},\bar{\pi}_{j},\bar{\pi}_{k}) (\widehat{\tau} - \tau)(\widehat{\pi}_{j} - \pi_{j}) +
G^{-1}_{\pi_{1}\pi_{1}}(\bar{\tau}_{jk},\bar{\pi}_{j},\bar{\pi}_{k}) (\widetilde{\pi}_{j} - \pi_{j})(\widehat{\pi}_{j} - \pi_{j}) \\
&+
G^{-1}_{\pi_{1}\pi_{2}}(\bar{\tau}_{jk},\bar{\pi}_{j},\bar{\pi}_{k}) (\widetilde{\pi}_{k} - \pi_{k})(\widehat{\pi}_{j} - \pi_{j}) +
G^{-1}_{\pi_{1}}({\tau}_{jk},{\pi}_{j},{\pi}_{k})(\pihat_{j} - \pi_{j})
}
for some $\bar{\tau}_{jk} \in (\widehat{\tau}_{jk},\tau_{jk})$, $\bar{\pi}_{j} \in (\widetilde{\pi}_{j}, \pi_{j})$, and $\bar{\pi}_{k} \in (\widetilde{\pi}_{k}, \pi_{k})$.
Thus, by the triangle inequality,
\bfl{
\label{inlem:I31_1}
\| G^{-1}_{\pi_{1}}(\bThat,\widetilde{\bpi})\circ(\bPihat - \bPi) \bbeta^{*} \|_{\infty}
\leq &
\|G^{-1}_{\pi_{1}\tau}(\bar{\bT},\bar{\bpi}) \circ (\bThat - \bT) \circ (\bPihat - \bPi) \bbeta^{*} \|_{\infty} \\
\label{inlem:I31_2}
&+
\|G^{-1}_{\pi_{1}\pi_{1}}(\bar{\bT},\bar{\bpi})  \circ (\widetilde{\bPi} - \bPi)\circ (\widehat{\bPi} - \bPi) \bbeta^{*}\|_{\infty} \\
\label{inlem:I31_3}
&+
\|G^{-1}_{\pi_{1}\pi_{2}}(\bar{\bT},\bar{\bpi}) \circ (\bPihat^{\top} - \bPi^{\top})\circ (\widehat{\bPi} - \bPi)\bbeta^{*} \|_{\infty} \\
\label{inlem:I31_4}
&+
\|G^{-1}_{\pi_{1}}(\bT,\bpi)\circ(\bPihat - \bPi)\bbeta^{*} \|_{\infty},
}
where $\bPi = \bpi \ones_{p}^{\top}$ and $\bPihat = \bpihat \ones_{p}^{\top}$. We consider \eqref{inlem:I31_1}--\eqref{inlem:I31_3} and \eqref{inlem:I31_4} separately as follows.

For \eqref{inlem:I31_1}--\eqref{inlem:I31_3}, we know that $|G^{-1}_{\pi_{1}\tau}|$, $|G^{-1}_{\pi_{1}\pi_{1}}|$, and $|G^{-1}_{\pi_{1}\pi_{2}}|$ are bounded above by some positive constants by Lemmas \ref{lem:Ginv_second_tau_pij}, \ref{lem:Ginv_second_pij_pij},  \ref{lem:Ginv_second_pik_pij}, respectively. Also, for some constant $C_{1},C_{2}>0$, $\|\bpihat - \bpi\|_{\infty} \leq C_{1} \sqrt{\log(p)/n}$ and $\|\bThat - \bT\|_{\infty} \leq C_{2} \sqrt{\log(p)/n}$ with probability at least $1-p^{-1}$ by Lemmas \ref{lem:pi_subgaussian} and \ref{lem:Tau_infinity_norm_bound}. Since $|\widetilde{\pi}_{j}-\pi_{j}| \leq |\widehat{\pi}_{j}-\pi_{j}|$, $\|\widetilde{\bPi} - \bPi\|_{\infty} \leq \|\bPihat - \bPi\|_{\infty}$ and
\bfln{
\|\bPihat - \bPi\|_{\infty} = \|\bPihat^{\top} - \bPi^{\top}\|_{\infty} = \|\bpihat - \bpi\|_{\infty}.
}
Hence, we can show that, for some constant $C>0$, \eqref{inlem:I31_1}--\eqref{inlem:I31_3} are bounded above by $C\sqrt{\log(p)/n}$ with probability at least $1-p^{-1}$ by following the steps of Lemma \ref{lem:I1}. 

For \eqref{inlem:I31_4}, we know that $|G^{-1}_{\pi_{1}}|$ is bounded above by some positive constant by Lemma \ref{lem:Ginv_first_pi}. Thus, by following the steps of Lemma \ref{lem:I2}, we can show that, for some constant $C>0$, \eqref{inlem:I31_4} is bounded above by $C\sqrt{\log(p)/n}$ with probability at least $1-p^{-1}$. By combining these results, we have that, for some constant $C>0$,
\bfl{\label{inlem:I3_first}
\| G^{-1}_{\pi_{1}}(\bThat,\widetilde{\bpi})\circ(\bPihat - \bPi) \bbeta^{*} \|_{\infty}
\leq C \sqrt{\frac{\log(p)}{n}}
}
with probability at least $1-p^{-1}$.

By symmetrically applying above steps to $I_{3,2}$, we also have that, for some constant $C>0$,
\bfl{\label{inlem:I3_second}
\| G^{-1}_{\pi_{2}}(\bThat,\widetilde{\bpi})\circ(\bPihat - \bPi)^{\top} \bbeta^{*} \|_{\infty}
\leq C \sqrt{\frac{\log(p)}{n}}
}
with probability at least $1-p^{-1}$. Combining \eqref{inlem:I3_first} and \eqref{inlem:I3_second} completes the proof.

\end{proof}

\begin{lemma}\label{lem:beta_norm_bounds}
Let $\bbeta^{*}=\bSigma_{22}^{-1}\bSigma_{21}$. Under Assumptions \ref{assumption:condition_number}--\ref{assumption:sparsity}
\bfln{
\|\bbeta^{*}\|_{2} <C_{\cov} \quad \text{and} \quad \|\bbeta^{*}\|_{1} < \sqrt{s} C_{\cov}.
}
\end{lemma}
\begin{proof}
Under Assumption \ref{assumption:condition_number}
\bfln{
\|\bbeta^{*}\|_{2} = \|\bSigma_{22}^{-1}\bSigma_{21}\|_{2}\leq  \|\bSigma_{22}^{-1}\|_{\textup{op}}\|\bSigma_{21}\|_{2}.
}
At the same time, using  $\be_{1}=(1,0,\ldots,0)^{\top}$, 
$$
\lambda_{\max}(\bSigma)\geq1\geq \|\bSigma \be_{1}\|_{2} =\sqrt{1+\|\bSigma_{21}\|_{2}^{2}} > \|\bSigma_{21}\|_2.
$$
Since $\|\bSigma_{22}^{-1}\|_{\textup{op}}= \{\lambda_{\min}(\bSigma)\}^{-1}$, we have that
\bfln{
\|\bbeta^{*}\|_{2} &\leq \|\bSigma_{22}^{-1}\|_{\textup{op}}\|\bSigma_{21}\|_{2}< \frac{\lambda_{\max}(\bSigma)}{\lambda_{\min}(\bSigma)} 
\leq C_{\cov}.
}
Under Assumption \ref{assumption:sparsity}, it follows
\bfln{
\|\bbeta^{*}\|_{1} &\leq  \sqrt{s} \|\bbeta^{*}\|_{2} <  \sqrt{s}C_{\cov}.
}
\end{proof}

\begin{lemma}\label{lem:pi_subgaussian}
For $j=1,\ldots,p$, let $\pi_{j} = \Phi(\Delta_{j})$ and $\pihat_{j}=n^{-1}\sum_{i=1}^{n}1(X_{ij}=0)$, where $\E( \pihat_{j} ) = \pi_{j}$. Also, let $\bpi=(\bpi_{1},\ldots,\bpi_{p})^{\top}$ and $\widehat{\bpi}=(\widehat{\bpi}_{1},\ldots,\widehat{\bpi}_{p})^{\top}$. Then for any deterministic $\|\bu\|_{2}=1$ and $t$,
\bfln{
\E \left\{ \exp(t \bu^{\top} (\bpihat - \bpi)) \right\} \leq \exp\left( \frac{t^2C}{n} \right)
}
for some constant $C>0$.
\end{lemma}
\begin{proof}
For $i=1,\ldots,n$, let $\bb_{i} = (1(X_{i1}=0),\ldots,1(X_{ip}=0))^{\top}$ such that $n^{-1}\sum_{i=1}^{n}\bb_{i}= \bpihat$.
By definition of the truncated latent Gaussian copula model, we have
\bfln{
b_{ij} &= 1(X_{ij}=0) =   1(X^{*}_{ij}\leq D_{j})=   1(Z_{ij}\leq\Delta_{j}) \\
&=  1(Z_{ij} - \Delta_{j}\leq 0 ) =   \frac{\sign(Z_{ij}-\Delta_{j})+1 }{2}.
}
Since $\widetilde\bz_{i} = \bz_{i} - \bDelta \overset{iid}{\sim} \N_{p}(-\bDelta,\bSigma_{22})$, where $\bz_{i}=(Z_{i1},\ldots,Z_{ip})^{\top}$ and $\bDelta=(\Delta_{1},\ldots,\Delta_{p})^{\top}$, $\sign(\widetilde{\bz}_{i}) - \E\{\sign(\widetilde{\bz}_{i})\} $ is $\mathsf{C}(\bSigma_{22})$-subgaussian
by Lemma \ref{lem:sign_subgaussian}, and thus $\bpihat - \bpi = n^{-1}\sum_{i=1}^{n}\{\bb_{i}-\E(\bb_{i})\}$
is sum of $n$ iid $\mathsf{C}(\bSigma_{22})$-subgaussians.
Thus,
\bfln{
\E \left\{ \exp(t \bu^{\top} (\bpihat - \bpi)) \right\} &= \E \left[ \exp\left\{\frac{t}{n} \sum_{i=1}^{n}\bu^{\top} (\bb_{i} - \E(\bb_{i}) )\right\} \right]\\
&\leq
\prod_{i=1}^{n}\exp\left( \frac{t^2C}{n^{2}} \right) = \exp\left( \frac{t^2C}{n} \right).
}
\end{proof}

\begin{lemma}\label{lem:RE} Let $\bSigma_{22}$ satisfy $RE(s, 3)$ with parameter $\gamma= \gamma(\bSigma_{22})$. Let Assumptions~\ref{assumption1}, \ref{assumption2} and $\ref{assumption:samplesize}$ hold. Then with probability $1 - O(p^{-1})$, $\widehat \bSigma_{22}$ satisfies $RE(s, 3)$ with 
$$
\widehat \gamma = \gamma(\widehat \bSigma_{22}) \leq C \gamma
$$
for some constant $C>1$.
\end{lemma}
\begin{proof} Let $\ba \in \mathcal{C}(S,3)=\{\ba \in \R^{p}: \|\ba_{S^{c}}\|_{1} \leq 3 \|\ba_{S}\|_{1} \}$. Let $T_0$ be the index set of the $s$ largest elements of $\ba$ in absolute value. Then it holds that $\ba \in \mathcal{C}(T_0, 3)$, and
\begin{align}\label{eq:a1normbound}
    \|\ba\|_1 = \|\ba_S\|_1 + \|\ba_S^c\|_1\leq 4\|\ba_S\|_1\leq4\sqrt{s}\|\ba_S\|_2\leq 4\sqrt{s}\|\ba_{T_0}\|_2.
\end{align}
Furthermore, following \eqref{eq:a2normConeBound}, 
it holds that
\begin{align}\label{eq:a2normbound}
    \|\ba\|_2 \leq \|\ba_{T_0}\|_2 + \|\ba\|_1/\sqrt{s} \leq 5 \|\ba_{T_0}\|_2.
\end{align}
Consider
\begin{align}\label{eq:aSigma}
\ba^{\top}\widehat \bSigma_{22}\ba = \ba^{\top} \bSigma_{22}\ba + \ba^{\top} (\widehat \bSigma_{22} - \bSigma_{22})\ba \geq  \ba^{\top} \bSigma_{22}\ba - |\ba^{\top} (\widehat \bSigma_{22} - \bSigma_{22})\ba|.
\end{align}
Following the proof of Theorem \ref{thm:key} and reparameterization of $\bDelta$ in terms of $\bPi = \bpi \ones_{p}^{\top}$, we have the following decomposition
\begin{equation}\label{eq:SigmaDecomp}
\begin{split}
\widehat \bSigma_{22} - \bSigma_{22} &=  G^{-1}_{\tau}({\bT}, \bDelta)\circ(\bThat - \bT)+ G^{-1}_{\tau\tau}(\bar{\bT}, \bDelta)\circ(\widetilde{\bT} - \bT)\circ(\bThat - \bT) \\
&\quad + \{G^{-1}(\bThat, \bDeltahat) - G^{-1}(\bThat, \bDelta)\}\\
&= G^{-1}_{\tau}({\bT}, \bDelta)\circ(\bThat - \bT)+ G^{-1}_{\tau\tau}(\bar{\bT}, \bDelta)\circ(\widetilde{\bT} - \bT)\circ(\bThat - \bT) \\
&\quad +\{G^{-1}(\bThat, \bPihat) - G^{-1}(\bThat, \bPi)\}\\
&= G^{-1}_{\tau}({\bT}, \bDelta)\circ(\bThat - \bT)+ G^{-1}_{\tau\tau}(\bar{\bT}, \bDelta)\circ(\widetilde{\bT} - \bT)\circ(\bThat - \bT) \\
&\quad +G^{-1}_{\pi_{1}}(\bThat,\widetilde{\bpi})\circ(\bPihat - \bPi) + G^{-1}_{\pi_{2}}(\bThat,\widetilde{\bpi})\circ(\bPihat - \bPi)^{\top}\\
& =  G^{-1}_{\tau}({\bT}, \bDelta)\circ(\bThat - \bT) + G^{-1}_{\pi_{1}}(\bT,\bpi)\circ(\bPihat - \bPi) + G^{-1}_{\pi_{2}}(\bT,\bpi)\circ(\bPihat - \bPi)^{\top}\\
&\quad+ G^{-1}_{\tau\tau}(\bar{\bT}, \bDelta)\circ(\widetilde{\bT} - \bT)\circ(\bThat - \bT) \\
&\quad + G^{-1}_{\pi_{1}\tau}(\bar{\bT},\bar{\bpi}) \circ (\bThat - \bT) \circ (\bPihat - \bPi) +  G^{-1}_{\pi_{2}\tau}(\bar{\bT},\bar{\bpi}) \circ (\bThat - \bT) \circ (\bPihat - \bPi)^{\top}\\
&\quad + G^{-1}_{\pi_{1}\pi_{1}}(\bar{\bT},\bar{\bpi})  \circ (\widetilde{\bPi} - \bPi)\circ (\widehat{\bPi} - \bPi) + G^{-1}_{\pi_{2}\pi_{2}}(\bar{\bT},\bar{\bpi})  \circ (\widetilde{\bPi} - \bPi)^{\top}\circ (\widehat{\bPi} - \bPi)^{\top}\\
&\quad + G^{-1}_{\pi_{1}\pi_{2}}(\bar{\bT},\bar{\bpi}) \circ (\bPihat - \bPi)^{\top}\circ (\widehat{\bPi} - \bPi) + G^{-1}_{\pi_{2}\pi_{1}}(\bar{\bT},\bar{\bpi}) \circ (\bPihat - \bPi)\circ (\widehat{\bPi} - \bPi)^{\top},
\end{split}
\end{equation}
where $\bar{\tau}_{jk} \in (\widetilde \tau_{jk},\tau_{jk})$, $\widetilde \tau_{jk} \in (\widehat \tau_{jk},\tau_{jk})$, $\bar \pi_j \in (\widetilde \pi_j, \pi_j)$, and $\widetilde \pi_j \in (\widehat \pi_j, \pi_j)$. We will use one technique to bound all first order terms in~\eqref{eq:SigmaDecomp}, and another technique to bound all second-order terms.

Consider second-order terms in~\eqref{eq:SigmaDecomp}. Each term is bounded in the same way using H\"older's inequality and bounds on second derivatives. Concretely, consider the term corresponding to $G^{-1}_{\tau \tau}$, that is
\begin{align*}
     |\ba^{\top}G^{-1}_{\tau\tau}(\bar\bT, \Delta)\circ(\widetilde{\bT} - \bT)\circ(\bThat - \bT)\ba| \leq \|\ba\|_1^2 \|G^{-1}_{\tau\tau}(\bar\bT, \Delta)\circ(\widetilde{\bT} - \bT)\circ(\bThat - \bT)\|_{\infty}.
\end{align*}
By Lemma~\ref{lem:Ginv_second_tau_tau}, the 2nd derivative is bounded $|G^{-1}_{\tau \tau}|\leq C$, thus, since $\widetilde \tau_{jk}$ is between $\widehat \tau_{jk}$ and $\tau_{jk}$,
\begin{align*}
    |\ba^{\top}G^{-1}_{\tau\tau}(\bar\bT, \bDelta)\circ(\widetilde{\bT} - \bT)\circ(\bThat - \bT)\ba| \leq  C \|\ba\|_1^2\|\widetilde{\bT} - \bT\|_{\infty}\|\bThat - \bT\|_{\infty}\leq C \|\ba\|_1^2\|\bThat - \bT\|_{\infty}^2.
\end{align*}
Using the bound~\eqref{eq:a1normbound} on $\|\ba\|_1$, the condition that $\bSigma_{22}$ satisfies RE$(s,3)$, and Lemma~\ref{lem:Tau_infinity_norm_bound}, it follows that with high probability
\begin{align*}
     |\ba^{\top}G^{-1}_{\tau\tau}(\bar\bT, \bDelta)\circ(\widetilde{\bT} - \bT)\circ(\bThat - \bT)\ba|\leq C_1\|\ba_S\|_2^2 \frac{s\log p}{n}\leq  \ba^{\top}\bSigma_{22}\ba\ C_1\gamma \frac{s\log p}{n}.
\end{align*}
All the remaining 2nd order terms have the same bound as all the second derivatives are bounded, that is $|G^{-1}_{\pi_j\pi_j}| \leq C$ by Lemma~\ref{lem:Ginv_second_pij_pij}, $|G^{-1}_{\pi_j\tau}| \leq C$ by Lemma~\ref{lem:Ginv_second_tau_pij}, $|G^{-1}_{\pi_j\pi_k}| \leq C$ by Lemma~\ref{lem:Ginv_second_pik_pij}. Also $\|\widehat \bpi - \bpi\|_{\infty}\leq C_1\sqrt{\log p/n}$ with high probability by Hoeffding's inequality combined with union bound, and $\|\bThat - \bT\|_{\infty}\leq C_2\sqrt{\log p/n}$ with high probability by Lemma~\ref{lem:Tau_infinity_norm_bound}.

Consider first-order terms in~\eqref{eq:SigmaDecomp}. Each term is bounded in the same way using sub-gaussian properties in Lemma~\ref{lem:pi_subgaussian} (for $\widehat \bpi$) and Lemma~\ref{lem:u(Th-T)v_subgaussian} (for $\widehat \bT$) combined with the fact that the first derivatives are both bounded and fixed. Concretely, consider the term corresponding to $G^{-1}_{\tau}$, that is
\begin{align*}
\left|\ba^{\top}G^{-1}_{\tau}(\bT, \bDelta)\circ(\bThat - \bT)\ba\right| &= \bigg|(\sum_{j=1}^p a_j \be_j)^{\top}G^{-1}_{\tau}(\bT, \bDelta)\circ(\bThat - \bT)\ba\bigg|\\
&=\bigg|\sum_{j=1}^p a_j \left\{\be_j^{\top}G^{-1}_{\tau}(\bT, \bDelta)\circ(\bThat - \bT)\ba\right\}\bigg|\\
&\leq \|\ba\|_1\max_{1\leq j \leq p} \left|\be_j^{\top}G^{-1}_{\tau}(\bT, \bDelta)\circ(\bThat - \bT)\ba\right|\\
&\leq 4\sqrt{s}\|\ba_{T_0}\|_2\max_{1\leq j \leq p}|\be_j^{\top}(\bThat - \bT)\bb_j|,
\end{align*}
where $\bb_j = \ba \circ G_{\tau}^{-1}(\bT_j, \bDelta)$ and the last inequality follow from~\eqref{eq:a1normbound}. From Theorem 6 of \citet{yoon2020sparse}, $|G_{\tau}^{-1}|\leq C$ for some constant $C>0$, hence using~\eqref{eq:a2normbound}
$$
\|\bb_j\|_2 \leq C \|\ba\|_2 \leq C_1\|\ba_{T_0}\|_2.
$$
Combining this bound with Lemma~\ref{lem:Tau_infinity_norm_bound}, and following the proof of Lemma~\ref{lem:I1} gives, with high probability,
$$
\max_j|\be_j^{\top}(\bThat - \bT)\bb_j| \leq C_2\|\ba_{T_0}\|_2\sqrt{\log(p)/n},
$$
and using that $\bSigma_{22}$ satisfies RE$(s,3)$ gives
$$
|\ba^{\top}G^{-1}_{\tau}(\bT, \bDelta)\circ(\bThat - \bT)\ba|\leq C \|\ba_{T_0}\|_2^2\sqrt{\frac{s\log(p)}{n}}\leq \ba^{\top}\bSigma_{22}\ba\ C\gamma \sqrt{\frac{s\log p}{n}}.
$$
All the remaining first order terms have the same bound as the first derivatives $G^{-1}_{\pi_j}$ are fixed and bounded by Lemma~\ref{lem:Ginv_first_pi}, and $\widehat \bpi$ satisfy Lemma~\ref{lem:pi_subgaussian}.

Combining the bounds on the first and second-order terms coupled with Assumption~\ref{assumption:samplesize} gives with high probability
$$
\left|\ba^{\top} (\widehat \bSigma_{22} - \bSigma_{22})\ba\right| \leq \ba^{\top}\bSigma_{22}\ba C_3 \gamma \sqrt{\frac{s\log p}{n}}.
$$
Combining this bound with~\eqref{eq:aSigma} gives
$$
\ba^{\top}\widehat \bSigma_{22}\ba \geq \ba^{\top}\bSigma_{22}\ba\left(1- C_3 \gamma \sqrt{\frac{s\log p}{n}}\right).
$$
Under the scaling of Assumption~\ref{assumption:samplesize}, $C_3 \gamma \sqrt{\frac{s\log p}{n}} = o(1)$, thus it follows that $\gamma(\widehat \bSigma) \leq C\gamma$ for some constant $C>1$.


\end{proof}

\subsection{Supporting lemmas based on existing results}

\begin{lemma}[\citet{barber2018rocket} Lemma 4.5]\label{lem:sign_subgaussian}
Let $\bz \sim \N_{p}(\bmu,\bSigma)$. Then $\sign(\bz) - \E\{\sign(\bz)\}$ is $\mathsf{C}(\bSigma)$-subgaussian.
\end{lemma}

\begin{lemma}[\citet{barber2018rocket} Lemma E.2]\label{lem:u(Th-T)v_subgaussian}
For fixed $\bu$ and $\bv$ with $\|\bu\|_{2}, \|\bv\|_{2}\leq1$, for any $|t|\leq n/C$,
\bfln{
\E\left[
\exp
\left\{ 
t\bu^{\top}\left(\bThat -  \bT\right)\bv
\right\}
\right]
\leq\exp
\left(
\frac{t^2C^2}{n}
\right).
}
\end{lemma}


\begin{lemma}[\citet{de2012decoupling} Theorem 4.1.8]\label{lem:Tau_infinity_norm_bound}
For any $\delta>0$, 
\bfln{
\|\bThat - \bT\|_{\infty} \leq \sqrt{ \frac{4\log\left( 2{p \choose 2}/\delta \right)}{n} }
}
with probability at least $1-\delta$. With $\delta=1/2p$, $\|\bThat-\bT\|_{\infty}^{2} \leq \frac{16\log(p)}{n}$ with probability at least $1-1/2p$.
\end{lemma}


\subsection{Bounds on partial derivatives of the inverse bridge function}\label{suppl:subsec:bridge_inv_bound}

\begin{lemma}\label{lem:Ginv_first_pi}
Let $G^{-1}(\tau)$ be the inverse bridge function for TT case, where $\tau = G_{TT}(r;\Delta_{j}, \Delta_{k})$.
Under Assumptions \ref{assumption1}--\ref{assumption2}, $|\partial G^{-1}(\tau)/\partial \pi_{j}|\leq C$ and $|\partial G^{-1}(\tau)/\partial \pi_{k}|\leq C$ for some constant $C>0$.
\end{lemma}
\begin{proof}
By the multivariate chain rule, we have
\bfl{ \label{eq:multichain1}
\frac{\partial G^{-1}(\tau)}{\partial \pi_{j} } = 
\frac{\partial G^{-1}(\tau)}{\partial \tau }
\frac{\partial \tau }{\partial \Delta_{j} }
\frac{\partial \Delta_{j} }{\partial \pi_{j} }&= 
\frac{\partial G^{-1}(\tau)}{\partial \tau }
\frac{\partial G(r;\Delta_{j},\Delta_{k}) }{\partial \Delta_{j} } 
\frac{\partial \Delta_{j}}{\partial \pi_{j}} \\ \nonumber
&\coloneqq A_{1}A_{2}A_{3}.
}
By Theorem 6 in \citet{yoon2020sparse}, $|A_{1}|\leq C$.
By Lemma \ref{lem:G_first_Deltaj}, $A_2$ is bounded. By Lemma  \ref{lem:Delta_wrt_pi_bounded}, $A_3$ is bounded.
The proof for $\pi_{k}$ is analogous.
\end{proof}

\begin{lemma}\label{lem:Ginv_second_tau_tau}
Let $G^{-1}(\tau)$ be the inverse bridge function for TT case, where $\tau =  G_{TT}(r;\Delta_{j}, \Delta_{k})$.
Under Assumptions \ref{assumption1}-- \ref{assumption2}, $|G^{-1}_{\tau\tau} = \partial^2 G^{-1}(\tau)/\partial \tau^{2}|\leq C$ for some constant $C>0$ independent of $r$, $\Delta_j$, $\Delta_k$.
\end{lemma}
\begin{proof} Let $h(r) = \partial G(r;\Delta_{j},\Delta_{k})/\partial r$ and consider
\bfln{
\frac{\partial^{2} G^{-1}(\tau)}{\partial \tau^{2}}
&=
\frac{\partial}{\partial \tau}
\left\{
\frac{\partial G(r;\Delta_{j},\Delta_{k})}{\partial r}
\right\}^{-1}
=
\frac{\partial}{\partial r}
\left\{
\frac{\partial G(r;\Delta_{j},\Delta_{k})}{\partial r}
\right\}^{-1}\frac{\partial r }{\partial \tau} \\
&=
\frac{\partial}{\partial r}
\left\{\frac{1}{h(r)}\right\} \left( \frac{\partial \tau }{\partial r } \right)^{-1}
= 
\frac{\partial}{\partial r}
\left\{\frac{1}{h(r)}\right\} \left\{
\frac{\partial G(r;\Delta_{j},\Delta_{k}) }{\partial r } 
\right\}^{-1}\\
&=
-
\frac{1}{h(r)^{2}}
\frac{\partial h(r)}{\partial r}\frac{1}{h(r)}
=
-
\frac{1}{h(r)^{3}}
\frac{\partial h(r)}{\partial r}
.
}
By Theorem 6 of \citet{yoon2020sparse}, $h(r)$ is positive and bounded from below by a positive constant independent of $r$, $\Delta_j$, $\Delta_k$. By Lemma \ref{lem:G_second_r_r}, $|\partial h(r)/\partial r|$ is bounded above by a positive constant. Thus, we have $|\partial^{2} G^{-1}(\tau)/\partial \tau^2|<C$ for some $C>0$.
\end{proof}

\begin{lemma}\label{lem:Ginv_second_pij_pij}
Let $G^{-1}(\tau)$ be the inverse bridge function for TT case, where $\tau = G_{TT}(r,\Delta_{j}, \Delta_{k})$.
Under Assumptions \ref{assumption1}--\ref{assumption2}, $|G^{-1}_{\pi_j\pi_j}=\partial^{2} G^{-1}(\tau)/\partial \pi_{j}^{2}|\leq C$ for some constant $C>0$ independent of $r$, $\Delta_j$, $\Delta_k$.
\end{lemma}
\begin{proof}
Let $h(r) = \partial G(r;\Delta_{j},\Delta_{k})/\partial r$ so that $\partial G^{-1}(\tau)/\partial \tau =(\partial G(r;\Delta_{j},\Delta_{k})/\partial r)^{-1}  = 1/h(r)$. By \eqref{eq:multichain1} and multivariate chain rule, we have
\bfln{
\frac{\partial^{2} G^{-1}(\tau)}{\partial \pi_{j}^{2} } 
&= 
\frac{\partial}{\partial \pi_{j}}
\left(
\frac{1}{h(r)}
\frac{\partial G(r;\Delta_{j},\Delta_{k}) }{\partial \Delta_{j} } 
\frac{\partial \Delta_{j}}{\partial \pi_{j}}
\right)\\
=& 
\left[ \frac{\partial}{\partial \pi_{j}}\left\{
\frac{1}{h(r)}\right\}
\right]
\frac{\partial G(r;\Delta_{j},\Delta_{k}) }{\partial \Delta_{j} } 
\frac{\partial \Delta_{j}}{\partial \pi_{j}}
+
\frac{1}{h(r)}
\left[ 
\frac{\partial^{2} G(r;\Delta_{j},\Delta_{k}) }{\partial \pi_{j} \partial \Delta_{j} } 
\right]
\frac{\partial \Delta_{j}}{\partial \pi_{j}}\\
&+
\frac{1}{h(r)}
\frac{\partial G(r;\Delta_{j},\Delta_{k}) }{\partial \Delta_{j} } 
\left[ \frac{\partial^{2} \Delta_{j}}{\partial \pi_{j}^{2}} \right] \\
=&
\left[ \frac{\partial\{h(r)\}^{-1}}{\partial \Delta_{j}}
\frac{\partial\Delta_{j}}{\partial \pi_{j}}
\right]
\frac{\partial G(r;\Delta_{j},\Delta_{k}) }{\partial \Delta_{j} } 
\frac{\partial \Delta_{j}}{\partial \pi_{j}}
+
\frac{1}{h(r)}
\left[ 
\frac{\partial^{2} G(r;\Delta_{j},\Delta_{k}) }{ \partial \Delta_{j}^2 }
\frac{\partial \Delta_{j}}{\partial \pi_{j}}
\right]
\frac{\partial \Delta_{j}}{\partial \pi_{j}}\\
&+
\frac{1}{h(r)}
\frac{\partial G(r;\Delta_{j},\Delta_{k}) }{\partial \Delta_{j} } 
\left[ \frac{\partial^{2} \Delta_{j}}{\partial \pi_{j}^{2}} \right]
\\
=&
\frac{\partial\{h(r)\}^{-1}}{\partial \Delta_{j}}
\frac{\partial G(r;\Delta_{j},\Delta_{k}) }{\partial \Delta_{j} } 
\left(
\frac{\partial \Delta_{j}}{\partial \pi_{j}}
\right)^2
+
\frac{1}{h(r)}
\frac{\partial^{2} G(r;\Delta_{j},\Delta_{k}) }{ \partial \Delta_{j}^2 }
\left(
\frac{\partial \Delta_{j}}{\partial \pi_{j}}
\right)^2\\
&+
\frac{1}{h(r)}
\frac{\partial G(r;\Delta_{j},\Delta_{k}) }{\partial \Delta_{j} } 
\frac{\partial^{2} \Delta_{j}}{\partial \pi_{j}^{2}} 
.
}
We next show that each term is bounded.

Consider $\partial \{h(r)\}^{-1}/\partial \Delta_{j}$. By the multivariate chain rule,
\bfln{
\frac{\partial \{h(r)\}^{-1}}{\partial \Delta_{j}} = 
-\frac{1}{h(r)^2}
\frac{\partial h(r)}{\partial \Delta_{j}}
=
-\frac{1}{h(r)^2}
\frac{\partial^2 G(r;\Delta_{j},\Delta_{k})}{\partial \Delta_{j} \partial r}.
}
The term $|\partial^2 G(r;\Delta_{j},\Delta_{k})/\partial \Delta_{j} \partial r|$ is bounded from above by Lemma \ref{lem:G_second_r_Deltaj}, and $|1/h(r)|$ is bounded from above by Theorem~6 in \citet{yoon2020sparse}. 
Furthermore,  $|\partial G(r;\Delta_{j},\Delta_{k})/ \partial \Delta_{j}|$ is bounded by Lemma \ref{lem:G_first_Deltaj}, $|\partial^2 G(r;\Delta_{j},\Delta_{k})/ \partial \Delta_{j}^2|$ is bounded by Lemma~\ref{lem:G_second_Deltaj_Deltaj}, and  $|\partial \Delta_{j}/\partial \pi_{j}|$, $|\partial^2 \Delta_{j}/\partial \pi_{j}^2|$ are both bounded by Lemma~\ref{lem:Delta_wrt_pi_bounded}. This concludes the proof.
\end{proof}

\begin{lemma}\label{lem:Ginv_second_pik_pij}
Let $G^{-1}(\tau)$ be the inverse bridge function for TT case, where $\tau = G_{TT}(r,\Delta_{j}, \Delta_{k})$.
Under Assumptions \ref{assumption1}--\ref{assumption2}, $|\partial^{2} G^{-1}(\tau)/\partial \pi_{k}\pi_{j}|\leq C$ for some constant $C>0$ independent of $r$, $\Delta_j$, $\Delta_k$.
\end{lemma}
\begin{proof}
Let $h(r) = \partial G(r;\Delta_{j},\Delta_{k})/\partial r$ so that $\partial G^{-1}(\tau)/\partial \tau =(\partial G(r;\Delta_{j},\Delta_{k})/\partial r)^{-1}  = 1/h(r)$. By \eqref{eq:multichain1} and multivariate chain rule, we have
\bfln{
\frac{\partial^{2} G^{-1}(\tau)}{\partial \pi_{k}\pi_{j} } 
=& 
\frac{\partial}{\partial \pi_{k}}
\left(
\frac{1}{h(r)}
\frac{\partial G(r;\Delta_{j},\Delta_{k}) }{\partial \Delta_{j} } 
\frac{\partial \Delta_{j}}{\partial \pi_{j}}
\right)\\
=& 
\left[ \frac{\partial}{\partial \pi_{k}}\left\{
\frac{1}{h(r)}\right\}
\right]
\frac{\partial G(r;\Delta_{j},\Delta_{k}) }{\partial \Delta_{j} } 
\frac{\partial \Delta_{j}}{\partial \pi_{j}}
+
\frac{1}{h(r)}
\left[ 
\frac{\partial^{2} G(r;\Delta_{j},\Delta_{k}) }{\partial \pi_{k} \partial \Delta_{j} } 
\right]
\frac{\partial \Delta_{j}}{\partial \pi_{j}}\\
&+
\frac{1}{h(r)}
\frac{\partial G(r;\Delta_{j},\Delta_{k}) }{\partial \Delta_{j} } 
\left[ \frac{\partial^{2} \Delta_{j}}{\partial \pi_{k} \partial \pi_{j}} \right].
}
As $\partial^2 \Delta_{j}/\partial \pi_{k}\partial \pi_{j}=0 $, we further have that
\bfln{
\frac{\partial^{2} G^{-1}(\tau)}{\partial \pi_{k}\pi_{j} } 
=&
\left[ \frac{\partial\{h(r)\}^{-1}}{\partial \Delta_{k}}
\frac{\partial\Delta_{k}}{\partial \pi_{k}}
\right]
\frac{\partial G(r;\Delta_{j},\Delta_{k}) }{\partial \Delta_{j} } 
\frac{\partial \Delta_{j}}{\partial \pi_{j}}
+
\frac{1}{h(r)}
\left[ 
\frac{\partial^{2} G(r;\Delta_{j},\Delta_{k}) }{ \partial \Delta_{k} \partial \Delta_{j} }
\frac{\partial \Delta_{k}}{\partial \pi_{k}}
\right]
\frac{\partial \Delta_{j}}{\partial \pi_{j}}\\
=&
\frac{\partial\{h(r)\}^{-1}}{\partial \Delta_{k}}
\frac{\partial G(r;\Delta_{j},\Delta_{k}) }{\partial \Delta_{j} } 
\left(
\frac{\partial \Delta_{k}}{\partial \pi_{k}}
\frac{\partial \Delta_{j}}{\partial \pi_{j}}
\right)
+
\frac{1}{h(r)}
\frac{\partial^{2} G(r;\Delta_{j},\Delta_{k}) }{ \partial \Delta_{k} \partial \Delta_{j} }
\left(
\frac{\partial \Delta_{k}}{\partial \pi_{k}}
\frac{\partial \Delta_{j}}{\partial \pi_{j}}
\right),
}
where $|\partial \Delta_{j}/\partial \pi_{j}|$ and $|\partial \Delta_{k}/\partial \pi_{k}|$ are bounded above by some constant $C>0$ by Lemma \ref{lem:Delta_wrt_pi_bounded}. Thus, by the triangle inequality,
\bfln{
\left|
\frac{\partial^{2} G^{-1}(\tau)}{\partial \pi_{k}\pi_{j} } 
\right|
&\leq
\left|
\frac{\partial\{h(r)\}^{-1}}{\partial \Delta_{k}}
\frac{\partial G(r;\Delta_{j},\Delta_{k}) }{\partial \Delta_{j} } 
\left(
\frac{\partial \Delta_{k}}{\partial \pi_{k}}
\frac{\partial \Delta_{j}}{\partial \pi_{j}}
\right)
\right|
+
\left|
\frac{1}{h(r)}
\frac{\partial^{2} G(r;\Delta_{j},\Delta_{k}) }{ \partial \Delta_{k} \partial \Delta_{j} }
\left(
\frac{\partial \Delta_{k}}{\partial \pi_{k}}
\frac{\partial \Delta_{j}}{\partial \pi_{j}}
\right)
\right| \\
&\leq
C^2
\left|
\frac{\partial\{h(r)\}^{-1}}{\partial \Delta_{k}}
\frac{\partial G(r;\Delta_{j},\Delta_{k}) }{\partial \Delta_{j} } 
\right|
+
C^2
\left|
\frac{1}{h(r)}
\frac{\partial^{2} G(r;\Delta_{j},\Delta_{k}) }{ \partial \Delta_{k} \partial \Delta_{j} }
\right|
}
We next show that each term is bounded.

Consider $\partial \{h(r)\}^{-1}/\partial \Delta_{k}$. By the multivariate chain rule,
\bfln{
\frac{\partial \{h(r)\}^{-1}}{\partial \Delta_{k}} = 
-\frac{1}{h(r)^2}
\frac{\partial h(r)}{\partial \Delta_{k}}
=
-\frac{1}{h(r)^2}
\frac{\partial^2 G(r;\Delta_{j},\Delta_{k})}{\partial \Delta_{k} \partial r}.
}
The term $|\partial^2 G(r;\Delta_{j},\Delta_{k})/\partial \Delta_{k} \partial r|$ is bounded from above by Lemma \ref{lem:G_second_r_Deltaj}, and $|1/h(r)|$ is bounded from above by Theorem~6 in \citet{yoon2020sparse}. 
Furthermore,  $|\partial G(r;\Delta_{j},\Delta_{k})/ \partial \Delta_{j}|$ is bounded by Lemma \ref{lem:G_first_Deltaj}, $|\partial^2 G(r;\Delta_{j},\Delta_{k})/ \partial \Delta_{k} \partial \Delta_{j}|$ is bounded by Lemma~\ref{lem:G_second_Deltak_Deltaj}. This concludes the proof.
\end{proof}

\begin{lemma}\label{lem:Ginv_second_tau_pij}
Let $G^{-1}(\tau)$ be the inverse bridge function for TT case, where $\tau = G_{TT}(r,\Delta_{j}, \Delta_{k})$.
Under Assumptions \ref{assumption1}--\ref{assumption2}, $|\partial^{2} G^{-1}(\tau) /\partial \pi_{j}\partial \tau  |\leq C$ for some constant $C>0$ independent of $r$, $\Delta_j$, $\Delta_k$.
\end{lemma}

\begin{proof}
Let $h(r) = \partial G(r;\Delta_{j},\Delta_{k})/\partial r$ so that $\partial G^{-1}(\tau)/\partial \tau =(\partial G(r;\Delta_{j},\Delta_{k})/\partial r)^{-1}  = 1/h(r)$.
By the multivariate chain rule,
\bfln{
\frac{\partial^{2} G^{-1}(\tau)}{ \partial \pi_{j} \partial \tau}
=&
\frac{\partial}{\partial \pi_{j} }
\frac{\partial G^{-1}(\tau)}{\partial \tau }
=
\frac{\partial}{\partial \pi_{j} }
\left\{\frac{1}{h(r)}\right\}
=
\frac{\partial}{\partial \Delta_{j} }\left\{
\frac{1}{h(r) }\right\}\frac{\partial \Delta_{j}}{\partial 
\pi_{j} }
\\
=&
-
\frac{1}{h(r)^2 }
\frac{\partial h(r)}{\partial \Delta_{j} }\frac{\partial \Delta_{j}}{\partial 
\pi_{j} }
=
-
\frac{1}{h(r)^2 }
\frac{\partial G(r;\Delta_{j},\Delta_{k})}{\partial \Delta_{j} \partial r}\frac{\partial \Delta_{j}}{\partial 
\pi_{j} }.
}
The terms $|\partial \Delta_{j}/\partial 
\pi_{j}|$, $|1/h(r)^2|$, and $|\partial^2 G(r;\Delta_{j},\Delta_{k})/\partial \Delta_{j} \partial r|$ are bounded above by constants by Lemma \ref{lem:Delta_wrt_pi_bounded}, Theorem 6 of \citet{yoon2020sparse}, and Lemma \ref{lem:G_second_r_Deltaj}, respectively. Thus, for some constant $C>0$, we have
\bfln{
\left|
\frac{\partial^{2} G^{-1}(\tau)}{ \partial \pi_{j} \partial \tau}
\right| \leq C.
}
\end{proof}

\subsection{Bounds on the partial derivatives of the bridge function}\label{suppl:subsec:bridge_bound}
Here we bound partial derivatives of the bridge function $\brg(r,\Delta_{j}, \Delta_{k})$ for TT case, where
\bfln{
\brg(r,\Delta_{j}, \Delta_{k}) = -2\Phi_{4}(-\Delta_{j},-\Delta_{k},0,0;\bSigma_{4a}) + 2\Phi_{4}(-\Delta_{j},-\Delta_{k},0,0;\bSigma_{4b}).
}
As the bridge function consists of two 4-dimensional Gaussian distribution functions, we will show that, whether  $\bSigma_{4} = \bSigma_{4a}$ or $\bSigma_{4} = \bSigma_{4b}$, the absolute values of partial derivatives of $\Phi_{4}(-\Delta_{j},-\Delta_{k},0,0;\bSigma_{4})$ are bounded from above.


\begin{lemma}\label{lem:G_first_Deltaj}
Under Assumptions \ref{assumption1}--\ref{assumption2}, $|\partial G(r;\Delta_{j},\Delta_{k})/\partial \Delta_{j}|$ and $|\partial G(r;\Delta_{j},\Delta_{k})/\partial \Delta_{k}|$ are bounded above by some constant $C>0$ independent from $r$, $\Delta_j$, $\Delta_k$.
\end{lemma}
\begin{proof}
By the Leibniz rule,
\bfl{ \nonumber
\frac{\partial}{\partial \Delta_{j}}
\Phi_{4}( -\Delta_{j}, -\Delta_{k}, 0, 0; \bSigma_{4})
&  =
\frac{\partial}{\partial \Delta_{j}}
\int_{-\infty}^{0}
\int_{-\infty}^{0} 
\int_{-\infty}^{-\Delta_{k}} 
\int_{-\infty}^{-\Delta_{j}}
\phi(z_{1},z_{2},z_{3},z_{4})
dz_{1} dz_{2}dz_{3}dz_{4} 
\\ \nonumber
&= 
(-1)
\int_{-\infty}^{0}
\int_{-\infty}^{0}
\int_{-\infty}^{-\Delta_{k}} 
\phi(-\Delta_{j},z_{2},z_{3},z_{4})
 dz_{2}dz_{3}dz_{4}.
}
Thus, regardless of $\bSigma_{4} = \bSigma_{4a}$ or $\bSigma_{4} = \bSigma_{4b}$,
\bfln{
\left| 
\frac{\partial}{\partial \Delta_{j}}
\Phi_{4}( -\Delta_{j}, -\Delta_{k}, 0, 0; \bSigma_{4})
\right|
&= 
\int_{-\infty}^{0}
\int_{-\infty}^{0}
\int_{-\infty}^{-\Delta_{k}} 
\phi(-\Delta_{j},z_{2},z_{3},z_{4})
 dz_{2}dz_{3}dz_{4}\\
 &= 
\phi(-\Delta_{j})
\int_{-\infty}^{0}
\int_{-\infty}^{0}
\int_{-\infty}^{-\Delta_{k}} 
\phi(z_{2},z_{3},z_{4}|-\Delta_{j})
dz_{2}dz_{3}dz_{4},
}
where $\phi(z_{2},z_{3},z_{4}|-\Delta_{j})$ is the conditional pdf given $Z_1 = -\Delta_j$. Therefore, the three-dimensional integral above corresponds to a probability (and is bounded by one), leading to
\bfln{
\left| 
\frac{\partial}{\partial \Delta_{j}}
\Phi_{4}( -\Delta_{j}, -\Delta_{k}, 0, 0; \bSigma_{4})
\right|
&\leq \phi(-\Delta_j) \leq \phi(0) = 1/\sqrt{2\pi}.
}
The proof for $\Delta_k$ is analogous.
\end{proof}


\begin{lemma}\label{lem:G_second_r_Deltaj} Under Assumptions \ref{assumption1}--\ref{assumption2}, $|\partial^2 G(r;\Delta_{j},\Delta_{k})|/\partial r \partial \Delta_{j}|$ and $|\partial^2 G(r;\Delta_{j},\Delta_{k})|/\partial r \partial \Delta_{k}|$ are bounded above by some constant $C>0$.
\end{lemma}

\begin{proof}
We start from the partial derivative with respect to $\Delta_{j}$ given in Lemma \ref{lem:G_first_Deltaj} as
\bfln{ 
\frac{\partial}{\partial \Delta_{j}}
\Phi_{4}( -\Delta_{j}, -\Delta_{k}, 0, 0; \bSigma_{4})
&= 
(-1)
\int_{-\infty}^{0}
\int_{-\infty}^{0}
\int_{-\infty}^{-\Delta_{k}} 
\phi(-\Delta_{j},z_{2},z_{3},z_{4};\bSigma_{4})
dz_{2}dz_{3}dz_{4}.
}
Let $\bSigma_{4}=[\rho_{jk}]_{1\leq j,k \leq 4}$ and consider the following multivariate chain rule:
\bfln{ 
\frac{\partial}{\partial r}
\int_{-\infty}^{0}
\int_{-\infty}^{0}
\int_{-\infty}^{-\Delta_{k}}
&\phi(-\Delta_{j},z_{2},z_{3},z_{4};\bSigma_{4})
dz_{2}dz_{3}dz_{4} \\
&= 
\sum_{j<k}
\left\{
\frac{\partial}{\partial \rho_{jk}}
\int_{-\infty}^{0}
\int_{-\infty}^{0}
\int_{-\infty}^{-\Delta_{k}} 
\phi(-\Delta_{j},z_{2},z_{3},z_{4};\bSigma_{4})
dz_{2}dz_{3}dz_{4}
\frac{\partial \rho_{jk}}{\partial r}
\right\}\\
&=
\sum_{j<k}
\left\{
\int_{-\infty}^{0}
\int_{-\infty}^{0}
\int_{-\infty}^{-\Delta_{k}} 
\frac{\partial \phi(-\Delta_{j},z_{2},z_{3},z_{4};\bSigma_{4})}{\partial \rho_{jk}}
dz_{2}dz_{3}dz_{4}
\frac{\partial \rho_{jk}}{\partial r}
\right\}
.
}
In the above, we only consider partial derivatives with respect to $\rho_{12}$, $\rho_{14}$, $\rho_{23}$, and $\rho_{34}$ because $ \rho_{13}$, $\rho_{24}$ do not involve $r$ whether $\bSigma_{4}=\bSigma_{4a}$ or $\bSigma_{4}=\bSigma_{4b}$, i.e., $\partial \rho_{jk}/\partial r = 0$. 

Consider the case $(j,k)=(2,3)$. By \citet{plackett1954reduction},
\bfl{\nonumber
\int_{-\infty}^{0}
\int_{-\infty}^{0}
\int_{-\infty}^{-\Delta_{k}} 
\frac{\partial \phi(-\Delta_{j},z_{2},z_{3},z_{4})}{\partial \rho_{23}}
dz_{2}dz_{3}dz_{4} 
&=
\int_{-\infty}^{0}
\int_{-\infty}^{0}
\int_{-\infty}^{-\Delta_{k}} 
\frac{\partial^{2} \phi(-\Delta_{j},z_{2},z_{3},z_{4})}{\partial z_{2}\partial z_{3}}
dz_{2}dz_{3}dz_{4} 
\\ \nonumber
&=
\int_{-\infty}^{0}
 \phi(-\Delta_{j},-\Delta_{k},0,z_{4})
dz_{4} \\ \nonumber
&=
\int_{-\infty}^{0}
\phi(z_{4}|-\Delta_{j},-\Delta_{k},0)\phi(-\Delta_{j},-\Delta_{k},0)
dz_{4} \\ \nonumber
&=
\phi(-\Delta_{j},-\Delta_{k},0)
\int_{-\infty}^{0}
 \phi(z_{4}|-\Delta_{j},-\Delta_{k},0)
dz_{4}, 
}
where $\phi(z_{4}|-\Delta_{j},-\Delta_{k},0)$ is the conditional pdf given $Z_{1}=-\Delta_{j}$, $Z_{2}=-\Delta_{k}$, and $Z_{3}=0$. Therefore, above integral corresponds to a probability (and is bounded by one), leading to
\bfl{\nonumber
\int_{-\infty}^{0}
\int_{-\infty}^{0}
\int_{-\infty}^{-\Delta_{k}} 
\frac{\partial \phi(-\Delta_{j},z_{2},z_{3},z_{4})}{\partial \rho_{23}}
dz_{2}dz_{3}dz_{4} 
\leq
\phi(-\Delta_{j},-\Delta_{k},0)
\leq |\bSigma_{4}|^{-1/2},
}
where above inequalities hold because $\phi(-\Delta_{j},-\Delta_{k},0)\leq \phi(0,0,0)= |\bSigma_{3}|^{-1/2} \leq |\bSigma_{4}|^{-1/2}$ and $\bSigma_{3}=\var\{(Z_{1},Z_{2},Z_{3})\}$. As Lemma \ref{lem:Sigma4_inv_det} provides that $|\bSigma_{4}|^{-1/2} \leq C$ for some constant $C>0$, we have the desired result. The case $(j,k)=(3,4)$ is similar with the same bound.

For $(j,k)=(1,2)$, again by \citet{plackett1954reduction}, we have
\bfln{
\int_{-\infty}^{0}
\int_{-\infty}^{0}
\int_{-\infty}^{-\Delta_{k}}
\frac{\partial \phi(-\Delta_{j},z_{2},z_{3},z_{4})}{\partial \rho_{12}}
dz_{2}dz_{3}dz_{4} 
&=
-\int_{-\infty}^{0}
\int_{-\infty}^{0}
\frac{\partial \phi(-\Delta_{j},-\Delta_{k},z_{3},z_{4})}{\partial \Delta_{j}}
dz_{3}dz_{4}.
}
For notational convenience, let $\by=(-\Delta_{j},-\Delta_{k},z_{3},z_{4})^{\top}=(y_{1},y_{2},y_{3},y_{4})^{\top}$ and write
\bfl{\nonumber
\int_{-\infty}^{0}
\int_{-\infty}^{0}
\frac{\partial \phi(-\Delta_{j},-\Delta_{k},z_{3},z_{4})}{\partial \Delta_{j}}
dz_{3}dz_{4}
&= 
\int_{-\infty}^{0}
\int_{-\infty}^{0}
\frac{\partial \phi(y_{1},y_{2},y_{3},y_{4})}{\partial y_{1}}
dy_{3}dy_{4}
\\\label{inlem:G_second_r_Deltaj_term2}
&=
\int_{-\infty}^{0}
\int_{-\infty}^{0}
(-\bom_{1}^{\top}\by)
\phi(y_{1},y_{2},y_{3},y_{4})
dy_{3}dy_{4},
}
where $\bom^{\top}_{i}$ is the $i$th row of $\bSigma_{4}^{-1}$.
Then, by extending the range of integrations, the absolute value of \eqref{inlem:G_second_r_Deltaj_term2} is bounded above as
\bfln{
\left|
\int_{-\infty}^{0}
\int_{-\infty}^{0}
(-\bom_{1}^{\top}\by)
\phi(y_{1},y_{2},y_{3},y_{4})
dy_{3}dy_{4}
\right|
&\leq
\int_{-\infty}^{0}
\int_{-\infty}^{0}
\left|\bom_{1}^{\top}\by\right|
\phi(y_{1},y_{2},y_{3},y_{4})
dy_{3}dy_{4}\\
&\leq
\int_{-\infty}^{\infty}
\int_{-\infty}^{\infty}
\left|\bom_{1}^{\top}\by\right|
\phi(y_{1},y_{2},y_{3},y_{4})
dy_{3}dy_{4}.
}
By the triangle inequality,
\bfln{
\int_{-\infty}^{\infty}
\int_{-\infty}^{\infty}
\left|\bom_{1}^{\top}\by\right|
&\phi(y_{1},y_{2},y_{3},y_{4})
dy_{3}dy_{4} \\
&=
\int_{-\infty}^{\infty}
\int_{-\infty}^{\infty}
\sum_{i'=1}^{4}
\left|\omega_{1i'}y_{i'}\right|
\phi(y_{1},y_{2},y_{3},y_{4})
dy_{3}dy_{4} \\
&\leq
\sum_{i'=1}^{2}|\omega_{1i'}y_{i'}|\phi(y_{1},y_{2}) 
+
\sum_{i'=3}^{4} |\omega_{1i'}|
\int_{-\infty}^{\infty}
\int_{-\infty}^{\infty}
|y_{i'}|
\phi(y_{1},y_{2},y_{3},y_{4})
dy_{3}dy_{4} \\
&\leq
|\bSigma_{4}|^{-1/2}
\left\{
\sum_{i'=1}^{2}|\omega_{1i'}y_{i'}|
+
\sum_{i'=3}^{4} |\omega_{1i'}|
\int_{-\infty}^{\infty}
\int_{-\infty}^{\infty}
|y_{i'}|
\phi(y_{3},y_{4}|y_{1},y_{2})
dy_{3}dy_{4}
\right\},
}
where the last inequality holds as $\phi(y_{1},y_{2})\leq |\bSigma_{2}|^{-1/2} \leq |\bSigma_{4}|^{-1/2}$. Under Assumption \ref{assumption2}, $|y_{1}| = |\Delta_j|\leq M$,$|y_{2}|  = |\Delta_{k}|\leq M$. By Lemma \ref{lem:Sigma4_inv_det}, whether $\bSigma_{4}=\bSigma_{4a}$ or $\bSigma_{4}=\bSigma_{4b}$, $|\bSigma_{4}|^{-1/2}$ is bounded above and all elements of $\bSigma_{4}^{-1}=[\omega_{\ell\ell'}]_{1\leq \ell\ell' \leq 4}$ are all bounded above. Thus, we have
\bfln{
\int_{-\infty}^{\infty}
\int_{-\infty}^{\infty}
\left|\bom_{1}^{\top}\by\right|
&\phi(y_{1},y_{2},y_{3},y_{4})
dy_{3}dy_{4}
\leq
C_{1}+ C_{2}
\sum_{i'=3}^{4}
\int_{-\infty}^{\infty}
\int_{-\infty}^{\infty}
|y_{i'}|
\phi(y_{3},y_{4}|y_{1},y_{2})
dy_{3}dy_{4}.
}
By Lemma \ref{lem:bivariate_conditional_moments_bound}, for some constant $C>0$,
\bfln{
\sum_{i'=3}^{4}
\int_{-\infty}^{\infty}
\int_{-\infty}^{\infty}
|y_{i'}|
\phi(y_{3},y_{4}|y_{1},y_{2}) dy_{3}dy_{4}
= 
\sum_{i'=3}^{4}
\E( Y_{i'}|Y_{1}=y_{1}, Y_{2}=y_{2} ) \leq C.
}
This concludes the proof and the proof for $\Delta_{k}$ is analogous.


\end{proof}

\begin{lemma}\label{lem:G_second_Deltaj_Deltaj}
Under Assumptions \ref{assumption1} and \ref{assumption2}, $|\partial^2 G(r;\Delta_{j},\Delta_{k})/\partial \Delta_{j}^2|$ and $|\partial^2 G(r;\Delta_{j},\Delta_{k})/\partial \Delta_{k}^2|$  are bounded above by some constant $C>0$.
\end{lemma}

\begin{proof}
From the proof of Lemma \ref{lem:G_first_Deltaj}, we have 
\bfln{ 
\frac{\partial}{\partial \Delta_{j}}
\Phi_{4}( -\Delta_{j}, -\Delta_{k}, 0, 0; \bSigma_{4})
&= 
(-1)
\int_{-\infty}^{0}
\int_{-\infty}^{0}
\int_{-\infty}^{-\Delta_{k}} 
\phi(-\Delta_{j}|z_{2},z_{3},z_{4})
\phi(z_{2},z_{3},z_{4})  dz_{2}dz_{3}dz_{4}.
}
By interchanging differentiation and integration,
\bfln{ 
\frac{\partial^{2}}{\partial \Delta_{j}^{2}}
\Phi_{4}( -\Delta_{j}, -\Delta_{k}, 0, 0; \bSigma_{4})
=& 
(-1)
\int_{-\infty}^{0}
\int_{-\infty}^{0}
\int_{-\infty}^{-\Delta_{k}} 
\frac{\partial}{\partial \Delta_{j}}
\phi(-\Delta_{j}|z_{2},z_{3},z_{4})
\phi(z_{2},z_{3},z_{4})  dz_{2}dz_{3}dz_{4} \\
=& 
\int_{-\infty}^{0}
\int_{-\infty}^{0}
\int_{-\infty}^{-\Delta_{k}} 
\frac{\Delta_{j}+\mu}{v^2}
\phi(-\Delta_{j}|z_{2},z_{3},z_{4})
\phi(z_{2},z_{3},z_{4})  dz_{2}dz_{3}dz_{4},
}
where $\E( z_{1}|z_{2},z_{3},z_{4}) = \mu$ and $\var(z_{1}|z_{2},z_{3},z_{4})= v^2$ as in Lemma \ref{lem:univariate_conditional}.
Thus
\bfl{\label{inlem:G_second_Deltaj_Deltaj_term1} 
\left|
\frac{\partial^{2}}{\partial \Delta_{j}^{2}}
\Phi_{4}( -\Delta_{j}, -\Delta_{k}, 0, 0; \bSigma_{4})
\right|
\leq& 
\int_{-\infty}^{0}
\int_{-\infty}^{0}
\int_{-\infty}^{-\Delta_{k}} 
\left|\frac{\Delta_{j}}{v^2}\right|
\phi(-\Delta_{j}|z_{2},z_{3},z_{4})
\phi(z_{2},z_{3},z_{4})  dz_{2}dz_{3}dz_{4} \\
\label{inlem:G_second_Deltaj_Deltaj_term2}
&+ 
\int_{-\infty}^{0}
\int_{-\infty}^{0}
\int_{-\infty}^{-\Delta_{k}} 
\left| \frac{\mu}{v^2} \right|
\phi(-\Delta_{j}|z_{2},z_{3},z_{4})
\phi(z_{2},z_{3},z_{4})  dz_{2}dz_{3}dz_{4}.
}
Consider the first term \eqref{inlem:G_second_Deltaj_Deltaj_term1}. Following the proof of Lemma~\ref{lem:G_first_Deltaj},
\bfl{\nonumber 
\int_{-\infty}^{0}
\int_{-\infty}^{0}
\int_{-\infty}^{-\Delta_{k}}
&
\left|\frac{\Delta_{j}}{v^2}\right|
\phi(-\Delta_{j}|z_{2},z_{3},z_{4})
\phi(z_{2},z_{3},z_{4})  dz_{2}dz_{3}dz_{4} \\
\nonumber
=&
\left|\frac{\Delta_{j}}{v^2}\right|
\int_{-\infty}^{0}
\int_{-\infty}^{0}
\int_{-\infty}^{-\Delta_{k}}
\phi(-\Delta_{j},z_{2},z_{3},z_{4})  dz_{2}dz_{3}dz_{4} 
\\ 
\leq&
\left|\frac{\Delta_{j}}{v^2}\right| \frac{1}{\sqrt{2\pi }}
\leq C,
}
where the last inequality holds as $|\Delta_{j}|\leq M$ under Assumption \ref{assumption2}, and $v^2$ is bounded below by Lemma \ref{lem:univariate_conditional}.

Consider the second term \eqref{inlem:G_second_Deltaj_Deltaj_term2}. Let $\bz_{-1}=(z_{2},z_{3},z_{4})^{\top}$ and write $\mu = rz_{2} + z_{3}/\sqrt{2} - rz_{4}/\sqrt{2} = \bu^{\top}\bz_{-1}$ as in Lemma \ref{lem:univariate_conditional}. Then, since $\phi(-\Delta_{j}|z_{2},z_{3},z_{4}) \leq 1/\sqrt{2\pi v^{2}}$,
\bfln{
\int_{-\infty}^{0}
\int_{-\infty}^{0}
\int_{-\infty}^{-\Delta_{k}} 
&
\left| \frac{\mu}{v^2} \right|
\phi(-\Delta_{j}|z_{2},z_{3},z_{4})
\phi(z_{2},z_{3},z_{4})  dz_{2}dz_{3}dz_{4} \\
&\leq
\frac{1}{v^3\sqrt{2\pi}}
\int_{-\infty}^{0}
\int_{-\infty}^{0}
\int_{-\infty}^{-\Delta_{k}} 
| \bu^{\top} \bz_{-1} |
\phi(z_{2},z_{3},z_{4})  dz_{2}dz_{3}dz_{4}.
}
Since $\bu^{\top}\bz_{-1} \sim \N(0, \frac{1+r^2}{2})$ and $|\bu^{\top}\bz_{-1}|$ follows the folded Gaussian with mean $\E|\bu^{\top}\bz_{-1}| = \sqrt{(1+r^2)/\pi}$, we further have that
\bfln{
\int_{-\infty}^{0}
\int_{-\infty}^{0}
\int_{-\infty}^{-\Delta_{k}} 
\left| \frac{\mu}{v^2} \right|
\phi(-\Delta_{j}|z_{2},z_{3},z_{4})
\phi(z_{2},z_{3},z_{4})  dz_{2}dz_{3}dz_{4} 
&\leq
\frac{1}{v^3\sqrt{2\pi}} \E|\bu^{\top}\bz_{-1}| \\
&= \frac{1}{v^3\pi}\sqrt{\frac{1+r^2}{2}} \\
&\leq C,
}
where the last inequality holds as $|r|\leq 1- \varepsilon_{r}$ and $v^3$ is bounded below under Assumption \ref{assumption1}.
\end{proof}

\begin{lemma}\label{lem:G_second_Deltak_Deltaj}
Under Assumptions \ref{assumption1} and \ref{assumption2}, $|\partial^2 G(r;\Delta_{j},\Delta_{k})/\partial \Delta_{k} \partial \Delta_{j}|$ is bounded above by some constant $C>0$.
\end{lemma}

\begin{proof}
By the Leibniz rule,
\bfln{ 
\frac{\partial^2}{\partial \Delta_{k}\partial \Delta_{j}}
\Phi_{4}( -\Delta_{j}, -\Delta_{k}, 0, 0; \bSigma_{4})
&= 
\int_{-\infty}^{0}
\int_{-\infty}^{0}
\phi(z_{3},z_{4}|-\Delta_{j},-\Delta_{k})
\phi(-\Delta_{j},-\Delta_{k})  dz_{3}dz_{4}\\
&= 
\phi(-\Delta_{j},-\Delta_{k})
\int_{-\infty}^{0}
\int_{-\infty}^{0}
\phi(z_{3},z_{4}|-\Delta_{j},-\Delta_{k}) dz_{3}dz_{4},
}
where $\phi(z_{3},z_{4}|-\Delta_{j},-\Delta_{k})$ is the conditional pdf given $Z_{1}=-\Delta_{j}$ and $Z_{2}=-\Delta_{k}$. Thus, the two-dimensional integral above corresponds to a probability (and is bounded by one), leading to
\bfln{
\left|
\frac{\partial^2}{\partial \Delta_{k}\partial \Delta_{j}}
\Phi_{4}( -\Delta_{j}, -\Delta_{k}, 0, 0; \bSigma_{4})
\right| \leq \phi(-\Delta_{j},-\Delta_{k}) \leq \phi(0,0) = |\bSigma_{2}|^{-1/2} \leq |\bSigma_{4}|^{-1/2},
}
where $\bSigma_{2}=\var\{(Z_{1},Z_{2})\}$. As Lemma \ref{lem:Sigma4_inv_det} provides that $|\bSigma_{4}|^{-1/2} \leq C$ for some constant $C>0$, this concludes the proof.
\end{proof}

\begin{lemma}\label{lem:G_second_r_r}
Under Assumptions \ref{assumption1} and \ref{assumption2}, $|\partial^2 G(r;\Delta_{j},\Delta_{k})/\partial r^2|$ is bounded above by some constant $C>0$.
\end{lemma}

\begin{proof}
We start from the partial derivative with respect to $r$ given in Theorem 6 of \citet{yoon2020sparse} as
\bfln{
\frac{\partial G(r,\Delta_{j},\Delta_{k} )}{\partial r} & = 
-2\frac{\partial \Phi_{4}\{\Delta_{i},\Delta_{j},0,0;\bSigma_{4a}(r)\} }{\partial r}
+
2\frac{\partial \Phi_{4}\{\Delta_{i},\Delta_{j},0,0;\bSigma_{4b}(r)\} }{\partial r}
\\
= &
\sqrt{2}h_{14a}(r) + \sqrt{2}h_{23a}(r) + 2 h_{23a}(r) +2 h_{12b}(r) + \sqrt{2}h_{14b}(r) + \sqrt{2} h_{23b}(r) + 2 h_{34b}(r),
}
where $h_{14a}(r)$ is defined as
\bfln{
h_{14a}(r) = \frac{\partial \Phi(a_{1},\ldots,a_{4};\bSigma_{4a})}{\partial \rho_{14}(r)}
=
\int_{-\infty}^{a_{3}}  \int_{-\infty}^{a_{2} }
\phi(a_{1},y_{2}, y_{3},a_{4};\bSigma_{4})dy_{2} dy_{3}
}
and the rest of $h_{ij}(r)$'s are analogously defined.

As $\partial G(r;\Delta_{j},\Delta_{k})/\partial r$ is a sum of $h_{ij}(r)$'s, we show that $|\partial h_{ij}(r)/\partial r|$ is bounded above for all $i$ and $j$ whether $\bSigma_{4} = \bSigma_{4a}$ and $\bSigma_{4} = \bSigma_{4b}$. 
Using the multivariate chain rule and triangle inequality,
\bfln{
\Big|
\frac{\partial h_{ij}(r)}{\partial r} 
\Big|
= 
\left|
\sum_{k<\ell} \frac{\partial h_{ij}(r)}{\partial \rho_{k\ell}}\frac{\partial \rho_{k\ell}}{\partial r}
\right|
\leq
\sum_{k<\ell} 
\left|
\frac{\partial h_{ij}(r)}{\partial \rho_{k\ell}} 
\right|
\left|
\frac{\partial \rho_{k\ell}}{\partial r}
\right|.
}
By Lemma \ref{lem:h_wrt_r_bound}, for all $1\leq i<j\leq 4$ and $1\leq k < \ell \leq 4$, $|\partial h_{ij}(r)/ \partial \rho_{k\ell}| \leq C$ for some constant $C>0$. Also, as $\rho_{k\ell}$'s are linear in $r$, $|\partial \rho_{k\ell}/\partial r|$'s are bounded above some positive constant. This concludes the proof.
\end{proof}

\subsection{Auxillary lemmas}
From Theorem 4 in \citet{yoon2020sparse}, the bridge function for TT case takes the following form 
\bfln{
\brg(r,\Delta_{j}, \Delta_{k}) = -2\Phi_{4}(-\Delta_{j},-\Delta_{k},0,0,;\bSigma_{4a}) + 2\Phi_{4}(-\Delta_{j},-\Delta_{k},0,0,;\bSigma_{4b})
}
with $\Delta_{j} = f_{j}(D_{j})$, $\Delta_{k} = f_{k}(C_{k})$,
\bfl{\label{eq:bridge_covs}
&\bSigma_{4a}=
\bpm
1 &  0  & 1/\sqrt{2}  & -r/\sqrt{2}\\
0 &  1  & -r/\sqrt{2} & 1/\sqrt{2}\\
1/\sqrt{2}  & -r/\sqrt{2} & 1 &  -r \\
-r/\sqrt{2} & 1/\sqrt{2} & -r & 1 
\epm, \quad
\bSigma_{4b}=
\bpm
1 &  r & 1/\sqrt{2}  & r/\sqrt{2}   \\
r & 1 & r/\sqrt{2} & 1/\sqrt{2}  \\ 
1/\sqrt{2}  & r/\sqrt{2} & 1 &  r \\
r/\sqrt{2} & 1/\sqrt{2} & r & 1 
\epm.
}

\begin{lemma}\label{lem:Sigma4_inv_det} Let $\bSigma_{4} = \bSigma_{4a}$ or $\bSigma_{4} = \bSigma_{4b}$ from above, and let its inverse be $\bSigma_{4}^{-1} = [\omega_{\ell\ell'}]_{1\leq \ell,\ell' \leq 4}$. Under Assumption \ref{assumption1}, $|\omega_{\ell\ell'}|\leq C_{1}$, $1\leq \ell,\ell' \leq 4$, for some constant $C_{1}>0$. Also, $|\bSigma_{4}|^{-1} \leq C_{2}$ for some constant $C_{2}>0$ regardless of $\bSigma_{4} = \bSigma_{4a}$ or $\bSigma_{4} = \bSigma_{4b}$. 
\end{lemma}
\begin{proof}
Computing determinants gives $|\bSigma_{4a}| = |\bSigma_{4b}| =(1-r^2)^2/4$, and thus by Assumption \ref{assumption1}, $|\bSigma_{4}|\geq (1-(1-\varepsilon_{r}^2)^2/4$, whether $\bSigma_{4} = \bSigma_{4a}$ or $\bSigma_{4} = \bSigma_{4b}$. Also,
the inverses of $\bSigma_{4a}$ and $\bSigma_{4b}$ are given by
\bfln{
&\bSigma_{4a}^{-1}=
\frac{1}{r^2-1}
\bpm
-2 &  2r  & \sqrt{2}  & -\sqrt{2}r\\
2r &  -2  & -\sqrt{2}r & \sqrt{2}\\
\sqrt{2}  & -\sqrt{2}r & -2 &  0 \\
-\sqrt{2}r & \sqrt{2} & 0 & -2 
\epm, \quad
\bSigma_{4b}^{-1}=
\frac{1}{r^2-1}
\bpm
-2 &  2r & \sqrt{2}  & -\sqrt{2}r   \\
2r & -2 & -\sqrt{2}r &  \sqrt{2}  \\ 
\sqrt{2}  & -\sqrt{2}r & -2 &  2r \\
-\sqrt{2}r & \sqrt{2} & 2r & -2
\epm,
}
respectively. Under Assumption \ref{assumption1}, $|\omega_{\ell\ell'}|$'s are all bounded above by $2/\{1-(1-\varepsilon_{r})^{2}\}$.
\end{proof}

\begin{lemma}\label{lem:Delta_wrt_pi_bounded}
Let $\Delta=\Phi^{-1}(\pi)$. Then, under Assumption \ref{assumption2},
\bfln{
\left|\frac{\partial \Delta}{\partial \pi}\right|\leq C_{1}
\quad \text{and} \quad
\left|\frac{\partial^{2} \Delta}{\partial \pi^{2} }\right|\leq C_{2}
}
for some constants $C_{1},C_{2}>0$.
\end{lemma}
\begin{proof}
Since $|\Delta|\leq M$ and $\phi(|x|)$ is a decreasing function, 
\bfln{
\frac{\partial \Delta}{\partial \pi} = 
\left( 
\frac{\partial \pi}{\partial \Delta}
\right)^{-1}
=
\left\{
\frac{\partial \Phi(\Delta)}{\partial \Delta}
\right\}^{-1}
=
\frac{1}{\phi(\Delta)}
\leq
\frac{1}{\phi(M)}.
}
Furthermore, as the second derivative is given by
\bfln{
\frac{\partial^2 \Delta}{\partial \pi^2}
=
\frac{\partial }{\partial \pi}
\frac{1}{\phi(\Delta)}
=
\frac{\partial \Delta }{\partial \pi}
\frac{\partial }{\partial \Delta}
\frac{1}{\phi(\Delta)}
=
-\frac{1}{\{\phi(\Delta)\}^3 }
\frac{\partial \phi(\Delta)}{\partial \Delta}
=
\frac{\Delta}{\{\phi(\Delta)\}^2 }
}
we have
\bfln{
\left|
\frac{\partial^2 \Delta}{\partial \pi^2}
\right| \leq \frac{M}{\{\phi(M)\}^2}.
}
\end{proof}
\begin{lemma}\label{lem:univariate_conditional}
Let $(Z_{1},Z_{2},Z_{3},Z_{4})^{\top} \sim \N_{4}(\zeros,\bSigma_{4})$. Then, it follows that regardless of $\bSigma_{4}=\bSigma_{4a}$ or $\bSigma_{4}=\bSigma_{4b}$, the conditional distribution of $Z_{1}$ given $Z_{2},Z_{3},Z_{4}$ is $\N(\mu,v^2)$, where 
\bfln{
\mu &\coloneqq \E( Z_{1}|Z_{2},Z_{3},Z_{4} ) = rZ_{2} +Z_{3}/\sqrt{2} - rZ_{4}/\sqrt{2} 
 \\
v^2 &\coloneqq \var( Z_{1}|Z_{2},Z_{3},Z_{4} ) =
(1-r^2)/2.
}
\end{lemma}
\begin{proof}
The results follow from the properties of conditional normal distribution using the form of $\bSigma_{4a}$ and $\bSigma_{4b}$.
\end{proof}

\begin{lemma}\label{lem:bivariate_conditional_moments_bound}
Let $\by \sim \N_{4}(\zeros,\bSigma_{4})$, where $\bSigma_{4}=\bSigma_{4a}$ or $\bSigma_{4}=\bSigma_{4b}$. Then, under Assumptions \ref{assumption1} and \ref{assumption2}, for any $1\leq k<\ell \leq 4$ and $1\leq i \leq 4$,
\bfln{
\E\left( |Y_{i}|~ \big|  Y_{k}=y_{k}, Y_{\ell} =y_{\ell} \right) 
\leq C_{1}, 
\\
0<
C_{2}
\leq
\var\left(Y_{i}~ \big|  Y_{k}=y_{k}, Y_{\ell} =y_{\ell} \right)
\leq C_{3}
}
for some $C_{1},C_{2},C_{3}>0$, where
\bfln{
y_{m} = 
\begin{cases}
-\Delta_{j}, & \text{if }  \quad m = 1;\\
-\Delta_{k}, & \text{if }  \quad m = 2;\\
0, & \text{otherwise. }
\end{cases}
}
\end{lemma}
\begin{proof}
We first calculate the conditional means and covariance matrices using the properties of multivariate Gaussian distribution to obtain:
\bfln{
\E(Y_{1},Y_{4}|Y_{2}=-\Delta_{k},Y_{3}=0; \bSigma_{4a}) &= \frac{1}{2-r^2} \bpm -\Delta_{k}r \\ \sqrt{2}\Delta_{k}(1-r^2) \epm, \\
\E(Y_{2},Y_{3}|Y_{1}=-\Delta_{j},Y_{4}=0; \bSigma_{4a}) &= \frac{1}{2-r^2} \bpm -\Delta_{j}r \\ \sqrt{2}\Delta_{j}(1-r^2) \epm, \\
\E(Y_{3},Y_{4}|Y_{1}=-\Delta_{j},Y_{2}=-\Delta_{k}; \bSigma_{4a}) &= \frac{1}{\sqrt{2}} \bpm \Delta_{j}r - \Delta_{i}   \\ \Delta_{i}r -\Delta_{j} \epm, \\
\E(Y_{1},Y_{2}|Y_{3}=0,Y_{4}=0; \bSigma_{4b}) &= \bpm 0 \\ 0 \epm, \\
\E(Y_{1},Y_{4}|Y_{2}=-\Delta_{k},Y_{3}=0; \bSigma_{4b}) &= \frac{1}{2-r^2} \bpm -\Delta_{k}r \\ -\sqrt{2}\Delta_{k}(1-r^2) \epm, \\
\E(Y_{2},Y_{3}|Y_{1}=-\Delta_{j},Y_{4}=0; \bSigma_{4b}) &= \frac{1}{2-r^2} \bpm -\Delta_{j}r \\ -\sqrt{2}\Delta_{j}(1-r^2) \epm, \\
\E(Y_{3},Y_{4}|Y_{1}=-\Delta_{j},Y_{2}=-\Delta_{k}; \bSigma_{4b}) &= \frac{1}{\sqrt{2}} \bpm -\Delta_{j}\\ -\Delta_{k}\epm, \\
}
and
\bfln{
&\var(Y_{1},Y_{4}|Y_{2},Y_{3};\bSigma_{4a}) = \var(Y_{2},Y_{3}|Y_{1},Y_{4};\bSigma_{4a}) \\
&= \var(Y_{1},Y_{4}|Y_{2},Y_{3};\bSigma_{4b}) = \var(Y_{2},Y_{3}|Y_{1},Y_{4};\bSigma_{4b}) =
\left(1-\frac{1}{2-r^2}\right)
\bpm
2 & \sqrt{2}r\\
\sqrt{2}r & 2 
\epm, \\
&\var(Y_{1},Y_{2}|Y_{3},Y_{4};\Sigma_{4b}) = \var(Y_{3},Y_{4}|Y_{1},Y_{2};\Sigma_{4b}) =
\frac{1}{2}
\bpm
1 & -r \\
-r & 1 
\epm,\\
&\var(Y_{3},Y_{4}|Y_{1},Y_{2};\Sigma_{4a}) =
\frac{2}{1-r^2}
\bpm
1 & 0\\
0 & 1
\epm
}

Consider $\var(Y_{i}|Y_{k} = y_{k}, Y_{\ell} = y_{\ell})$. From the above, it is clear that all conditional variances are bounded below by some positive constant. It can be also seen that all conditional variances are bounded above as long as $1-r^2 \geq C$ for some constant $C>0$. Under Assumptions \ref{assumption1}, $|r|\leq 1-\varepsilon_{r}$ and thus $1-r^2 \geq 1- (1-\varepsilon_{r})^2 >0$.

Consider $\E\left( |Y_{i}|~ \big|  Y_{k}=y_{k}, Y_{\ell} =y_{\ell} \right)$. If $i=j$ or $i=\ell$, then $\E\left( |Y_{i}|~ \big|  Y_{k}=y_{k}, Y_{\ell} =y_{\ell} \right) = |y_{i}| $ and the result is immediate under Assumption \ref{assumption2}. For $i \neq j,k$, let $\E(Y_{i}|Y_{k} = y_{k}, Y_{\ell} = y_{\ell}) = \mu_{i}$ and $\var(Y_{i}|Y_{k} = y_{k}, Y_{\ell} = y_{\ell}) = \sigma_{i}^2$, where detailed expressions are given above. Then, by Lemma \ref{lem:foldedGaussian_mean}, we have that
\bfln{
\E\left\{|Y_{i}|~ \big|  Y_{k}=y_{k}, Y_{\ell}=y_{\ell} \right\}
&=
\left[
\sigma_{i}\sqrt{ \frac{2}{\pi} } \exp\left(-\frac{\mu_{i}^{2}}{2\sigma_{i}^2}\right) +
\mu_{i}\left\{ 1 - 2\Phi\left(-\frac{\mu_{i}}{\sigma_{i}}\right)\right\}
\right] \\
&\leq
\sigma_{i}\sqrt{ \frac{2}{\pi} }  + |\mu_{i}|.
}
We can see from the above conditional means that, under Assumptions \ref{assumption1} and \ref{assumption2}, $|\mu_{i}|$ is bounded above by some positive constant. As we already showed that $\sigma_{i}^2$ is bounded above, the proof is complete.
\end{proof}

\begin{lemma}\label{lem:univariate_conditional_moments_bound}
Let $\by \sim \N_{4}(\zeros,\bSigma_{4})$, where $\bSigma_{4}=\bSigma_{4a}$ or $\bSigma_{4}=\bSigma_{4b}$. Also let $\by_{-i}$ be the 3-dimensional random vector without the $i$th component and $y_{-i}=(y_{j},y_{k},y_{\ell})^{\top}$ be its realization such that  
\bfln{
y_{m} = 
\begin{cases}
-\Delta_{j}, & \text{if }  \quad m = 1;\\
-\Delta_{k}, & \text{if }  \quad m = 2;\\
0, & \text{otherwise. }
\end{cases}
}
Then, under Assumptions \ref{assumption1} and \ref{assumption2}, for any $1\leq i \leq 4$, 
\bfln{
\E\left( |Y_{i}|~ \big|  \by_{-i}=y_{-i}\right) 
\leq C
}
for some constant $C>0$.
\end{lemma}
\begin{proof}
It follows by the conditional mean and variance formulas of the multivariate Gaussian distribution that, regardless of $\bSigma_{4}=\bSigma_{4a}$ or $\bSigma_{4}=\bSigma_{4b}$,
\bfln{
\var(Y_{i}|\by_{-i}=y_{-i};\bSigma_{4}) = \frac{(1-r^2)}{2}, \quad i=1,\ldots,4,
}
and
\bfln{
\E(Y_{1}|\by_{-1}=y_{-1};\bSigma_{4})&= -\Delta_{k}r, \\
\E(Y_{2}|\by_{-2}=y_{-2};\bSigma_{4})&= -\Delta_{j}r, \\
\E(Y_{3}|\by_{-3}=y_{-3};\bSigma_{4})&= - \frac{\Delta_{j}-\Delta_{k}r}{\sqrt{2}}, \\
\E(Y_{4}|\by_{-4}=y_{-4};\bSigma_{4})&= - \frac{\Delta_{k}-\Delta_{j}r}{\sqrt{2}}. \\
}
Under Assumptions \ref{assumption1} and \ref{assumption2}, the absolute values of the conditional means and conditional variances are bounded above by $\sqrt{2}M$ and $1/2$, respectively. Then the result follows by Lemma \ref{lem:foldedGaussian_mean}.
\end{proof}

\begin{lemma}\label{lem:bound_expection_of_product}
Let $\by \sim \N_{4}(\zeros,\bSigma_{4})$, where $\bSigma_{4}=\bSigma_{4a}$ or $\bSigma_{4}=\bSigma_{4b}$. Then, for any $1\leq k<\ell \leq 4$ and $1\leq i<j \leq 4$,
\bfln{
\E\left\{|Y_{i}Y_{j}|~ \big|  Y_{k}=y_{k}, Y_{\ell} =y_{\ell} \right\} \leq C
}
for some constant $C>0$. 
\end{lemma}
\begin{proof}
Let $I=\{k,\ell\}$ and write $\by_{I} = (y_{k},y_{\ell})^{\top} = (y_{I_{1}},y_{I_{2}})^{\top}$.  We prove this lemma by considering the following three cases: $\card\left(\{i,j\} \cap I\right) =0,1,2$, namely cases 1, 2, and 3, respectively.

\paragraph{Case 1.} Consider the case $\card(\{i,j\} \cap I) = 0$. Let 
\bfln{
\E(Y_{i}|\by_{I}) = \mu_{i}, &\quad \var(Y_{i}|\by_{I}) = \sigma_{i}^2, \\
\E(Y_{j}|\by_{I}) = \mu_{j}, &\quad 
\var(Y_{j}|\by_{I}) = \sigma_{j}^2,
}
whose detailed expressions are provided in Lemma \ref{lem:bivariate_conditional_moments_bound}. Also, let $Z_{i}=Y_{i}/\sqrt{2}\sigma_{i} \sim \N(\mu_{i}/\sqrt{2}\sigma_{i},1/2)$ and $Z_{j}=Y_{j}/\sqrt{2}\sigma_{j} \sim \N(\mu_{j}/\sqrt{2}\sigma_{j},1/2)$ and write 
\bfln{
|Y_{i}Y_{j}| = 2\sigma_{i}\sigma_{j} 
\left|
\frac{Y_{i}}{\sqrt{2}\sigma_{i}}
\frac{Y_{j}}{\sqrt{2}\sigma_{j}}
\right| 
&=
2 \sigma_{i}\sigma_{j} 
\left| 
\frac{1}{4}(Z_{i}+Z_{j})^2 - 
\frac{1}{4}(Z_{i}-Z_{j})^2
\right| \\
&\leq
2\sigma_{i}\sigma_{j} 
\left\{
\frac{1}{4}(Z_{i}+Z_{j})^2
+
\frac{1}{4}(Z_{i}-Z_{j})^2
\right\}
,
}
where the last inequality holds by the triangle inequality. We have that $(Z_{i}+Z_{j})^{2}$ and $(Z_{i}-Z_{j})^{2}$ follow non-central $\chi^2_{\textup{df}=1}$ distributions with non-centrality parameters $\lambda_{+} = \mu_{i}^2/(2\sigma_{i}^2) + \mu_{j}^2/(2\sigma_{j}^2)$ and  $\lambda_{-} = \mu_{i}^2/(2\sigma_{i}^2) - \mu_{j}^2/(2\sigma_{j}^2)$, 
and thus,
\bfln{
\E\left\{|Y_{i}Y_{j}|~ \big|  Y_{I_{1}},Y_{I_{2}} \right\}
\leq
\frac{\sigma_{i}\sigma_{j}}{2}
\left\{ \lambda_{+} + \lambda_{-} +2 \right\}.
}
By Lemma \ref{lem:bivariate_conditional_moments_bound}, we have $\E\left\{|Y_{i}Y_{j}|~ \big|  Y_{I_{1}},Y_{I_{2}} \right\}<C$ for some constant $C>0$.

\paragraph{Case 2.} For the case $\card\left(\{i,j\} \cap I\right)=1$, we assume, without loss of generality that, $\{i,j\} \cap I= \{i\}$ and write
\bfln{
\E\left\{|Y_{i}Y_{j}|~ \big|  Y_{I_{1}} = y_{I_{1}},Y_{I_{2}} = y_{I_{2}} \right\}
&=
|y_{I_{1}}| ~ \E\left\{|Y_{j}|~ \big|  Y_{I_{1}} = y_{I_{1}},Y_{I_{2}} = y_{I_{2}} \right\} \\
&\leq 
M ~ \E\left\{|Y_{j}|~ \big|  Y_{I_{1}} = y_{I_{1}},Y_{I_{2}} = y_{I_{2}} \right\},
}
where $|y_{k}|\leq M$ by Assumption \ref{assumption2}. Then, by Lemma \ref{lem:bivariate_conditional_moments_bound}, we have
\bfln{
\E\left\{|Y_{i}Y_{j}|~ \big|  Y_{I_{1}} = y_{I_{1}},Y_{I_{2}} = y_{I_{2}} \right\} \leq C
}
for some constant $C>0$.

\paragraph{Case 3.} For the case $\{i, j\} \cap I = \{i,j\}$, Assumption \ref{assumption2} gives that
\bfln{
\E\left\{|Y_{i}Y_{j}|~ \big|  Y_{I_{1}} = y_{I_{1}},Y_{I_{2}} = y_{I_{2}} \right\} = |y_{I_{1}}y_{I_{2}}| \leq M^2.
}
This concludes the proof.
\end{proof}

\begin{lemma}\label{lem:h_wrt_r_bound}
Let $h_{ij}(r) = \partial \Phi(a_{1},\ldots,a_{4};\bSigma_{4})/\partial \rho_{ij}(r)$, where $\bSigma_{4}=[\rho_{ij}(r)]_{1\leq i,j \leq 4}$. Then, for any $1\leq i<j \leq 4$ and $1\leq k<\ell \leq 4$, 
\bfln{
\left|\frac{\partial h_{ij}(r)}{\partial\rho_{k\ell}}\right| \leq C
}
for some constant $C>0$.
\end{lemma}
\begin{proof}
Let $I=\{i,j\}$, $K=\{k,\ell\}$. We write $\bx_{I}=(x_{i},x_{j})^{\top}=(x_{I_{1}},x_{I_{2}})^{\top}$ and $\Rcal_{I}=\{(x_{i},x_{j})| x_{i}<a_{i}, x_{j}<a_{j}\} \subset \R^{2}$, where. 
\bfln{
a_{m} =
\begin{cases}
-\Delta_{j}, & \text{if }  \quad m = 1;\\
-\Delta_{k}, & \text{if }  \quad m = 2;\\
0, & \text{otherwise. }
\end{cases}
}
We consider three cases where $\card(I\cap K)=2,1,0$. 

Consider the case $\card(I\cap K)=0$, i.e., $K= \{1,\ldots,4\}-I = I^{c}$. By \citet{plackett1954reduction}
\bfln{
\frac{\partial h_{I}(r)}{ \partial \rho_{I^{c}}} 
=
\frac{\partial}{\partial \rho_{I^{c}}}
\int_{\Rcal_{I^{c}}}
\phi(\ba_{I}, \bx_{I^{c}};\bSigma_{4})d\bx_{I^{c}} = \phi(\ba_{I}, \ba_{I^{c}};\bSigma_{4})
\leq
|\bSigma_{4}|^{1/2}
}
because $\phi(\ba_{I}, \ba_{I^{c}};\bSigma_{4})\leq \phi(0,0,0,0;\bSigma_{4}) = |\bSigma_{4}|^{-1/2}$. By Lemma \ref{lem:Sigma4_inv_det}, we have
\bfln{
\left|
\frac{\partial h_{I}(r)}{ \partial \rho_{I^{c}}}
\right|
\leq C
}
for some constant $C>0$.

Consider the case $\card(I\cap K)=2$, i.e., $I=K$. For notational convenience, we write $\ba_{I}=\by_{I}$ and $ \bx_{I^{c}} = \by_{I^{c}}$. Then, 
\bfln{
\frac{\partial h_{I}(r)}{ \partial \rho_{I}} 
&=
\frac{\partial}{\partial \rho_{I}}
\int_{\Rcal_{I^{c}}}
\phi(\ba_{I}, \bx_{I^{c}};\bSigma_{4})d\bx_{I^{c}} 
=
\frac{\partial}{\partial \rho_{I}}
\int_{\Rcal_{I^{c}}}
\phi(\by_{I}, \by_{I^{c}};\bSigma_{4})d\by_{I^{c}} \\
&=
\int_{\Rcal_{I^{c}}}
\frac{\partial}{\partial \rho_{I}}
\phi(\by_{I}, \by_{I^{c}};\bSigma_{4})d\by_{I^{c}}
=
\int_{\Rcal_{I^{c}}}
\frac{\partial^2}{\partial y_{I_{1}} \partial y_{I_{2}} }
\phi(\by_{I}, \by_{I^{c}};\bSigma_{4})d\by_{I^{c}},
}
where the last equality is due to \citet{plackett1954reduction}.
Let $\bom_{j}^{\top}$ be the $j$th row of $\bSigma_{4}^{-1}=[\omega_{ij}]_{1\leq i,j \leq 4}$, $\bSigma_{I} = \var(\by_{I})$, and $\bSigma_{I^{c}|I}=\var(\by_{I^{c}}|\by_{I})$. By differentiating the multivariate Gaussian density, we have 
\bfln{
\left|
\int_{\Rcal_{I^{c}}}
\frac{\partial^2}{\partial y_{I_{1}} \partial y_{I_{2}} }
\phi(\by_{I}, \by_{I^{c}};\bSigma_{4})d\by_{I^{c}}
\right|
&=
\left|
\int_{\Rcal_{I^{c}}}
\left\{(\bom_{I_{1}}^{\top}\by)(\bom_{I_{2}}^{\top}\by) - \omega_{I}\right\}
\phi(\by_{I}, \by_{I^{c}};\bSigma_{4})d\by_{I^{c}}
\right| \\
&=
\left|
\phi(\by_{I};\bSigma_{I})
\int_{\Rcal_{I^{c}}}
\left\{(\bom_{I_{1}}^{\top}\by)(\bom_{I_{2}}^{\top}\by) - \omega_{I}\right\}
\phi(\by_{I^{c}}|\by_{I};\bSigma_{I^{c}|I})
d\by_{I^{c}}
\right|.
}
We also have that $\phi(\by_{I};\bSigma_{I}) \leq \phi(0,0;\bSigma_{I}) = |\bSigma_{I}|^{-1/2} \leq |\bSigma_{4}|^{-1/2}$ for any $I$, and by Lemma \ref{lem:Sigma4_inv_det}, $|\bSigma_{4}|^{-1/2}|\leq C_{2}^{1/2}$ for some constant $C_{2}>0$. Thus, 
\bfln{
\left|
\int_{\Rcal_{I^{c}}}
\frac{\partial^2}{\partial y_{I_{1}} \partial y_{I_{2}} }
\phi(\by_{I}, \by_{I^{c}};\bSigma_{4})d\by_{I^{c}}
\right|
&\leq
C_{2}^{1/2}
\left|
\int_{\Rcal_{I^{c}}}
\left\{(\bom_{I_{1}}^{\top}\by)(\bom_{I_{2}}^{\top}\by) - \omega_{I}\right\}
\phi(\by_{I^{c}}|\by_{I};\bSigma_{I^{c}|I})
d\by_{I^{c}}
\right|.
}
The absolute value of the integral of the last term is bounded as
\bfln{
\left|
\int_{\Rcal_{I^{c}}}
\left\{(\bom_{I_{1}}^{\top}\by)(\bom_{I_{2}}^{\top}\by) - \omega_{I}\right\}
\phi(\by_{I^{c}}|\by_{I};\bSigma_{I^{c}|I})
d\by_{I^{c}}
\right|
&\leq 
\int_{\Rcal_{I^{c}}}
\left|
(\bom_{I_{1}}^{\top}\by)(\bom_{I_{2}}^{\top}\by) - \omega_{I}
\right|
\phi(\by_{I^{c}}|\by_{I};\bSigma_{I^{c}|I})
d\by_{I^{c}}\\
&\leq 
\int_{\R^2}
\left|
(\bom_{I_{1}}^{\top}\by)(\bom_{I_{2}}^{\top}\by) - \omega_{I}
\right|
\phi(\by_{I^{c}}|\by_{I};\bSigma_{I^{c}|I})
d\by_{I^{c}}\\
&\leq 
\int_{\R^2}
\left|
(\bom_{I_{1}}^{\top}\by)(\bom_{I_{2}}^{\top}\by)
\right|
\phi(\by_{I^{c}}|\by_{I};\bSigma_{I^{c}|I})
d\by_{I^{c}} 
+ |\omega_{I}|,
}
where the second inequality is due to expanding the range of integration, and the third inequality is due to the triangle inequality.
By Lemma \ref{lem:Sigma4_inv_det}, we know that, whether $\bSigma_{4}=\bSigma_{4a}$ or $\bSigma_{4}=\bSigma_{4b}$, $|\omega_{jk}|\leq C_{1}$, for all $1\leq j,k \leq 4 $. Also, by the triangle inequality,
\bfln{
\left|(\bom_{I_{1}}^{\top}\by)(\bom_{I_{2}}^{\top}\by)\right| 
= 
\left|
\sum_{i'=1}^{4} \sum_{j'=1}^{4} \omega_{I_{1}i'}\omega_{I_{2}j'} y_{i'}y_{j'}
\right|
\leq
C_{1}^{2}
\sum_{i'=1}^{4} \sum_{j'=1}^{4} \left| y_{i'}y_{j'}
\right|.
}
Hence, for some constant $C>0$,
\bfln{
\left|
\frac{\partial h_{I}(r)}{ \partial \rho_{I}} 
\right| 
&\leq
C_{2}^{1/2}
\int_{\R^2}
\left|
(\bom_{I_{1}}^{\top}\by)(\bom_{I_{2}}^{\top}\by)
\right|
\phi(\by_{I^{c}}|\by_{I};\bSigma_{I^{c}|I})
d\by_{I^{c}}  + C_{2}^{1/2}C_{1} \\
&\leq
C_{2}^{1/2}C_{1}^2
\sum_{i'=1}^{4}\sum_{j'=1}^{4}
\int_{\R^2}
\left|
y_{i'}y_{j'}
\right|
\phi(\by_{I^{c}}|\by_{I};\bSigma_{I^{c}|I})
d\by_{I^{c}}  + 
C_{2}^{1/2}C_{1}\\
&\leq C,
}
where the last inequality holds by Lemma \ref{lem:bound_expection_of_product}.

Consider the case $\card(I\cap K)=1$. We assume, without loss of generality, that $I=\{i,j\}$ and $K=\{j,\ell\}$, i.e., $I\cap K = \{j\}$. Then, by \citet{plackett1954reduction},
\bfln{
\frac{\partial h_{ij}(r)}{ \partial \rho_{j\ell}} 
&=
\frac{\partial}{\partial \rho_{j\ell}}
\int_{-\infty}^{a_{\ell}}
\int_{-\infty}^{a_{k}}
\phi(a_{i}, a_{j}, x_{k}, x_{\ell};\bSigma_{4})dx_{k}dx_{\ell} 
= 
\int_{-\infty}^{a_{k}}
\frac{\partial}{\partial a_{j}}
\phi(a_{i},a_{j}, x_{k},a_{\ell};\bSigma_{4})dx_{k}.
}
For notational convenience, let $\by=(a_{i},a_{j},x_{k},a_{\ell})^{\top}=(y_{i},y_{j},y_{k},y_{\ell})^{\top}$ and write
\bfln{
\frac{\partial h_{ij}(r)}{ \partial \rho_{j\ell}} 
&= 
\int_{-\infty}^{a_{k}}
\frac{\partial}{\partial y_{j}}
\phi(\by;\bSigma_{4})dy_{k} 
=
\int_{-\infty}^{a_{k}}
(-\bom^{\top}_{j}\by)
\phi(\by;\bSigma_{4})dy_{k}.
}
Then, we have that
\bfln{
\left| 
\frac{\partial h_{ij}(r)}{ \partial \rho_{j\ell}} 
\right|
&=
\left| 
\int_{-\infty}^{a_{k}}
(-\bom^{\top}_{j}\by)
\phi(\by;\bSigma_{4})dy_{k}
\right| \\
&\leq
\int_{-\infty}^{a_{k}}
| \bom^{\top}_{j}\by |
\phi(\by;\bSigma_{4})dy_{k}  \\
&\leq
\int_{-\infty}^{\infty}
| \bom^{\top}_{j}\by |
\phi(\by;\bSigma_{4})dy_{k}
\quad \text{(by expanding the range of integration)}
\\
&\leq
\sum_{i'\neq k}|\omega_{ji'}y_{i'}|\phi(y_{i},y_{j}, y_{\ell};\bSigma_{3}) + 
\phi(y_{i},y_{j}, y_{\ell};\bSigma_{3})
\int_{-\infty}^{\infty}
| \omega_{jk} y_{k} |
\phi(y_{k}|y_{i},y_{j}, y_{\ell};\bSigma_{4})dy_{k},
}
where the last inequality holds by the triangle inequality.
By Lemma \ref{lem:Sigma4_inv_det}, $|\omega_{jk}| \leq C_{1}$ for all $1\leq j,k \leq 4$, and by Assumption \ref{assumption2}, $|y_{i}|\leq M$ for all for $1\leq i \leq 4$. This gives
\bfln{
\left| 
\frac{\partial h_{ij}(r)}{ \partial \rho_{j\ell}} 
\right|
&\leq
\sum_{i'\neq k}|\omega_{ji'}y_{i'}|\phi(y_{i},y_{j}, y_{\ell};\bSigma_{3}) + \int_{-\infty}^{\infty}
| \omega_{jk} y_{k} |
\phi(y_{k}|y_{i},y_{j}, y_{\ell};\bSigma_{4})dy_{k}
\left\{ \phi(y_{i},y_{j}, y_{\ell};\bSigma_{3})\right\} \\
&\leq
3C_{1}M |\bSigma_{4}|^{-1/2} + 
C_{1}|\bSigma_{4}|^{-1/2}
\int_{-\infty}^{\infty}
| y_{k} |
\phi(y_{k}|y_{i},y_{j}, y_{\ell};\bSigma_{4})dy_{k}.
}
Again, by Lemma \ref{lem:Sigma4_inv_det}, $|\bSigma_{4}|^{-1/2}\leq C_{2}^{1/2}$, and by Lemma \ref{lem:univariate_conditional_moments_bound}, above integral is bounded above by some positive constant. Thus, 
\bfln{
\left| 
\frac{\partial h_{ij}(r)}{ \partial \rho_{j\ell}} 
\right| \leq C
}
for some constant $C>0$.
\end{proof}

\begin{lemma}\label{lem:foldedGaussian_mean}
Let $X \sim \N(\mu,\sigma^2)$. Then $\E(|X|) \leq \sigma \sqrt{2/\pi} + |\mu|$
\end{lemma}
\begin{proof}
For $X \sim \N(\mu,\sigma^2)$, $|X|$ follows the followed Gaussian distribution with mean
\bfln{
\E(|X|) &= \sigma \sqrt{ \frac{2}{\pi} } \exp\left(-\frac{\mu^{2}}{2\sigma^2}\right) +
\mu \left\{ 1 - 2\Phi\left(-\frac{\mu}{\sigma}\right)\right\}.
}
Since $\exp(-a^2)\leq 1$ and $0 \leq a\{1-2\Phi(a)\} \leq |a|$, $a\in \R$, we have that
\bfln{
\E(|X|) &= \sigma \sqrt{ \frac{2}{\pi} } \exp\left(-\frac{\mu^{2}}{2\sigma^2}\right) +
\mu \left\{ 1 - 2\Phi\left(-\frac{\mu}{\sigma}\right)\right\} \\
& \leq
\sigma \sqrt{ \frac{2}{\pi} }  +
|\mu|.
}
\end{proof}

\bibliographystyle{chicago}
\bibliography{References}

\end{document}

%% file: cldaNotations.tex
\usepackage{amsmath}
\usepackage{amsfonts}
\usepackage{amssymb}
\usepackage{amsthm}
\usepackage{algorithm}
\usepackage{algpseudocode}
\usepackage{subcaption}
\usepackage[colorlinks,citecolor=blue,urlcolor=blue]{hyperref}
\usepackage{fullpage}
\usepackage{enumitem}

\newtheorem{theorem}{Theorem}

\newtheorem{lemma}{Lemma}

\newtheorem{definition}{Definition}
\newtheorem{assumption}{Assumption}

\newcommand{\Ccal}{\mathcal{C}}
\newcommand{\Rcal}{\mathcal{R}}
\def\ba{\boldsymbol{a}}
\def\bb{\boldsymbol{b}}

\def\be{\boldsymbol{e}}
\def\bdf{\boldsymbol{f}}
\def\bg{\boldsymbol{g}}

\def\bm{\boldsymbol{m}}

\def\bu{\boldsymbol{u}}
\def\bv{\boldsymbol{v}}

\def\bx{\boldsymbol{x}}
\def\by{\boldsymbol{y}}
\def\bz{\boldsymbol{z}}

\def\bA{\boldsymbol{A}}
\def\bB{\boldsymbol{B}}
\def\bC{\boldsymbol{C}}

\def\bI{\boldsymbol{I}}

\def\bN{\boldsymbol{N}}

\def\bS{\boldsymbol{S}}
\def\bT{\boldsymbol{T}}

\def\ones{\boldsymbol{1}}
\def\zeros{\boldsymbol{0}}

\def\bbeta{\boldsymbol{\beta}}
\def\bdelta{\boldsymbol{\delta}}
\def\bmu{\boldsymbol{\mu}}

\def\bDelta{\boldsymbol{\Delta}}

\def\bPi{\boldsymbol{\Pi}}

\def\bSigma{\boldsymbol{\Sigma}}

\def\bGamma{\boldsymbol{\Gamma}}
\def\bpi{\boldsymbol{\pi}}
\def\bom{\boldsymbol{\omega}}

\def\bDeltahat{\widehat{\boldsymbol{\Delta}} }
\def\bPihat{\widehat{\boldsymbol{\Pi}} }
\def\bThat{\widehat{\boldsymbol{T}} }
\def\bpihat{\widehat{\boldsymbol{\pi}}}
\def\pihat{\widehat{\pi}}
\def\tauhat{\widehat{\tau}}
\def\sigmahat{\widehat{\sigma}}
\def\Deltahat{\widehat{\Delta}}

\def\sign{\textup{sign}}
\def\diag{\textup{diag}}

\def\new{\text{new}}

\DeclareMathOperator*{\argmin}{argmin}

\def\R{\mathbb{R}}
\def\E{\mathbb{E}}

\newcommand{\PP}{\mathbb{P}} 

\DeclareMathOperator*{\cov}{Cov}
\DeclareMathOperator*{\var}{Var}

\DeclareMathOperator*{\card}{card}

\def\NPN{\textup{NPN}}

\def\N{\textup{N}}

\def\Unif{\textup{Unif}}

\def\brg{G}

\newcommand{\bpm}{\begin{pmatrix}}
\newcommand{\epm}{\end{pmatrix}}

\usepackage[textwidth=2cm, textsize=scriptsize]{todonotes}

\def\namedlabel#1#2{\begingroup
    #2%
    \def\@currentlabel{#2}%
    \phantomsection\label{#1}\endgroup
}

\def\bSigmah {\widehat{\boldsymbol{\Sigma}} }
\def\bbetah {\widehat{\boldsymbol{\beta}} }

\newcommand{\bfl}[1]{\begin{flalign} #1 \end{flalign}}
\newcommand{\bfln}[1]{\begin{flalign*} #1 \end{flalign*}}
